\documentclass[12pt]{article}

\usepackage[title]{appendix}
\usepackage{amsfonts}
\usepackage{amssymb}
\usepackage{amsthm}
\usepackage[fleqn]{amsmath}
\usepackage[mathcal]{euscript}
\usepackage{mathrsfs}

\usepackage[latin1,ansinew]{inputenc}
\usepackage{fullpage}
\usepackage{psfrag}
\usepackage{setspace}
\usepackage{multirow}
\usepackage{longtable}
\usepackage{graphicx}
\usepackage{booktabs}
\usepackage{rotating} 

\usepackage{float}

      \usepackage{array}
      \usepackage{calc}
      \usepackage{hhline}
      \usepackage{ifthen}

\input{epsf}
\graphicspath{{./Figures/}}

\usepackage[usenames,dvipsnames]{color}
\usepackage[usenames,dvipsnames]{colortbl}

\newtheorem{corollary}{Corollary}

\newtheorem{lemma}{Lemma}
\newtheorem{proposition}{Proposition}
\newtheorem{remark}{Remark}

\newtheorem{conjecture}{Conjecture}

\newcommand{\setdiff}{{\!\setminus\!}}

\renewcommand{\iint}{{\int\!\!\!\!\int}}

\newcommand{\Ref}[1]{\eqref{#1}}

\newcommand{\sign}{\mathrm{sign}\,}

\newcommand{\arccot}{\mathrm{arccot}\,}

\newcommand{\eps}{\epsilon}
\newcommand{\veps}{\varepsilon}

\newcommand{\cosec}{\mathrm{cosec}}

\newcommand{\vect}[1] {\boldsymbol{{ #1}} }

\newcommand{\pV}{{\vect{p}}}           
\newcommand{\qV}{{\vect{q}}}           
\newcommand{\xV}{{\vect{x}}}           

\DeclareMathAlphabet{\mathpzc}{OT1}{pzc}{m}{it}

\newcommand\pzcE{{\mathpzc{E}}}

\newcommand{\uli}[1]{\underline #1 }

\newcommand{\Nset}{\mathbb{N}}
\newcommand{\Rset}{\mathbb{R}}
\newcommand{\Sset}{\mathbb{S}}

\newcommand{\Psp}{\mathfrak{P}}

\newcommand{\cC}{{\cal C}}
\newcommand{\cM}{{\cal M}}
\newcommand{\cE}{{\cal E}}
\newcommand{\cL}{{\cal L}}

\newcommand{\centered}{{$\Sset^2$-adjusted\ }}
\newcommand{\recentered}{{$\Sset^2$-re-adjusted\ }}

\newcommand{\re}{\mathop{\mathrm{Re}}}

\newcommand{\Pochhsymb}[2]{{\left(#1\right)_{#2}}}

\DeclareMathOperator{\DirichletL}{L}
\DeclareMathOperator{\HyperF}{F}

\newcommand{\Hypergeom}[5]{{\sideset{_#1}{_#2}\HyperF\!\left(\substack{\displaystyle#3\\\displaystyle#4};#5\right)}}

\DeclareMathOperator{\zetafcn}{\zeta}

\begin{document}


\title{``Magic'' Numbers in Smale's 7th Problem }

\author{\normalsize \sc{R. Nerattini$^{a}$, J. S. Brauchart$^b$, and M. K.-H. Kiessling$^c$}\\[-0.1cm]
	\normalsize $^a$ Dipartimento di Fisica e Astronomia, Universit\`a di Firenze, and \\[-0.1cm]
	\normalsize Istituto Nazionale di Fisica Nucleare (INFN), Sezione di Firenze, \\[-0.1cm]
	\normalsize  Via G. Sansone 1, Sesto Fiorentino (FI), I-50019, Italy;\\[-0.1cm]
	\normalsize $^b$ School of Mathematics and Statistics, University of New South Wales,\\[-0.1cm]
	\normalsize Sydney, NSW, 2052, Australia;\\[-0.1cm]
	\normalsize $^c$ Department of Mathematics, Rutgers University,\\[-0.1cm]
	\normalsize 110 Frelinghuysen Rd., Piscataway, NJ 08854, USA}
\vspace{-0.3cm}
\date{$\phantom{nix}$}
\maketitle
\vspace{-1.9cm}
\begin{abstract}
\vspace{-.1cm}

\noindent
        This paper inquires into the concavity of the map $N\mapsto v_s(N)$ from the integers $N\geq 2$ into the minimal average
standardized Riesz pair-energies $v_s(N)$ of $N$-point configurations on the sphere $\Sset^2$ for various $s\in\Rset$.
	The standardized Riesz pair-energy of a pair of points on $\Sset^2$ a chordal distance $r$ apart 
is $V_s(r)= s^{-1}\left(r^{-s}-1 \right)$, $s \neq 0$, which becomes $V_0(r) = \ln\frac1r$ 
in the limit $s\to 0$. 
        Averaging it over the {\tiny{$\left(\hskip-4pt\begin{array}{c} N\\2\end{array}\hskip-4pt\right)$}}
distinct pairs in a configuration and minimizing over all possible $N$-point configurations defines $v_s(N)$.
	It is known that $N\mapsto v_s(N)$ is strictly increasing for each $s\in\Rset$, and for $s<2$ also
bounded above, thus ``overall concave.''
	It is (easily) proved that $N\mapsto v_{-2}^{}(N)$ is even locally strictly concave, and 
that so is the map $2n\mapsto v_s(2n)$ for $s<-2$.
	By analyzing computer-experimental data of putatively minimal average Riesz pair-energies $v_s^x(N)$ 
for $s\in\{-1,0,1,2,3\}$ and $N\in\{2,...,200\}$, it is found that the  map $N\mapsto {v}_{-1}^x(N)$ is
locally strictly concave, while $N\mapsto {v}_s^x(N)$ is not always locally strictly concave for $s\in\{0,1,2,3\}$: 
concavity defects occur whenever $N\in\cC^{x}_+(s)$ (an $s$-specific empirical set of integers). 
	It is found that the empirical map $s\mapsto\cC^{x}_+(s),\ s\in\{-2,-1,0,1,2,3\}$, is set-theoretically 
increasing; moreover, the percentage of odd numbers in $\cC^{x}_+(s),\ s\in\{0,1,2,3\}$ is found to increase with $s$.
	The integers in $\cC^{x}_+(0)$ are few and far between, forming a curious sequence of numbers,
reminiscent of the ``magic numbers'' in nuclear physics.
	It is conjectured that these new ``magic numbers'' are associated with optimally symmetric optimal-log-energy $N$-point 
configurations on $\Sset^2$.
        A list of interesting open problems is extracted from the empirical findings, and some rigorous first steps toward their
solutions are presented.
        It is emphasized how concavity can assist in the solution to Smale's $7$th Problem, which asks for an efficient algorithm
to find near-optimal $N$-point configurations on $\Sset^2$ and higher-dimensional spheres. 
\end{abstract}

\vfill
\hrule
\smallskip\noindent
{\footnotesize
Typeset in \LaTeX\ by the authors. To appear in Journal of Statistical Physics under the title:
\emph{Optimal $N$-Point Configurations on the Sphere: ``Magic'' Numbers and Smale's 7th Problem} 
\\
In celebration of Doron Zeilberger's $(1+\surd\epsilon)\cdot60$-th birthday.

\noindent
\copyright 2014  The authors. This preprint may be reproduced for noncommercial purposes.}
\newpage

	\section{Introduction}
        In various fields of science, ranging from biology over chemistry and physics to computer science, 
one encounters $N$-point optimization problems for which the following one is archetypical.
	Consider $N\geq 2$ distinct points on the two-sphere  $\Sset^{2}$.
	Any such $N$-point configuration  will be denoted by $\omega_N \subset\Sset^2$.
	The positions of the $N$ points are conveniently given by vectors $\qV_{k}\in\Rset^3$ of 
Euclidean length $|\qV_{k}|=1$, $k=1,\dots,N$, so that $\left|\qV_{i}-\qV_{j}\right|$ is  
the \emph{chordal distance} between the two points in the pair $(i,j)$.
        Any pair $(i,j)$ is now assigned a 
\emph{standardized Riesz pair-energy}\footnote{Traditionally the Riesz pair-energy
		is defined as $\widetilde{V}_s(r)=r^{-s}$ for $s\neq 0$, and $\widetilde{V}_0(r)=-\ln r$ for $s=0$. 
		This has the disadvantages that $\widetilde{V}_0(r)\neq\lim_{s\to 0}\widetilde{V}_s(r)$, and
		that one has to seek energy-minimizing configurations for $s\geq 0$ yet 
		energy-maximizing ones for $s<0$.}
$V_{s}(\left|\qV_{i}-\qV_{j}\right|)$, with 
\begin{eqnarray}
	V_{s}(r) &\equiv&\label{RieszVs}
	s^{-1}\left(r^{-s}-1 \right), 	\qquad s\in \Rset, \quad s \neq 0;\\
	V_{0}(r) &\equiv& \label{limiteNULL}
	-\ln{r}                
\qquad\qquad\qquad \Big(\!\!=
\lim_{s\rightarrow 0} V_{s}(r)\Big).
\end{eqnarray}
	The \emph{average standardized Riesz pair-energy of a configuration} is given  by
\begin{equation} 
	\langle V_s\rangle (\omega_N)
\equiv\label{aveRIESZpairENERGY}
	\frac{2}{N(N-1)}\;\sum\sum\limits_{\hskip-.6truecm 1\leq i < j\leq N}^{} V_s(|\qV_i-\qV_j|),
\end{equation}		
and the \emph{minimal average standardized Riesz pair-energy} by\footnote{Our $v_s^{}(N)$ equals $2\veps_s(N)$, where
        $\veps_s(N)$ denotes the so-called ``pair-specific ground state energy'' in physics
        (cf. \cite{EyinkSpohn,KieRMP,KieJSPeNULL}).
        While $\veps_s(N)$ is indeed a physically meaningful quantity, its attribute ``pair-specific''
	is a misnomer --- it should actually refer to  the statistically meaningful $v_s^{}(N)$, for the 
	number of different pairs is $N(N-1)/2$.}
\begin{equation} 
	v_s^{}(N)
\equiv\label{MINaveRIESZpairENERGY}
	\inf_{\omega_N\subset\Sset^2} \langle V_s\rangle(\omega_N).
\end{equation}		
	The problem is to determine $v_s^{}(N)$ together with the minimizing configuration(s) $\omega_N^s$
(also known as $N$-tuple of elliptic $s$-Fekete points\footnote{Originally, 
		M. Fekete (cf. \cite{Fekete}) studied points from an infinite compact set in the complex plane 
		that maximize the product of all mutual distances, which is equivalent to minimizing the average 
		standardized Riesz pair-energy for $s\to 0$.})
whenever such exist.\footnote{By the lower semi-continuity of the standardized
	Riesz pair-energy and the compactness of the sphere, there always exist $N$ labeled points 
	(not necessarily pairwise different if $s \leq -2$) whose average pair-energy equals $v_s^{}(N)$.
	A minimizing set of $N$ labeled points is not a	proper minimizing $N$-{\em point configuration} 
        unless all points are pairwise different.}
For the convenience of the non-expert reader,
in our Appendix \ref{sec:appdx.A} we present a brief survey of this intriguingly beautiful and rich, but also challengingly hard 
mathematical problem to which nobody knows the general solution.
	Only for one distinguished value of $s$ has this problem been solved for all $N$, and only for
a few $N$-values has it been conquered for all $s$.
        Also computers are soon overwhelmed when $N$ becomes large.

	Indeed, as apparently first noticed in \cite{ErberHockneyTWO}, 
computer-assisted searches (see,
e.g.,
\cite{RSZa},\hskip2pt
\cite{RSZb},\hskip2pt
\cite{AWRTSDW},\hskip2pt
\cite{PerezGetal},\hskip2pt
\cite{ErberHockneyTWO}, \hskip2pt
\cite{BCNTlett},\hskip2pt
\cite{BCNTfull},\hskip2pt
\cite{BCEG}, 

\newpage
\noindent
\cite{WalesUlker}, \cite{WalesMackayAltshuler})
for the minimizing configuration(s) suggest that the number of local minimum energy configurations (most of which are not globally
minimizing) grows exponentially\footnote{The 
     growth rate should have a significance similar to ``the \emph{complexity} of the energy landscape,'' 
     see \cite{WalesBOOK}.
          Studies of the Riesz $s$-energy landscape for $N$-point configurations on $\Sset^2$ have only begun recently,
          see \cite{CalefETal2013} and references therein.}
with $N$.
         An exponential proliferation of local minimizers, and by implication of critical points, eliminates the possibility
of a polynomial-in-$N$ algorithm which first solves the algebraic problem of finding all critical points, then evaluates 
their energies, and finally picks the lowest energy configuration(s) amongst all critical points.\footnote{For an exponential 
  time algorithm which provides rational points on the sphere whose logarithmic energy differs from the optimal value by at most 1/9,
  see Proposition 1.11 in \cite{BeltranA}.}
	Since it is so difficult to find the optimizing configurations, one may need to settle for less and employ computer-assisted 
(random) searches.\footnote{A good collection of existing search algorithms can be found at the website \cite{BCM}.} 
        Unfortunately, with increasing likelihood when $N$ becomes large a practically feasible random search 
will find only one of these exponentially many non-global minima \emph{without} guarantee of its energy being close to the optimum.
        This is not good enough. 
	What one wants is a controlled approximation.
        Smale's 7th Problem \cite{Smale} is formulated in this spirit: 

\vskip-18pt$\phantom{nix}$

\begin{quote}
\emph{Find an algorithm which, upon input $N$, in polynomial time returns a configuration $\omega_N$ on $\Sset^2$
whose average standardized Riesz pair-energy does not deviate from the optimal value obtained with $\omega_N^s$ by more than a 
certain conjectured $s$-specific function of $N$.} 
\end{quote}

\vskip-24pt$\phantom{nix}$

\begin{remark}\label{rem:Smale7}
Smale's problem was originally posed for $s=0$, viz. $V_0(r)=-\ln r$, 
and then not for the average logarithmic pair-energy but for the total logarithmic energy of the $N$-point configurations, 
i.e. for ${\mathcal{E}}_0(N)=${\scriptsize{$\left(\hskip-4pt\begin{array}{c}N\\2\end{array}\hskip-4pt\right)$}}$v_0^{}(N)$. 
	The ``$s$-specific function of $N$'' in this original formulation is the fourth term of the partially proved, 
partially conjectured large-$N$ asymptotic expansion
of the optimal logarithmic energy of $N$-point configurations on $\Sset^2$~\cite{RSZa,RSZb},

\vskip-17pt
\begin{equation}
{\mathcal{E}}_0(N)
=\label{asympCONJlogS2}
 a N^2 + b N\ln N +cN + d \ln N + \mathcal{O}(1),
\end{equation}

\vskip-5pt\noindent
with $a= \textstyle{\frac{1}{4}}\ln\textstyle{\frac{e}{4}}$ and $b = - \textstyle{\frac{1}{4}}$ rigorously known, and
with rigorous upper and lower bounds on $c$, and numerical estimates for $d$, given in \cite{RSZa} (for an update, see 
 \cite{BrHaSa2012}).\footnote{In \cite{BrHaSa2012} it is conjectured
		that $c = \ln\big(2(2/3)^{1/4}\pi^{3/4}/\Gamma(1/3)^{3/2}\big)$.
		Recently, a rigorous determination of $c$ for weighted logarithmic Fekete problems in $\Rset^2$, to 
		which the logarithmic Fekete problem on $\Sset^2$ is related by stereographic projection,
		was given in \cite{SandierSerfaty}; unfortunately, their conditions on the weights barely 
		miss the weight obtained by stereographic projection. 
                (Note added: After submission of the revised version of our paper we were informed by Laurent B\'etermin that 
                in \cite{Betermin} the order-$N$ term in  (\ref{asympCONJlogS2}) 
                is proved with the Sandier--Serfaty method; see also \cite{BeterminZhang}.)}
	The coefficient ``$d$'' in Smale's problem is unspecified and allowed to be bigger than any asymptotically 
determined ``$d.\!$''\,\footnote{Currently, only numerical evidence is available for the fourth term  in the putative 
  asymptotic expansion, and it is also conceivable that this term is actually not truly asymptotic.}

        Subsequently Smale extended his problem to $s\in[0,2)$; and he remarked 
	that analogous problems can be formulated for higher-dimensional spheres 
$\Sset^{\mathrm{d}},{\mathrm{d}}=3,4,...$.
\end{remark}

\newpage
        To the best of our knowledge, no algorithm has yet been found which delivers what Smale is asking for.\footnote{For a
          state-of-the-art survey, see \cite{BeltranB}.}
        Instead, as already mentioned, when $N$ gets (too) large, random searches\footnote{For a link to random
          polynomials, see  \cite{ABS}; in particular see their Thm.0.2.}
 and educated guesses (inevitably also of the type ``trial and error'') are employed which produce many different 
local energy minimizers for the same $N$, amongst which the one with the lowest energy is a putatively energy-optimizing 
configuration --- until a better one is found eventually, perhaps.
        In this situation it becomes important to search for necessary criteria that can test those configurations, which
have produced the lowest energy amongst all empirically found configurations with a given $N$, for their potential optimality.
        We emphasize that such types of tests can only identify non-optimal data, but not confirm optimal ones.


        One such test, based on the strict monotonic increase of $N\mapsto {v}_s^{}(N)$ (see \cite{Landkof} for a 
proof),\footnote{When the monotonicity proof was recently rediscovered \cite{KieRMP,KieJSPeNULL}, that author remarked
		(\cite{KieJSPeNULL}, p.~276) that the ``[monotonic increase of $N\mapsto v_s^{}(N)$] 
		and its proof are quite elementary and presumably known, yet after a serious search 
		in the pertinent literature I came up empty-handed,...''. (cf. also \cite{KieRMP}, p.~1188).
                M.K. likes to thank Ed Saff for subsequently pointing out to him that the monotonicity result 
                and its proof were already given in \cite{Landkof}, indeed.
		Happily, the applications of the monotonicity presented in \cite{KieRMP} and \cite{KieJSPeNULL} were novel.}
was proposed and carried out successfully in \cite{KieJSPeNULL}.
   There, about two dozen experimentally found putatively minimal energies ${v}_s^x(N)$ were identified, 
in publically available lists of empirical data, at which the empirical map $N\mapsto {v}_s^x(N)$ failed to be
monotonically increasing --- hence, these data could not possibly be true minimal energies $v_s(N)$, making it plain that
it was / is worth an effort to do better.

    Of course, all empirical maps $N\mapsto v_s^{x}(N)$ which we analyzed in this paper we also tested for whether 
they strictly increase --- all data passed this first derivative test.

\begin{remark}\label{rem:vOFsISincreasing}
    Incidentally, also the map $s\mapsto v_s^{}(N)$ is strictly monotonically increasing (see Appendix~\ref{sec:appdx.C} for
a proof), supplying another necessary criterion for optimality of putative standardized Riesz pair-energy minimizers.
   All empirical maps $s\mapsto v_s^{x}(N)$ which we analyzed in this paper we also tested for whether 
they strictly increase --- all data passed this first derivative test, too.
\end{remark}

   The main purpose of the present paper is to report the results of our quest for additional tests in form of
necessary criteria for minimality, based not on the first, but on the second discrete derivative of $N\mapsto {v}_s^{}(N)$.
	Our point of departure was the observation that strict monotonic increase of $N\mapsto v_s^{}(N)$ 
in concert with its boundedness above  for $s < 2$ (a simple variational estimate) 
implies that the \emph{overall shape} of the graph $\{(N,v_s^{}(N)):N = 2,3,....\}$ must be ``concave in the large'' 
for each $s<2$.
        This raised the question whether this graph is perhaps even \emph{locally}, at each $N>2$, strictly concave when $s<2$.
        Moreover, although $v_s^{}(N)$ is not bounded above for $s\geq 2$, the leading-order terms of the asymptotic 
large-$N$ expansion of $v_s^{}(N)$, namely $v_2(N) \propto \ln N$ \cite{KuijlaarsSaff} and $v_s^{}(N) \propto N^{(s-2)/2}$ 
for $s > 2$ \cite{HardinSaffTWO}, are strictly locally concave for $2\leq s<4$, so that it was even conceivable that so was 
$N\mapsto v_s^{}(N)$.
	So the question we asked ourselves was whether the discrete second derivative of $v_s^{}(N)$, given by
\begin{equation} \label{DDOTvs}
\ddot{v}_s^{}(N) = v_s^{}(N-1) - 2v_s^{}(N) + v_s^{}(N+1),\qquad N >2,
\end{equation}
is perhaps strictly negative for all $N>2$ when $s<2$ (or possibly even when $s\leq 4$).
        Clearly, knowledge of any $s$-values for which the map $N\mapsto {v}_s^{}(N)$ is everywhere strictly concave will 
yield a necessary criterion for minimality that can be fielded as a test for lists of empirical data of those putatively 
minimal Riesz $s$-energies.
 
\begin{remark}\label{rem:S1concavity}
        As with Smale's 7th problem, analogous concavity questions can be raised about the minimal average Riesz $s$-energy 
$v_s^{}(N)$ for the higher-dimensional spheres $\Sset^{\rm d}$, $\rm d =3,4,...$, and also for the circle $\Sset^1$
(although no optimality test for putative minimizers on $\Sset^1$ is needed --- all optimizers with $N\geq 2$ and $s\geq -2$ 
are explicitly known, see \cite{BrHaSa2009}).
        Interestingly enough, for the problem on $\Sset^1$ the $N$-dependence of $N\mapsto v^{}_{s}(N)$ is polynomial for
special $s$-values \cite{Br2011arXiv} --- in particular, the map $N\mapsto v^{}_{2}(N)$ is affine linear; furthermore, our own 
partly analytical / partly numerical studies of the explicit finite sum formulas for ${v}_s^{}(N)$ strongly support the conjecture 
that $N\mapsto v^{}_{s}(N)$ is locally strictly concave for $s\in (-2,2)$ and locally strictly convex for $s> 2$, indeed!
\end{remark}

        What we found out about the minimal average Riesz $s$-energy ${v}_s^{}(N)$ for $\Sset^2$ --- mostly empirically yet 
partly rigorously --- went beyond our expectations, including not only strict concavity of the theoretical map 
$N\mapsto {v}_s^{}(N)$, respectively the empirical map $N\mapsto {v}_s^{x}(N)$, for some $s$-values, but also hints 
at a monotonic increase with $s$ of the set of convexity points whenever concavity of ${v}_s^{x}(N)$ failed, suggesting 
novel and unexpected test criteria for optimality!
        In particular, the hypothesis that the empirically suggested monotonicities are factual has led us to discover 
in a published data list three non-optimal data points which had not previously been detected.

        More precisely, we readily affirmed the strict concavity of $N\mapsto {v}_{s}^{}(N)$ for the special value $s=-2$, 
for which the problem of the elliptic $s$-Fekete points is exactly solvable for all $N$; cf. Appendix \ref{sec:appdx.A}.
        Thus, simply by differentiating the expression \Ref{vSUBminusTWO} for $v_{-2}(N)$ twice one gets a strictly
negative expression for the second discrete derivative.
	Furthermore, twofold discrete differentiation of $v_s^{}(2n)$ when $s<-2$ (see \Ref{vsFORsBELOWminusTWO})
showed that also $2n\mapsto v_s^{}(2n)$ is strictly locally concave for $s<-2$;
of course, this does not prove that $N\mapsto v_s^{}(N)$ is strictly concave for all $N>2$ when $s<-2$, but it suggests 
that one may be able to prove it rigorously.
        This already exhausts the $s$-values for which we were able to rigorously prove some strict local 
concavity result.
        Unfortunately, the regime $s\leq -2$ is of rather academic interest --- in particular, for the
exactly solvable case $s=-2$ one does not need any necessary criteria for minimality to test putative optimizers!
        The practically interesting regime is $s>-2$.

	In the absence of any closed form expressions of $v_s^{}(N)$ for $s>-2$ we decided 
to gather some experimental input and turned to the empirical data published in \cite{ErberHockneyTWO,Sloanetal,RSZb,Ca2009}, 
and \cite{WalesUlkerDATAbase}, and to those publicly available at the website \cite{BCM} (some of which we generated ourselves).
        By carefully scrutinizing these data lists we discovered that strict local concavity of $N\mapsto v_s^{}(N)$ may in fact
hold for $s=-1$, but not for $s\in\{0,1,2,3\}$ ---  we then proved, quasi-rigorously for $s\in\{0,1,2,3\}$, rigorously for 
$s\geq 10$, that local concavity of $N\mapsto v_s^{}(N)$ \emph{fails}.

        We then inspected the sets of the $N$-values at which $\ddot{v}_s^x(N)\geq 0$ more closely.
        Empirically we found that these sets were set-theoretically increasing with $s$ --- except for
two relatively large $N$-values, namely $N=177$ and $N=197$, when $s$ was 2 or 3.
        Since the empirical data become less trustworthy the larger $N$ becomes, and also the larger $s$ becomes,
we used the ``Thomson applet'' at \cite{BCM} to see whether we could find configurations with lower energy using 
an ``educated guess'' as input  configuration for $N=177$ and $N=197$. 
        Happily we succeeded, and when we used these better energy data points for $N=177$ and $N=197$ at $s=2$, respectively 3, 
those $N$-values were no longer exceptions to the empirical overall monotonic increase of the sets of the $N$-values at which 
$\ddot{v}_s^x(N)\geq 0$.

        To summarize, we have collected empirical evidence for the following conjectures: 

\begin{conjecture}\label{conj:SconcavityFORsMINone}
The map $N\mapsto v_{-1}^{}(N)$ is locally strictly concave for all $N>2$.
\end{conjecture}

\begin{conjecture}\label{conj:noSconcavityFORsIN0123}
When $s\in\{0,1,2,3\}$ the map $N\mapsto v_s^{}(N)$ is not locally strictly concave, 
and the set of $N$-values at which strict concavity fails is expanding with $s$.
\end{conjecture}
\noindent
        Clearly, if proven true, these regularities will serve as useful test criteria for optimality of 
empirical large $N$-configurations at the above $s$-values.
        Yet our empirical findings also suggest a catalog of interesting new questions for general $s$-values.

\smallskip
\centerline{\textbf{ A Catalog of Interesting Questions (and some Partial Answers)}}
\smallskip

	To pose our questions sharply, we define several new quantities.
	First of all, for each $s$ we partition the integer subset $\{N > 2\}$ 
into three mutually disjoint subsets:
the \emph{set of strict local concavity} $\cC_{-}(s)$,  
the \emph{set of strict local convexity} $\cC_{+}(s)$,  and
the \emph{set of local linearity} $\cC_{0}(s)$, defined as
\begin{align}
\cC_{-}(s) &\equiv \{N>2| \ddot{v}_{s}(N) < 0 \}, \label{CONCsetOFs} \\
\cC_{+}(s) &\equiv \{N>2| \ddot{v}_{s}(N) > 0 \}, \label{CONVsetOFs} \\
\cC_{0\,}(s)  &\equiv \{N>2| \ddot{v}_{s}(N) =0 \}, \label{LINsetOFs}
\end{align}
respectively.
	We note that $\cC_{-}(s)\cup\cC_{0}(s)\cup\cC_{+}(s)= \{N > 2 \}$.
	We are now ready to raise our first mathematical question.

	Empirically we found that $\cC_{+}^{x}(0)\subset\cC_{+}^{x}(1)\subset \cC_{+}^{x}(2)\subset\cC_{+}^{x}(3)$. 
	Thus we ask:
\begin{enumerate}
\item[Q~1:] \hskip-2pt\emph{Is $s\mapsto\cC_{+}(s), s\!\in\!\Rset$, monotonic increasing in the sense of set-theoretic inclusion?}
\end{enumerate}

        If the answer to Q~1 is affirmative, then this set-theoretical monotonicity supplies a potentially useful test for optimality of
putative energy minimizers for any real $s$, not just for the handful of integer $s$-values which we have studied empirically.
        In particular, it would imply that $\cC_{+}(s) =\emptyset$ for all $s \leq -2$
(for $\cC_{+}(-2)=\emptyset$), and even for all $s \leq -1$ if Conjecture~\ref{conj:SconcavityFORsMINone} is true.
        More precisely, an affirmative answer to Q~1 would imply that for some $s_*^{}\geq-2$ one has $\cC_{+}(s) =\emptyset$ 
whenever $s \leq s_*^{}$.

        Incidentally, by itself this would not yet imply \emph{strict} local concavity of the map $N\mapsto v_s^{}(N)$ for $s < s_*^{}$, 
only local concavity, because we cannot conclude that $\cC_{0}(s)$ would be empty for all $s < s_*^{}$;\footnote{Of course, as a test 
  criterion for \emph{empirical} data concavity or strict concavity are equally fine.}
        yet we suspect that in fact $\cC_{0}(s)=\emptyset$ for all $s < s_*^{}$.

        To address the question of strict local concavity of $N\mapsto v_s^{}(N)$ independently of whether or not the answer to Q~1 is
affirmative, we define  $s_*^{}$ as follows,
\begin{equation}
s_*^{}  \equiv \label{sSTARdef}
\sup\{ s^\prime | \cC_{0}(s)\cup\cC_{+}(s)=\emptyset\ \forall\ s \leq s^\prime \}.
\end{equation}
	Note that the ``sup'' cannot be a ``max'' because for any fixed configuration $\omega_N$ the
map $s\mapsto \langle V_s\rangle(\omega_N)$ is a $C^\infty$ function, which implies that the map $s\mapsto v_s^{}(N)$ 
is a $C^0$ function. 
        ($C^\infty$ regularity of $s\mapsto v_s^{}(N)$ should hold almost everywhere, but for each $N> 2$ the minimizing 
configuration may change discontinuously at some $s$-value(s), at which only $C^0$ regularity can be guaranteed.)
	Alternatively, $s_*^{}$ is defined as 
\begin{equation}
s_*^{}
\equiv \label{sSTARdefALT}
\inf\{s| \text{ $\ddot{v}_s^{}(N) \geq 0$ for at least one $N > 2$} \},
\end{equation}
and ``$\inf{}=\min{}$'' if and only if $s_*^{}>-\infty$.
        Thus we ask:
\begin{enumerate}
\item[Q~2:] \emph{Is $s_*^{}>-\infty$?}
\item[Q~3:] \emph{If $s_*^{}>-\infty$, then what is the value of $s_*^{}$?}
\end{enumerate}

	Because of the strict local concavity of the maps $N\mapsto v_{-2}(N)$ and $2n\mapsto v_{s}(2n)$ for $s<-2$,
we not only suspect that $s_*^{}>-\infty$, but that $s_*^{}\geq -2$.
	In fact, the strict local concavity of the empirical map $N\mapsto v_{-1}^{x}(N)$ even suggests that $s_*^{}\geq -1$,
while the violations of strict local concavity by the empirical map $N\mapsto v_{0}^{x}(N)$ suggest that $s_*^{} <0$.
        If confirmed that $s_*^{}\in(-2,0)$, strict local concavity of $N\mapsto v_{s}^{}(N)$ can be fielded as rigorous test
criterion for optimality of putative optimizers for $s\in(-\infty,s_*^{})$, even though only $s\in(-2,s_*^{})$ would seem to be of 
practical interest.

\emph{If $s_*^{}>-\infty$, then} 
$\cC_{+}(s)=\emptyset$ for $s\leq s_*^{}$, and $\cC_{0}(s)=\emptyset$ for $s< s_*^{}$ but not for $s=s_*^{}$.  
	This leads us to now define the \emph{critical set of local linearity},
\begin{equation}
\cL_*
\equiv \label{magicSETsSTARdef}
\cC_{0}(s_*^{})\ (\mbox{if}\ s_*^{}>-\infty).
\end{equation}
        Assuming that $s_*^{}>-\infty$, we now expand our list of mathematical questions:
\begin{enumerate}
\item[Q~4:] \emph{Is the set $\cL_*$ finite or infinite?} 
\item[Q~5:] \emph{Can one explicitly compute the set $\cL_*$?}
\item[Q~6:] \emph{Which optimal configurations correspond to the set $\cL_*$?}
\end{enumerate}

	Again supposing that $s_*^{}>-\infty$, there are then also interesting questions to ask about the 
regime $s>s_*^{}$.
        In particular, since $\Nset$ is countable while $\Rset$ is not, the continuity of $s\mapsto \ddot{v}_s^{}(N)$ suggests
that $\cC_{0}(s)$ is empty almost everywhere (w.r.t. Lebesgue measure). 
        Thus it is natural to ask:
\begin{enumerate}
\item[Q~7:] \emph{Is the set of $s$-values for which $\cC_{0}(s)\neq\emptyset$ finite or infinite?}
\item[Q~8:] \emph{Can one compute the set of $s$-values at which $\cC_{0}(s)\neq\emptyset$?}
\end{enumerate}

As with $\cL_*$, one can raise similar questions for all non-empty $\cC_{0}(s)$, thus
\begin{enumerate}
\item[Q~9:] \emph{For each non-empty $\cC_{0}(s)$: is it finite or infinite?}
\item[Q10:] \emph{For each non-empty $\cC_{0}(s)$: can one compute it?}
\item[Q11:] \emph{For each non-empty $\cC_{0}(s)$: which configurations does it represent?}
\end{enumerate}

 Questions Q~2 -- Q11 do not presuppose that Q~1 is answered affirmatively.
 Yet if the answer to Q~1 is affirmative, then 
both $\lim_{s\to-\infty} \cC_{+}(s)\equiv\cC_{+}(-\infty)$ 
and  $\lim_{s\to\infty} \cC_{+}(s)\equiv\cC_{+}(\infty)$ exist.
   In particular, then $\cC_{+}(-\infty)=\emptyset=\cC_{+}(s)\;\forall s\leq s_*^{}$, 
with $s_*^{}\geq -2$, as we already know.
   As for the opposite limit, since the large-$N$ asymptotics of ${v}_s^{}(N)$ is 
locally strictly convex for $s>4$ \cite{HardinSaffTWO}, it is tempting to speculate whether there exists $s^*<\infty$ such that 
the map $N\mapsto {v}_s^{}(N)$ itself is locally strictly convex for all $s>  s^*$, which would mean that 
$\cC_{+}(\infty)=\Nset\setminus\{1,2\}= \cC_{+}(s)\;\forall\; s>  s^*$; equivalently, 
$\cC_{-}(s)=\emptyset$ for all $s\geq s^*$.
   However, we shall provide partly rigorous, partly numerical evidence for
$\cC_{+}(\infty)\neq\Nset\setminus\{1,2\}$.
   So $\cC_{+}(\infty)$, if it exists, is likely more complicated, and more interesting than the full set of
admissible integers $\Nset\setminus\{1,2\}$.
   Thus, under the hypothesis that Q~1 is answered affirmatively, we also ask:
\begin{enumerate}
\item[Q12:] \emph{Can one explicitly characterize $\cC_{+}(\infty)$?}
\end{enumerate}

    Finally, the question of $s^*$ persists:
\begin{enumerate}
\item[Q13:] \emph{Does there exist an $s^*< \infty$ such that $\cC_{+}(s)=\cC_{+}(\infty)$ for all $s>  s^*$?}
\item[Q14:] \emph{If yes, can one compute $s^*< \infty$?}
\end{enumerate}

        All these questions are presumably quite difficult to answer.
        In any event, it is reasonable to expect valuable insights even from partial answers.
        For example, to answer Q~2, and also Q~3 conditional on an affirmative answer for Q~2, one
may want to try proving a negative upper bound on $\ddot{v}_s^{}(N)$ for all $s$ below some critical value.
        Any upper bound on $\ddot{v}_s^{}(N)$ obtained in the process, even if not negative, would offer
a test criterion for optimality. 
        In this vein, in subsection \ref{sec:rigorous.upper.lower.bounds} of this paper we prove the following bounds:
\begin{proposition} \label{prop:dd.v.s.bounds.1}
	For $s < 0$ the second discrete derivative of $v_s^{}(N)$ is bounded above and below as follows,
\begin{equation*}
	\frac{2}{( N + 1 ) ( N - 2 )} \left( v_s^{}( N ) + \frac{1}{s} \right) 
\leq 
	\ddot{v}_s^{}(N) 
\leq 
	- \frac{2}{(N + 1) N} \left( v_s^{}( N ) + \frac{1}{s} \right).
\end{equation*}
\end{proposition}
       Of course, our Proposition \ref{prop:dd.v.s.bounds.1} is not strong enough to offer an answer even to Q~2. 

       Coming to Q~3, we will provide some (quasi-)rigorous upper bounds on $s_*^{}$. 
       We already mentioned that the large-$N$ asymptotics of the map $N\mapsto {v}_s^{}(N)$ is strictly convex 
for $s>4$ \cite{HardinSaffTWO}, which implies that $s_*^{}\leq 4$; yet, empirically from the analysis of putative 
minimizers we expect that $s_*^{}<0$. 
       Indeed, with the help of the partly empirically / partly rigorously known optimizers for $N\in\{5,7\}$ in the range $s\in[0,2)$
and the rigorously known optimizer for $N=6$, in subsection \ref{sec:specific.DDvs.6} we will prove the following:
\begin{proposition}\label{prop:N6BOUNDonSstar}
	Under the assumption that for $N\in\{5,7\}$ the optimizing configurations 
in the range $s\in[0,2)$ are given by the regular triangular and pentagonal bi-pyramids, respectively, one has
$
s^{}_* < 0 .
$
\end{proposition}       
       We will rigorously prove Proposition \ref{prop:N6BOUNDonSstar}; yet the bound on $s_*^{}$  in
Proposition \ref{prop:N6BOUNDonSstar} is called quasi-rigorous, for the named configurations are not rigorously known to be
optimizers for all $s\in[0,2)$, although there can be hardly any doubt that they are.

 With the help of the rigorously known optimizing configurations for $N\in\{4,6\}$ and the only partly rigorously 
known, hence putatively optimizing configurations for $N=5$ in the respective ranges 
$s\in(-2,s^\dagger]$ and $s\in[s^\dagger,\infty)$, with $s^\dagger =15.04807...$, 
in subsection \ref{sec:specific.DDvs.5} we will also supply partly rigorous and partly numerical evidence for
\begin{conjecture}\label{prop:vsN5isCONCAVE} 
	Under the assumption that the optimizing configurations for $N =5$ in the range $s\in(-2,\infty)$ are given by
the regular triangular bi-pyramid (when $s\leq s^\dagger$) and the square pyramid with adjusted height (when $s\geq s^\dagger$),
we have that
$$
\forall\; s\in(-2,\infty):\quad \ddot{v}_s(5)<0.
$$
\end{conjecture}       
\noindent
    Thus, while it is locally strictly concave for $s=-2$ and presumably for all $s\leq -2$, and possibly
for all $s\leq -1$, the map $N\mapsto {v}_s^{}(N)$ is \emph{very likely not locally strictly convex for any} $s$ --- even though 
its large-$N$ asymptotics is, when $s>4$.
    So as for Q12, very likely $\cC_{+}(\infty)$ (if it exists) is not $\Nset\setminus\{1,2\}$.

        We have mentioned large-$N$ asymptotics several times already. 
        We ourselves shall produce several well-motivated conjectures that relate the concavity of $N\mapsto {v}_s^{}(N)$ 
to its asymptotics at large $N$ for which even computer-assisted searches of the optimal $N$-point configurations are hopeless.
        We defer stating our conjectures to Section \ref{sec:asymptotics}, for we need technical preparations beyond the scope 
of this introduction.

        Our last remark makes it plain that we consider our mathematical questions to be of theoretical interest in their own right, too,
irrespective of whether some test for empirical data will ensue from their answers or not.
        In this spirit we also raise an intriguing question which cannot be so sharply formulated:

\smallskip
\centerline{\textbf{``Magic'' numbers: ``Optimally optimal'' configurations?}}\label{MAGIC}
	For $s=0$, the smallest $s$-value for which we found empirical violations of strict local 
$N$-concavity, i.e. for the logarithmic pair interaction invoked in the original formulation 
of Smale's $7$th Problem, the violations of strict local concavity were few and far between.
They occurred at the following  experimental sequence of integers:
\begin{equation}
\cC_{+}^x(0)
= \label{magicNforVnull}
\big\{\uli{6},\uli{{12}},\uli{{24}},32,\uli{{48}},\uli{{60}},67,\uli{{72}},
	80,104,\uli{{108}},122,\uli{{132}},137,... \big\}.
\end{equation}
Curiously, the majority of the numbers in the sequence \eqref{magicNforVnull} 
	are multiples of $6$ (underlined), or almost multiples (like $67$ and $137$) --- coincidence?

        We note that the logarithmic-energy minimizers for the first two
``integers of convexity,'' i.e. $N=6$ and $N=12$, are two ``optimally
symmetric'' configurations, namely Platonic polyhedra: the
octahedron ($N=6$) and icosahedron ($N=12$); also the (putative)
minimizers for $N\in\{24,48,60\}$ are highly symmetric configurations; in
particular, the one for $N=24$ is an Archimedean polyhedron
(also for $N\in\{48,60\}$ there are Archimedean polyhedra,
but these are NOT log-energy optimizers).
        To be sure, there is an integer inbetween which is not divisible
by $6$, namely $N=32$ (the highly symmetric optimizer is a Catalan
polyhedron), and also the ``odd-balls'' $N=67$ and (of all integers!)
$N=137$ show up.

        Yet, assuming that $\cC_{+}^{x}(0)=\cC_{+}^{}(0)$, it is an intriguing thought that the $N$-values in $\cC_{+}^{}(0)$ 
may correspond to $\log$-energy-optimizing configurations which are ``optimally symmetric'' in the following sense.
	Most of the $\log$-energy-optimizing configurations associated with $\cC_{+}^{x}(0)$
are separated by longer $N$-intervals in which $N\mapsto v_{0}^{x}(N)$ is strictly concave.
	This suggests that, perhaps, the configurations in an interval of concavity form
a family of more-and-more symmetric optimizers which better-and-better approximate a highly symmetric 
endpoint configuration.
	Once an endpoint configuration is reached, the addition of the next point inevitably will destroy 
a high amount of symmetry, for which an extra large amount of energy may be required.

These ``concave families'' would thus be vaguely analogous to the  ``periods'' in 
the so-called periodic table of the chemical atoms. 
	The endpoints of the periods are the chemically very inert noble gases which are
associated with highly symmetric {``electronic configurations''\footnote{Actually, what is symmetric is 
the structure of the wave function of the electrons.}} 
about the nuclei with charge number $Z\in\{2,10,18,...\}$.
	Incidentally, also the atomic nuclei seem to form something akin to ``periods,''
in the sense that the set of nucleon numbers $\{2, 8, 20, 28, 50, 82, 126,...?...\}$ 			
is associated with nuclei that have a particular high binding energy per nucleon.
	This set of nucleon numbers is known as the {\em Magic Numbers} of nuclear physics.\footnote{Since 
		there are protons and neutrons in the nucleus, some nuclei are ``doubly magic.''}
	By analogy, we call the set $\cC^{}_+(0)$ (for now: $\cC^{x}_+(0)$) the ``\emph{Magic Numbers of Smale's 7th Problem}.''
\smallskip

\newpage

\centerline{\textbf{The structure of the remaining sections}}

\vskip-25pt
$\phantom{nix}$

\begin{itemize}
\item 
	In Section~\ref{sec:data.analysis} we present the details of our analysis of the data of 
\cite{Sloanetal,RSZb,Ca2009,BCM,WalesUlkerDATAbase}, which induced us to formulate questions Q~1 -- Q14.
\item 
	In Section~\ref{sec:rigorsSTAR} we address Q~3, obtaining (quasi-)rigorous upper bounds on $s^{}_*$.
	More to the point, we give a computer-assisted proof of Proposition \ref{prop:N6BOUNDonSstar}.
	For this we explicitly compute the (putative) expressions for $s\mapsto\ddot{v}_s^{}(N)$ for $N\in\{3,4,5,6\}$ when $s$ runs 
through certain intervals in which these functions are elementary, easily discussed, and readily evaluated with {\sc{maple}}, 
{\sc{mathematica}}, or {\sc{matlab}}.
	With the help of the exactly computable $s\mapsto \ddot{v}_s^{}(3)$ we also show rigorously 
that $\cC_{+}(s)$ is nonempty whenever $s \geq 10> s^{}_*$.
\item 
	In Section~\ref{sec:rigorous.upper.bound} we prove various rigorous upper and lower bounds on  
$\ddot{v}_s^{}(N)$ which go to zero like a power of $N$ when $N\to\infty$ and $s<0$.
        In particular, we prove Proposition \ref{prop:dd.v.s.bounds.1}.
        We re-emphasize that our upper and lower bounds can serve for testing optimality of putative energy minimizers.
        While our rigorous bounds are not strong enough to prove negativity of $\ddot{v}_s^{}(N)$ uniformly in $N>2$ for any $s>-2$,
	we do rigorously prove negativity of $\ddot{v}_s^{}(N)$ for $N\in\{4,6\}$ in some regime of negative $s$-values $>-2$
by taking advantage of the explicitly known optimizers for these $N$-values (and for $N=2$).
	We also obtain a rigorous upper bound for $s\mapsto\ddot{v}_s^{}(12)$, but this bound is positive because the control of the
opitimizers for $N\in\{11,13\}$ is too weak.
	We also vindicate Conjecture \ref{prop:vsN5isCONCAVE}.
\item 
	In Section~\ref{sec:asymptotics} we present an asymptotic analysis of $N\mapsto\ddot{v}_s^{}(N)$ for the large-$N$ regime and
produce several well-motivated conjectures that relate the concavity of $N\mapsto {v}_s^{}(N)$ to the large-$N$ asymptotics.
       	The character of the asymptotics depends on whether $s$ is in the potential regime $s\in (-2,2)$, in the hypersingular 
regime $s>2$, or exactly inbetween --- at the singular $s=2$.
	A discussion of the ``degeneracy regime'' $s\leq -2$ will be left for some future work. 

\item 
	In Section~\ref{sec:summary} we summarize our findings and suggest future inquiries.
\item 
        In Appendix \ref{sec:appdx.A} we briefly survey some distinguished minimal Riesz $s$-energy problems 
	for $N$-point configurations on $\Sset^2$ and the pertinent literature.
\item 
	In Appendix~\ref{sec:appdx.B} we prove relations~\eqref{eq:s.to.infty.A} and \eqref{eq:s.to.infty.B} 
        needed in Appendix \ref{sec:appdx.A}.
\item 
	In Appendix~\ref{sec:appdx.C} we prove the strict monotonic increase of $s\mapsto {v}_s^{}(N)$.
\item 
	In Appendix~\ref{sec:appdx.D} we include a brief study of data of spherical digital nets.
\end{itemize}

\newpage

	\section{Data analysis         
          for $s\in\{-1,0,1,2,3\}$ and $N\leq 200$} \label{sec:data.analysis}
%
\vskip-7pt
        There are many studies of putatively minimal standardized Riesz $s$-energies, but only a few feature
data lists for \emph{consecutive} $N$-values which are sufficiently long for our purposes.
        In \cite{ErberHockneyTWO} the first 110 consecutive data for the Thomson problem ($s=1$) 
are reported; similarly, on the website \cite{Sloanetal} the first 130 consecutive data for $s=1$ are listed,\footnote{Save 
  the exactly computable data for $N=2$ and $N=3$.}
and on the website \cite{WalesUlker}, 391 consecutive data for $s=1$, starting with $N=10$, are reported ---
all these are never worse than those of \cite{ErberHockneyTWO}.
        In \cite{RSZb} one finds the first 200 consecutive data for $s\in\{-1,0,1\}$, and these authors remark that their data for
$s=1$ agree with those of \cite{Sloanetal} for the same $N$-values.
        M. Calef in his thesis~\cite{Ca2009} lists 180 consecutive data for minimizing configurations, starting at $N=20$, 
covering the cases $s\in\{0,1,2,3\}$; he also identified the number of ``stable'' configurations observed during many trials. 
	In $16$ cases the obtained results for $s = 0$ and $1$ differed by more than $10^{-6}$ compared to
results in \cite{RSZb} ($s = 0$ and $1$) and \cite{MoDeHo1996} ($s = 1$), some lower and some higher than those of these
other two works.
	On the interactive website \cite{BCM} putatively minimal Riesz energies are reported for 
$s\in\{0,1,\dots,12\}$, 
yet for variously many consecutive $N$-values. 
	Although there one finds consecutive data for $s=0$ and $s=1$ up to $N$ in the thousands, data
become less trustworthy with increasing $N$;\footnote{For $N\gg 200$ failures 
  of monotonicity were spotted in some data lists at \cite{BCM}; cf. \cite{KieJSPeNULL}.}
yet whenever a user finds a lower-energy than previously observed, this new record holder is substituted for the old one.

        We chose to work with about 200 numerical data each for $s\in\{-1,0,1,2,3\}$.
        In each case, we selected the lowest-energy data available from any of the mentioned lists. 
        Thus, for $s\in\{-1,0,1\}$ we worked with the data from \cite{RSZb}, except that for $s=0$ and $1$ we replaced a few data 
points by lower-energy data from \cite{Ca2009}; for $s=1$ they agree with those at \cite{WalesUlker}.
        For  $s\in\{2,3\}$ and $N\in\{20,...,200\}$, we used the data from \cite{Ca2009}, supplemented by data from \cite{BCM} 
for $s\in\{2,3\}$ and\footnote{We have completed all lists by computing $v_{s}(2)=(1/s)(2^{-s}-1)$ whenever necessary.}
$N\in\{3,...,19\}$; however, for $s=3$ consecutive data were available at \cite{BCM} only for $N\in\{3,..,12\}$,
together with data for $N\in\{16,18,19\}$. 
       So we used the applet \cite{BCM} to create our own experimental $s=3$ data for $N\in\{13,14,15,17\}$.
       Moreover, we also used the applet \cite{BCM} to create our own experimental $s=2$ and $3$ data for $N=177$ and $197$,
improving over those reported in \cite{Ca2009}; see below.

       The experimental data $\cE^{x}_s(N)$ reported in \cite{RSZb,Ca2009,BCM} (and in the 
other above-cited publications) have been computed with the conventional Riesz $s$-energy.
       We converted the data into putatively minimal \emph{average standardized} 
Riesz pair-energies, using the formula $v^{x}_{s}(N)=\frac1s\big(\frac{2}{N(N-1)}\cE^{x}_s(N)-1\big)$ for $s\in\{-1,1,2,3\}$; 
when $s=0$ only multiplication by $2/N(N-1)$ was required to obtain $v^{x}_{0}(N)$.

	Since the forward derivative $N\mapsto\dot{v}_s^+(N)\equiv {v}_s^{}(N+1)-{v}_s^{}(N)$ is strictly increasing,
for each $s\in\{-1,0,1,2,3\}$ we first checked whether $v^{x}_{s}(N+1)>v^{x}_{s}(N-n)$ for all $0\leq n\leq N-2$.
	All data that we pooled together\footnote{The data lists are given in a \emph{supplementary section} after the bibliography.  
	}
for each $s$-value passed this test.

	\subsection{Plots of ${v}_s^{x}(N)$ and $\ddot{v}_s^{x}(N)$ vs. $N$} \label{sec:figs}


        A first impression was gained by plotting $v_s^x(N)$ versus  $N$
for $N\in\{2, \dots, 200\}$ and $s\in\{-1,0,1,2,3\}$.\footnote{For $s=0$ and $1$ see, respectively, also Fig.s~1 and 2 in \cite{KieJSPeNULL}.}
        The easy-to-prove ordering (see Appendix~\ref{sec:appdx.C})
\begin{equation}\label{eq:vOfsISincreasing}
v_s^{}(N) 
>
v_t^{}(N)\quad \mbox{for}\quad s>t
\end{equation}
suggested to us to plot all five graphs $(N,v_s^x(N))$, $s\in\{-1,0,1,2,3\}$, of the \emph{empirical data} 
jointly into a single figure; see Fig.~\ref{vmin1to3vsN}, which shows 
(a) the graph $(N, v_{-1}^{x}(N))$ computed with consecutive data for $\cE^{x}_{-1}(N)$ from \cite{RSZb},
(b) the graphs $(N,v_{0}^{x}(N))$ and $(N,v_{1}^{x}(N))$, both computed with consecutive data for $\cE^{x}_{0}(N)$,
respectively $\cE^{x}_{1}(N)$, pooled from \cite{RSZb,Ca2009}, and
(c) the graphs $(N,v_{2}^{x}(N))$ and $(N,v_{3}^{x}(N))$, both computed with consecutive data for $\cE^{x}_{s}(N)$, $s\in\{2,3\}$,
pooled from \cite{Ca2009} and \cite{BCM} (all data points with the same $s$-value are joined by thin solid lines to guide the eye).
Fig.~\ref{vmin1to3vsN} reveals that $v^{x}_{-1}(N)<v^{x}_{0}(N)<v^{x}_{1}(N)<v^{x}_{2}(N)< v^{x}_{3}(N)$; thus,
all empirical data that entered  Fig.~\ref{vmin1to3vsN} pass the test implied by \eqref{eq:vOfsISincreasing} 
(cf. Remark~\ref{rem:vOFsISincreasing}).

\begin{figure}[H]
\centering
\includegraphics[scale=.9]{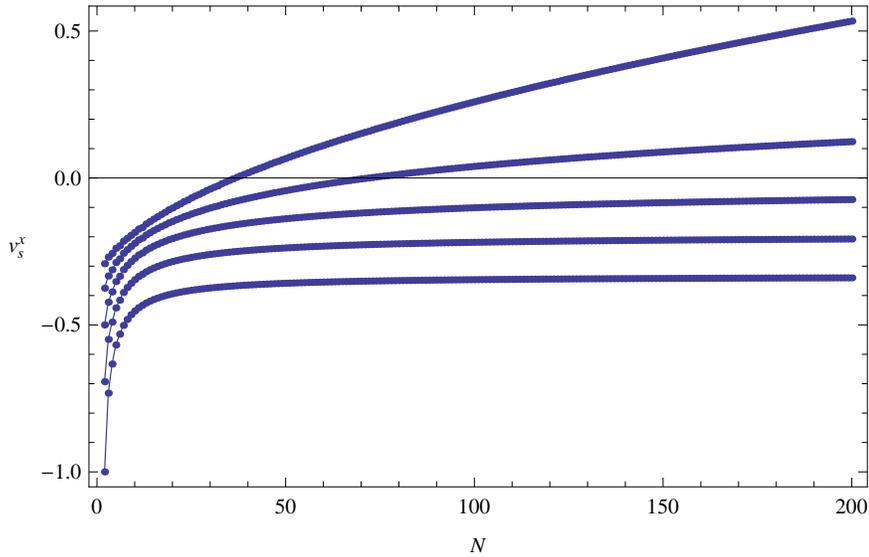}
\caption{\footnotesize{The empirical functions $v^{x}_{s}(N)$ vs. $N$ for $s\in\{-1,0,1,2,3\}$. }}\label{vmin1to3vsN}
\end{figure}
\noindent

Fig.~\ref{vmin1to3vsN} reveals also that the five empirical graphs $N\mapsto v_{s}^x(N)$,  $s\in\{-1,0,1,2,3\}$, 
are strictly increasing (as already mentioned above) and overall concave (as they should).
Furthermore, \emph{to the human eye} the empirical map $N\mapsto v_{s}^x(N)$ appears to be strictly locally concave for $s\in\{-1,0\}$, 
while convexity defects seem to occur for $s\in\{1,2,3\}$; interestingly, the online version of Fig.~\ref{vmin1to3vsN} allows one to 
zoom in, re-enforcing the impression about strict concavity vs. convexity defects.
	Of course, the human eye can only distinguish so much,\footnote{In 
    addition, the limited resolution of the plotting programs can yield deceptive plots.}
so we next plotted the second discrete derivative $\ddot{v}_s^x(N)$ for $s\in\{-1,0,1,2,3\}$, and $N\in\{3,...,199\}$.

  Shown in Fig.~\ref{ddotvmin1vsN} is the graph $(N,\ddot{v}^{x}_{-1}(N))$ for $2<N<200$, together with a zoom-in of the 
domain $100\leq N < 200$.
  (Since the amplitude of $\ddot{v}_{-1}^x(N)$ decays from about $10^{-1}$ to less than $10^{-6}$ when $N$ ranges
from 3 to 199, no single plot can show the global structure of the graph as well as its fine structure.)
  The data $\ddot{v}^{x}_{-1}(N)$ in Fig.~\ref{ddotvmin1vsN} remain below the $N$ axis, 
confirming the impression that $N\mapsto {v}^{x}_{-1}(N)$ is strictly locally concave.
\begin{figure}[H]
\centering
\includegraphics[scale=.9]{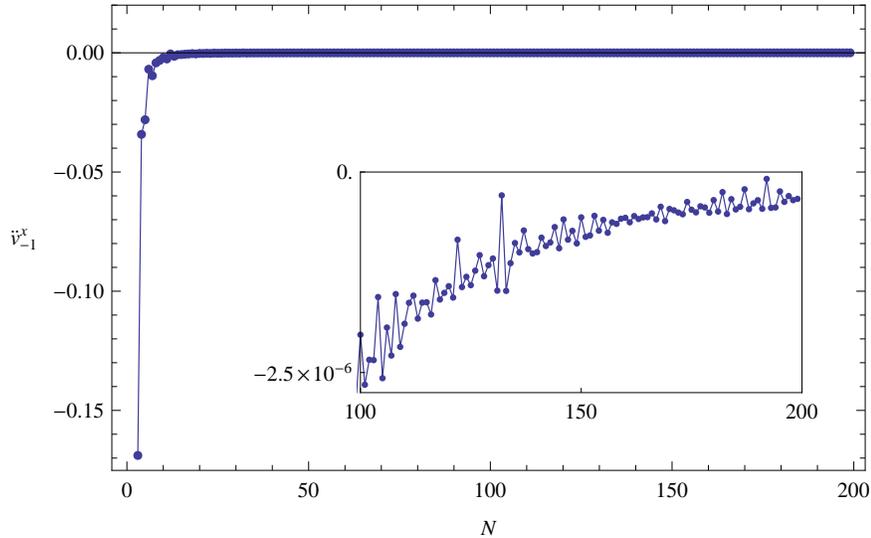}
\caption{\footnotesize{$\ddot{v}^{x}_{-1}(N)$ as a function of $N$. The solid  line is drawn to guide the eye.}}\label{ddotvmin1vsN}
\end{figure}

\vspace{-5pt}
        While Fig.s~1 and 2 do show strict local concavity of $N\mapsto v_{-1}^{x}(N)$, 
we also confirmed the optical impression of the strict local concavity of the empirical function $N\mapsto v^{x}_{-1}(N)$ with a 
refined data analysis; see below.
	The strict local concavity of $N\mapsto v_{-1}^x(N)$ for $2\leq N\leq 200$ can be taken as mild
empirical support for our \emph{Conjecture} \ref{conj:SconcavityFORsMINone} that $N\mapsto v_{-1}(N)$ is strictly locally concave.

	Also the plot of $N\mapsto v_{0}^x(N)$ in  Fig.~\ref{vmin1to3vsN} 
does look pretty much strictly concave everywhere; however, the graph of $N\mapsto \ddot{v}_{0}^x(N)$, shown in Fig.~\ref{ddotv0vsN} 
for $N\in\{3,...,199\}$, reveals that the graph of $N\mapsto v_{0}^x(N)$ is not locally strictly concave!
	Despite the strictly concave \emph{optical appearance} of $N\mapsto v_{0}^x(N)$, tiny violations
of concavity occur every now and then, in fact already when $N$ is less than a dozen.
        The positivity of some small-$N$ data in the graph of $N\mapsto v_{0}^x(N)$ is clearly visible, 
but not for $N$ larger than a dozen (say), because the amplitude of $\ddot{v}_s^x(N)$ now decays from a value of order $10^{-1}$ 
to less than $10^{-5}$ when $N$ ranges from 3 to 199. 
        To aid the visualization in the global graph, we plotted positive data points with black filled diamonds, negative ones with 
blue filled circles (colors online); again, we also inserted a zoom-in for the domain $100\leq N< 200$ --- note the different vertical
scale.

\begin{figure}[H]
\centering
\includegraphics[scale=1]{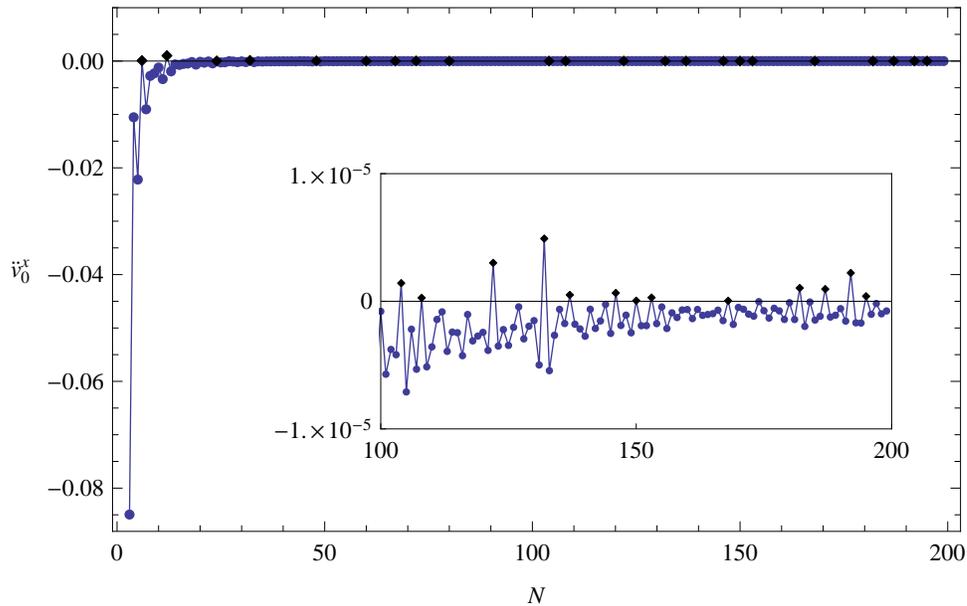}
\caption{\footnotesize{$\ddot{v}^{x}_{0}(N)$ as a function of $N$. The solid  line is drawn to guide the eye.}}\label{ddotv0vsN}
\end{figure}

\vspace{-5pt}
        We already mentioned that zooming into Fig.~\ref{vmin1to3vsN} reveals that neither of the graphs $N\mapsto v_s^{x}(N)$,
$s\in\{1,2,3\}$, appears locally strictly concave, not even to the human eye.
        Indeed, the graphs $(N,\ddot{v}^{x}_{s}(N))$ for $N\in\{3,...,199\}$ and $s\in\{1,2,3\}$ 
clearly cross the $\ddot{v}^{x}_{s}(N)=0$ axis many times; see
Fig.s~\ref{ddotv1vsN},~\ref{ddotv2vsN}, and~\ref{ddotv3vsN}, which are designed like Fig.~\ref{ddotv0vsN}.
\begin{figure}[H]
\centering
\includegraphics[scale=1]{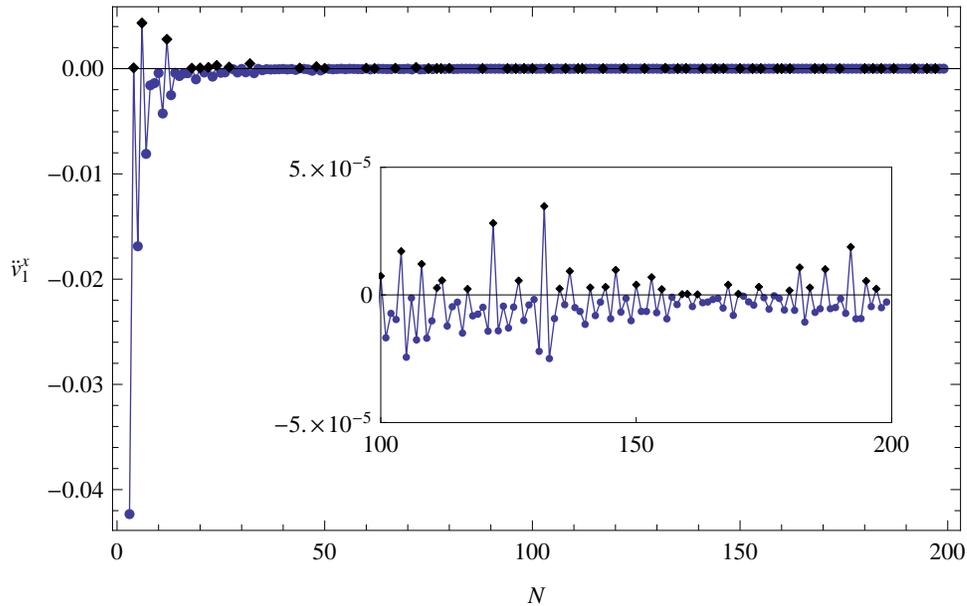}
\caption{\footnotesize{$\ddot{v}^{x}_{1}(N)$ as a function of $N$. The solid  line is drawn to guide the eye.}}\label{ddotv1vsN}
\end{figure}

\newpage

\begin{figure}[H]
\centering
\includegraphics[scale=1]{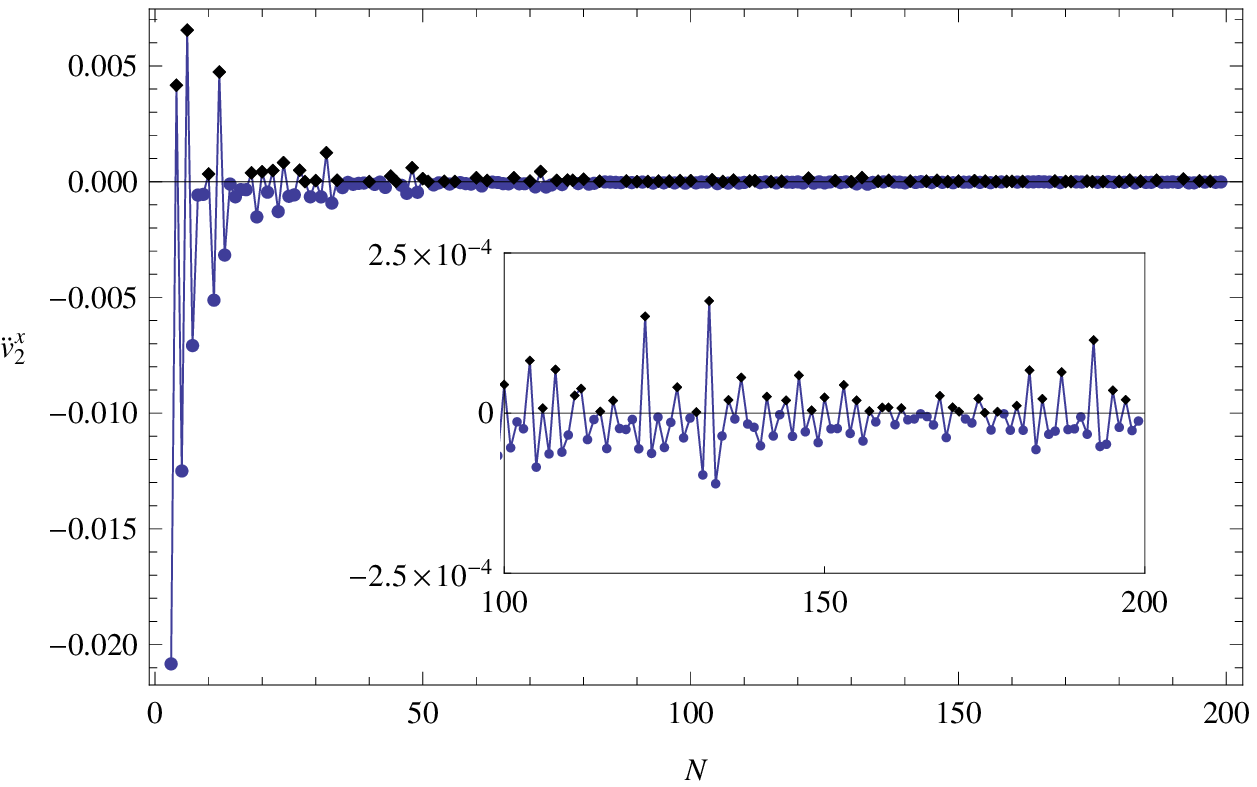}
\caption{\footnotesize{  $\ddot{v}^{x}_{2}(N)$ as a function of $N$. The solid  line is drawn to guide the eye.}}\label{ddotv2vsN}
\end{figure}
\begin{figure}[H]
\centering
\includegraphics[scale=1]{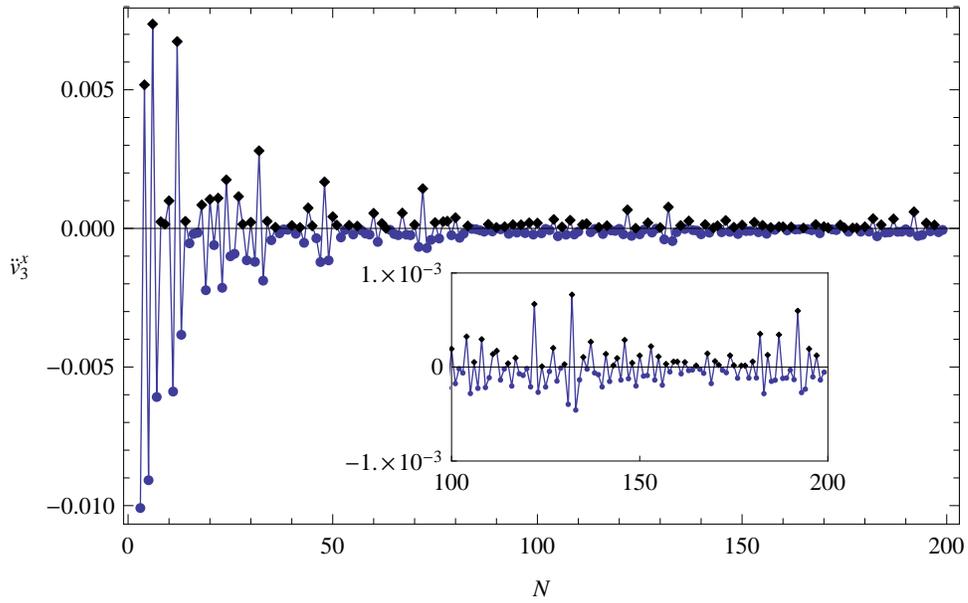}
\caption{\footnotesize{$\ddot{v}^{x}_{3}(N)$ as a function of $N$. The solid  line is drawn to guide the eye.}}\label{ddotv3vsN}
\end{figure}

	Since the minimizing configurations for $N$ less than a dozen or so have been determined numerically with a high 
degree of confidence, the empirical violations of strict local concavity for $s\in\{0,1,2,3\}$ at the smaller $N$ 
values do not seem to be due to incorrect computations of the minimal average pair-energies. 

\subsection{The sets $\cC_{+}^{x}(s)$}\label{allOURconvCs}

	We next inspected the numerical data of the second discrete derivative $\ddot{v}_{s}^{x}(N)\equiv \dot{v}_s^+(N)-\dot{v}_s^-(N)$ 
(where $\dot{v}_s^-(N)\equiv {v}_s^{}(N)-{v}_s^{}(N-1)$ is the first backward derivative) for  $N\in\{3,...,199\}$ and $s\in\{-1,0,1,2,3\}$.
	For each of these $s$-values we have collected the $N$-values at which $\ddot{v}_{s}^{x}(N)>0$ \footnote{We note that 
	it is futile to look for $N$-values for which $\ddot{v}_{s}^{x}(N)=0$ in the empirical data.} 
into the empirical set $\cC_{+}^{x}(s)$.

	Empirically, we thereby found the following \emph{experimental sets of convexity}:
\begin{eqnarray}
	&\cC_{+}^x(-1)  
= \label{C+xminONE}
	\emptyset,\phantom{nixxxxxxxxxxxxxxxxxxxxxxxxxxxxxxxxxxxxxxxxxx}\\
	&\cC_{+}^x(\;0\;) 
=  \label{C+xNULL}
	\big\{6,12,24,32,48,60,67,72,80,104,108,122,132,\phantom{nixxxxxxx} \hfill\\ 
&\nonumber\phantom{6}
	137,146,150,153,168,182,187,192,195,...?... \big\},
\\
&         \cC_{+}^x(\;1\;)
 = \label{C+xONE}
	\big\{4,6,12,18,20,22,24,27,32,44,48,50,60,62,67,72,75,\phantom{nix} \hfill\\ 
&\nonumber \phantom{equaaas}
        77,78,80,88,94,96,98,100,104,108,111,112,117,122, \\
&\nonumber \phantom{equalLs}
        127,132,135,137,141,144,146,150,153,155,159,160,\\
& \nonumber\phantom{equaLLLs}
        162,168,170,174,180,182,184, 187,192,195,197,...?... \big\},
\\
\!\!\!\!\!
&\cC_{+}^x(\;2\;)
 = \label{C+xTWO}
	\big\{4,6,10,12,18,20,22,24,27,28,30,32,34,40,44,45,48,\phantom{nix} \hfill\\ 
&\nonumber\phantom{equl}
		50,51,54,56,60,62,67,70,72,75,77,78,80,83,88,\\
&\nonumber\phantom{equals}
		90,92,94,96,98,100,104,106,108,111,112,115,117,\\
&\nonumber\phantom{equalls}
		122,127,130,132,135,137,141,144,146,148,150,153, \\
&\nonumber\phantom{equalls}
		155,157,159,160,162,168,170,171,174,175,177,180,\\
&\nonumber\phantom{equals}
		182,184,187,192,195,{\bf 197},...?... \big\},\phantom{nixxxxxxxxxxxx}
\\
&	\;\cC_{+}^x(\;3\;)
 = \label{C+xTHREE}
	\big\{4,6,8,9,10,12,14,18,20,22,24,27,28,30,32,34,36,40,42,\\
&\nonumber\phantom{equaLLL1}
		44,45,48,50,51,54,56,60,62,63,67,70,72,75,77,78,80,\\
&\nonumber\phantom{equaLLl}
		83,88,90,92,94,96,98,100,104,106,108,111,112,115,\\
&\nonumber\phantom{equaLs}
		117,122,124,127,130,132,135,137,141,143,144,146,\\
&\nonumber\phantom{equalls}
		148,150,153,155,157,159,160,162,165,168,170,171,\\
&\nonumber\phantom{equaLLLs}
		174,175,{\bf 177},178,180,182,184,187,192,195,{\bf 197},...?... \big\}.
\end{eqnarray}

\noindent
        Here, the numbers $N=177$, marked {\bf bold} in $\cC_{+}^{x}(3)$, and $N=197$, marked {\bf bold}
in $\cC_{+}^{x}(s)$ for $s\in\{2,3\}$, were absent when we used the data of \cite{Ca2009}, but appeared after we replaced the 
corresponding energy data from \cite{Ca2009} with lower energy data discovered ourselves by using the applet \cite{BCM}.
        The reason for why we became suspicious of those data in \cite{Ca2009} is explained next.

	Namely, inspection of the empirical sets $s\mapsto\cC_{+}^{x}(s)$ reveals a quite interesting property: 
the map $s\mapsto\cC_{+}^{x}(s)$ increases, set-theoretically!
        More precisely, before replacing the data of \cite{Ca2009} for $N=177$ when $s=3$, and for $N=197$ when 
$s\in\{2,3\}$, with lower-energy data obtained with the help of the applet \cite{BCM}, we had noticed that
the map $s\mapsto\cC_{+}^{x}(s)$ seemed to be ``mostly increasing,''  set-theoretically, with the exception of
exactly those three data points. 
        Thus we found that $\cC_{+}^{x}(0)\backslash\cC_{+}^{x}(1)=\emptyset$, whereas
\begin{equation}
\begin{split}
\cC_{+}^{x}(1)\backslash\cC_{+}^{x}(0)
=\label{C+xONEminC+xZERO} 
 \{&4,18,20,22,27,44,50,62,75,77,78,88,94,96,98,\\
 &100,111,112,117,127,135,141,144,155,159,\\
 &160,162,170,174,180,184,197\};
\end{split}
\end{equation}
and without the {\bf bold} data we found $\cC_{+}^{x}(1)\backslash\cC_{+}^{x}(2)=\{197\}$, whereas
\begin{equation}
\begin{split}
\cC_{+}^{x}(2)\backslash \cC_{+}^{x}(1)
=\label{C+xTWOminC+xONE} 
\{&10,28,30,34,40,45,51,54,56,70,83,90,92,\\
&106,115,130,148,157,171,175,177\};
\end{split}
\end{equation}
and finally, without the {\bf bold} data we found $\cC_{+}^{x}(2)\backslash\cC_{+}^{x}(3)=\{177\}$, whereas
\begin{equation}
\begin{split}
\cC_{+}^{x}(3)\backslash \cC_{+}^{x}(2)
= \label{C+xTHREEminC+xTWO} 
\{&8, 9, 14, 36, 42, 63, 124, 143, 165, 178\}.
\end{split}
\end{equation} 
        So, restricted to $N\leq 176$, and to $s\in\{-2,-1,0,1,2,3\}$, the map $s\mapsto\cC_{+}^{x}(s)$ increased set-theoretically.
        Also $\cC_{+}^{x}(s)\cap\{178,...,196\}$  increased set-theoretically.       
	Curiously, $N=177\in \cC_{+}^{x}(2)$ but $177\not\in \cC_{+}^{x}(3)$ without the {\bf bold} data;
even more curious was the fact that $N=197\in \cC_{+}^{x}(1)$ but $197\not\in \cC_{+}^{x}(2)$ and $197\not\in \cC_{+}^{x}(3)$
without the {\bf bold} data.
        Since these non-monotonic outliers occurred for quite large $N$-values, it was tempting to 
conjecture that $s\mapsto\cC_{+}(s)$ may actually increase set-theoretically. 
        To test this hypothesis we tried to beat Calef's $N=177$ data for $s=3$ and $N=197$ data for $s=2$ and $s=3$; cf. \cite{Ca2009}. 
         We achieved this goal by loading the (presumably optimal) $s=1$ configurations for $N=177$ and
$N=197$, respectively, and evaluated their Riesz $s$-energies for $s=2$ and $s=3$. 
        These new data are now available at \cite{BCM}.
        With these new {\bf bold} data in place for those of \cite{Ca2009}, the set-theoretical differences became
$\cC_{+}^{x}(1)\backslash\cC_{+}^{x}(2)=\emptyset$, as well as
$\cC_{+}^{x}(2)\backslash\cC_{+}^{x}(3)=\emptyset$, while the set-theoretical differences
$\cC_{+}^{x}(2)\backslash\cC_{+}^{x}(1)$ and 
$\cC_{+}^{x}(3)\backslash\cC_{+}^{x}(2)$ remained as shown in \eqref{C+xTWOminC+xONE} and \eqref{C+xTHREEminC+xTWO}.

        We summarize: with the best available data in place, the empirical map $s\mapsto\cC_{+}^{x}(s),s\in\{-1,0,1,2,3\}$,
increases strictly, set-theoretically.
 \emph{This is the basis of our Conjecture \ref{conj:noSconcavityFORsIN0123} that the actual map $s\mapsto\cC_{+}(s), s\in\Rset$, 
increases set-theoretically.}

	The increase of the  map $s\mapsto\cC_{+}^{x}(s)$ becomes easier discernable 
by collecting the sets $\cC^{x}_{+}(s)\subset\{3,...,199\}$ for $s\in\{-1,0,1,2,3\}$ into the following table. 
\newpage

\begin{center}
\begin{longtable}{|c|c|c|c|c|} 
\caption{\footnotesize{The sets $\cC^{x}_{+}(s)$ for $s\in\{-1,0,1,2,3\}$ and $2<N<200$. }}
\label{N_jumps}\\
\hline	
\multicolumn{1}{|c|}{$\cC_{+}^{x}(-1)$} & 
\multicolumn{1}{|c|}{$\cC_{+}^{x}(\;0\;)$} & 
\multicolumn{1}{|c|}{$\cC_{+}^{x}(\;1\;)$} &
\multicolumn{1}{|c|}{$\cC_{+}^{x}(\;2\;)$} &
\multicolumn{1}{|c|}{$\cC_{+}^{x}(\;3\;)$}  \\ 
\hline 
\endfirsthead
\multicolumn{5}{c}%
{{\bfseries \tablename\ \thetable{} \footnotesize{-- continued from previous page}}} \\
\hline 
\multicolumn{1}{|c|}{$\cC_{+}^{x}(-1)$} & 
\multicolumn{1}{|c|}{$\cC_{+}^{x}(\;0\;)$} & 
\multicolumn{1}{|c|}{$\cC_{+}^{x}(\;1\;)$} &
\multicolumn{1}{|c|}{$\cC_{+}^{x}(\;2\;)$} &
\multicolumn{1}{|c|}{$\cC_{+}^{x}(\;3\;)$}  \\ 
\hline 
\endhead
\hline 
\multicolumn{5}{|r|}{{\footnotesize{Continued on next page}}} \\ 
\hline
\endfoot
\hline
\hline
\endlastfoot
{\cellcolor[gray]{.85}}&		   {\cellcolor[gray]{.85}}&        	{\cellcolor[gray]{.85}}  &	{\cellcolor[gray]{.85}}  &      {\cellcolor[gray]{.85}}\\  \hline
{\cellcolor[gray]{.85}}&		   {\cellcolor[gray]{.85}}&        	\footnotesize{4}  &	\footnotesize{4}  &      \footnotesize{4}\\  \hline
{\cellcolor[gray]{.85}}&		   {\cellcolor[gray]{.85}}&        	{\cellcolor[gray]{.85}}  &	{\cellcolor[gray]{.85}}  &      {\cellcolor[gray]{.85}}\\  \hline
{\cellcolor[gray]{.85}}&           \footnotesize{6}&	\footnotesize{6}  &	\footnotesize{6}  &      \footnotesize{6}\\  \hline
{\cellcolor[gray]{.85}}&		   {\cellcolor[gray]{.85}}&        	{\cellcolor[gray]{.85}}  &	{\cellcolor[gray]{.85}}  &      {\cellcolor[gray]{.85}}\\  \hline
{\cellcolor[gray]{.85}}&        		   {\cellcolor[gray]{.85}}&        			  {\cellcolor[gray]{.85}}&        			  {\cellcolor[gray]{.85}}&              \footnotesize{8}\\  \hline
{\cellcolor[gray]{.85}}&        		   {\cellcolor[gray]{.85}}&        			  {\cellcolor[gray]{.85}}&        			  {\cellcolor[gray]{.85}}&              \footnotesize{9}\\  \hline
{\cellcolor[gray]{.85}}&        		   {\cellcolor[gray]{.85}}&        			  {\cellcolor[gray]{.85}}&        	\footnotesize{10} &      \footnotesize{10}\\
\hline
\end{longtable}
\end{center}

\vskip-0.5truecm
	It is also of interest to supplement the qualitative statement, that  $s\mapsto\cC_{+}^{x}(s)$ increases  
monotonically with $s\in\{-1,0,1,2,3\}$, by quantitative information about bulk properties of these increasing sets 
$\cC_{+}^{x}(s)$ as $s$ varies through $\{-1,0,1,2,3\}$.
	Thus, the percentage of integers in $\{3,...,199\}$ belonging to $\cC_{+}^{x}(s)$ increases with $s$ as follows:

\centerline{$\cC_{+}^{x}(-1)$ contains $\ 0\%$ of the integers from  $\{3,...,199\}$;} 

\centerline{$\cC_{+}^{x}(\;0\;)$ contains $11\%$ of the integers from  $\{3,...,199\}$;}

\centerline{$\cC_{+}^{x}(\;1\;)$ contains $27\%$ of the integers from  $\{3,...,199\}$;}

\centerline{$\cC_{+}^{x}(\;2\;)$ contains $38\%$ of the integers from  $\{3,...,199\}$;}

\centerline{$\cC_{+}^{x}(\;3\;)$ contains $43\%$ of the integers from  $\{3,...,199\}$.}
\medskip

\noindent
        These percentages suggest that  $\cC_{+}^{x}(s)$ contains more and more integers when $s$ increases
continuously, but they do not reveal that $\cC_{+}^{x}(s)$ increases monotonically in~$s$.
        
	Similarly, it is readily seen that:

\centerline{\quad 5 out of 22 integers in $\cC_{+}^{x}(\;0\;)$ are odd, or $23\%$;}

\centerline{ 16 out of 54 integers in $\cC_{+}^{x}(\;1\;)$ are odd, or $30\%$;}

\centerline{ 24 out of 75 integers in $\cC_{+}^{x}(\;2\;)$ are odd, or $32\%$;}

\centerline{ 28 out of 85 integers in $\cC_{+}^{x}(\;3\;)$ are odd, or $33\%$.}

\noindent
	This second table of percentages reveals yet another monotonicity: 
\emph{the percentage of odd integers in $\cC_{+}^{x}(s)$ increases monotonically with $s$}.
        This raises the question whether such a  monotonicity is perhaps a property of the theoretical 
map $s\mapsto\cC_{+}(s)$ for $s\in\Rset$.
        If the percentage of odd numbers in $\cC_{+}^{}(s)$ is indeed increasing, then --- since it cannot increase beyond
100\% --- it will converge to some limit when 
$s$ increases (provided the percentage remains meaningful, i.e. as long as $\cC_{+}^{}(s)$ is not empty),
and also when $s\downarrow s_*^{}<0$, if such a $s_*^{}$ exists.
        In particular, it makes one wonder whether only even integers will remain in  $\cC_{0}^{}(s_*^{})$.
	\subsection{Further visualization of the increase of the sets $\cC_{+}^{x}(s)$} \label{sec:figsB}

        To further aid the visualization of the set-theoretical differences for $s_1^{}<s_2^{}$, we defined 
``signed indicator functions'' $\Delta_{s_1}^{s_2}(N)\equiv {\bf I}_{\cC_{+}^{x}(s_2^{})}(N) - {\bf I}_{\cC_{+}^{x}(s_1^{})}(N)$,
where ${\bf I}_{\cC_{+}^{x}(s)}(N)=1$ if $N\in {\cC_{+}^{x}(s)}$, and ${\bf I}_{\cC_{+}^{x}(s)}(N)=0$ if $N\not\in {\cC_{+}^{x}(s)}$.
        Thus, 
\begin{equation}
\Delta_{s_1}^{s_2}(N)
=
\begin{cases} 
 \ \ 1&\mbox{\footnotesize{if}} \;N\not\in\cC_{+}(s_1^{})\; \mbox{\footnotesize{and}}\; N\in\cC_{+}(s_2^{}),\\
 \ \ 0&\mbox{\footnotesize{if}}\;N\in\cC_{+}(s_1^{})\cap\cC_{+}(s_2^{})\;
      \mbox{\footnotesize{or}}\; 
				N\not\in\cC_{+}(s_1^{})\cup\cC_{+}(s_2^{}),\\
 - 1&\mbox{\footnotesize{if}} \;N\in\cC_{+}(s_1^{})\; \mbox{\footnotesize{and}}\; N\not\in\cC_{+}(s_2^{}).
\end{cases}
\end{equation}
        Whenever a convexity point $N'$ is lost by passing from $s_1^{}$ to $s_2^{}$, the signed indicator function
$N\mapsto\Delta_{s_1}^{s_2}(N)$ will take the value $-1$ at $N=N'$; otherwise $N\mapsto\Delta_{s_1}^{s_2}(N)$ is non-negative. 
        This affords a convenient way of checking for set-theoretical monotonicity, compared to the painstaking sifting through
numerical tables.

	Fig.s~\ref{Deltas1s0}, \ref{Deltas2s1}, and \ref{Deltas3s2} below show 
$\Delta_{0}^{1}(N)$, $\Delta_{1}^{2}(N)$, and $\Delta_{2}^{3}(N)$ for $N\in\{3,...,199\}$.
	Since $\cC_{+}^{x}(0) \subset \cC_{+}^{x}(1)\subset \cC_{+}^{x}(2)\subset \cC_{+}^{x}(3)$,
these diagnostic functions only take the values $0$ and $1$. 

\vskip-.2truecm
\begin{figure}[H]
\centering
\includegraphics[width=11truecm,height=4truecm]{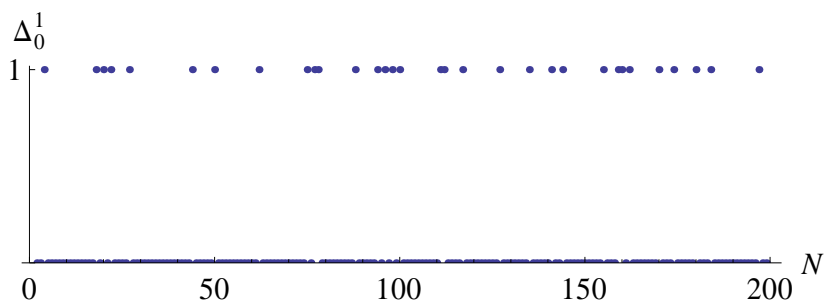}

\vskip-.3truecm
\caption{\footnotesize{$N\mapsto\Delta_{0}^{1}(N)$. See text.}}\label{Deltas1s0}
\end{figure}
\vskip-.7truecm
\begin{figure}[H]
\centering
\includegraphics[width=11truecm,height=4truecm]{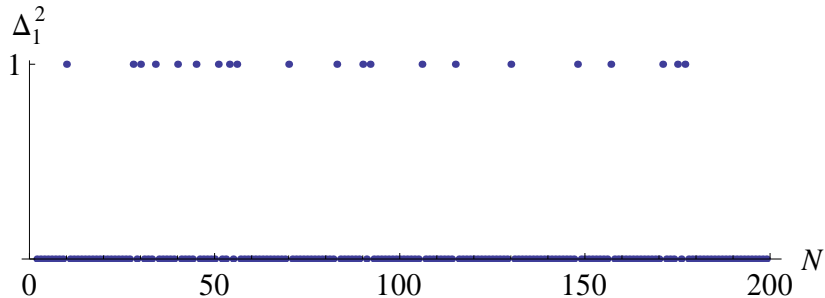}

\vskip-.3truecm
\caption{\footnotesize{$N\mapsto\Delta_{1}^{2}(N)$. See text.}}\label{Deltas2s1}
\end{figure}
\vskip-.7truecm
\begin{figure}[H]
\centering
\includegraphics[width=11truecm,height=4truecm]{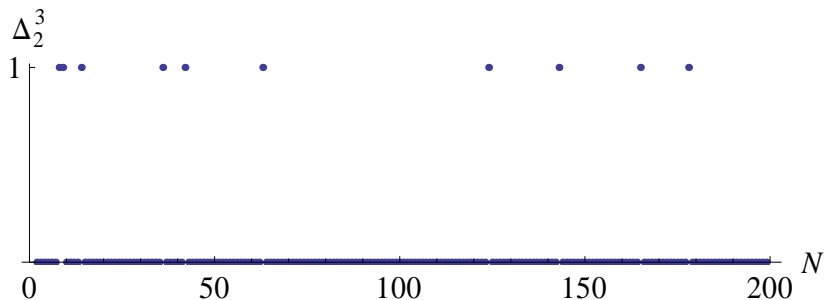}

\vskip-.3truecm
\caption{\footnotesize{$N\mapsto\Delta_{2}^{3}(N)$. See text.}}\label{Deltas3s2}
\end{figure}

\vskip-.3truecm
      This concludes our data analysis section. 
      We next present our theoretical results.

\newpage

	\section{(Quasi-)rigorous upper bounds on $s_*^{}$} \label{sec:rigorsSTAR}

	While it may not be so easy to obtain a rigorous lower bound on $s_*^{}$, it is 
relatively easy to find a rigorous upper bound on $s_*^{}$ by studying the $s$-dependence of $\ddot{v}_s^{}(3)$;
and with the help of computer
evaluations of $\ddot{v}_s^{}(N)$ for $N\in\{3,4,5,6\}$ 
we easily obtained better, though only ``quasi-rigorous,'' upper bounds on $s_*^{}$.

	To study the $s$-dependence of $\ddot{v}_s^{}(N)$ for $N\in\{3,4,5,6\}$ when $s>-2$,
we need to know the minimizing $N$-point configurations for  $N\in\{2,3,4,5,6,7\}$ if $-2\leq s< \infty$. 
	We begin by summarizing what we know about these.

	\subsection{${\omega}_N^s$ for $N\in\{2,3,4,5,6,7\}$ and $s\in(-2,\infty)$} \label{sec:configs.s.N}

	Below we list the four rigorously known optimizers $\omega_N^s$, $N\in\{2,3,4,6\}$,
together with the partly rigorously known but mostly computer-generated putative optimizers\footnote{Since the $s$-range has 
  not been --- and cannot be --- covered exhaustively with a computer, our list of 5-point and 7-point optimizers should 
  be seen as preliminary.}  
$\omega_5^s$ and $\omega_7^s$ for $s\in(-2,\infty)$.
	Note that for $N\in\{2,3,4,6\}$ each optimizer $\omega_N^s$ is independent of $s\in(-2,\infty)$,\footnote{Recall 
		that $\omega_N^s$ depends on $s\in(-\infty,-2]$ when $N$ is odd; recall also
		that at $s=-2$ the minimizer is not unique even after factoring out $SO(3)$, except when $N=2$.}
while the (putative) optimizers $\omega_5^s$ and $\omega_7^s$ display a non-trivial dependence on $s\in(-2,\infty)$:
\begin{eqnarray}
\omega_2^s&:&\label{omTWO}\quad
 	\mbox{two antipodal points}		\hskip55pt s\in(-2,\infty) ;
\\
\omega_3^s&:&\label{omTHR}\quad
	\mbox{equatorial equilateral triangle}  \hskip10pt s\in(-2,\infty) ;
\\
\omega_4^s&:&\label{omFOU}\quad
	\mbox{regular tetrahedron} 		\hskip62pt s\in(-2,\infty) ;
\\
\omega_5^s&:&\label{omFIV}
\begin{cases} 
	\mbox{triangular bi-pyramid}           &\quad\hskip20pt s\in(-2,15.04807...],
\\
	\mbox{square pyramid $(f=1)$}          &\quad\hskip20pt s\in[15.04807...,\infty);
\end{cases}
\\
\omega_6^s&:&\label{omSIX}\quad
	\mbox{regular octahedron}   		\hskip65pt s\in(-2,\infty) ;
\\
\omega_7^s&:&\label{omSEV}
\begin{cases} 
        C_2(1^12^3)(f=5) 			& \quad\hskip24pt s\in(-2,0], 
\\
	\mbox{pentagonal bi-pyramid} 		& \quad\hskip24pt s\in[0,2],  
\\
	C_2(1^12^3)(f=5) 			& \quad\hskip24pt s\in[2,5],
\\
	C_{2v}(1^14^12^1)(f=3) 			& \quad\hskip24pt s\in[5,5.5979...],
\\
	C_{3v}(1^13^2)(f=2)			& \quad\hskip24pt s\in[5.5979...,\infty).
\end{cases}
\end{eqnarray}
	In the above we use the  Sch\"onflies point group notation; 
$f$ is the number of degrees of freedom of the configuration; 
e.g., for the square pyramid at $N=5$ the height varies with $s$.
	The configurations for $N\in\{5,7\}$ are taken from \cite{MeKnSm1977} for $s\geq 1$,
and from \cite{BermanHanes} for $s=-1$.
	 As to the $C_2$-symmetric $N=7$ configuration at $s=-1$, we quote from \cite{BermanHanes}: ``It consists of two points 
almost antipodal and the remaining five points sprinkled around an equatorial band.''
	This suggested to us that this configuration belongs to a family
which bifurcates off of the pentagonal bi-pyramid, which we then computed numerically ourselves to happen at $s = 0$.
	Moreover, by the nearly complete degeneracy of the $s=-2$ problem, this family of configurations
will have another bifurcation point at $s=-2$.
	Intriguingly, geometrically this is the ``same'' $C_2$ family of configurations which 
was discovered by \cite{MeKnSm1977} to bifurcate off of the pentagonal bi-pyramid at $s=2$ and 
to merge with the $C_{2v}$ configuration at $s=5$ (note, though, that \cite{MeKnSm1977} did not 
carry out a complete bifurcation analysis); yet,  presumably the five degrees of freedom in the family will 
be differently optimized in each family, in the sense that for most (if not all) $s\in(-2,0)$         
there is no $s\in(2,5)$ with the same configuration.

	\subsection{Explicit maps $s\mapsto {v}_s^{}(N)$ for $N\in\{2,3,4,5,6,7\}$} \label{sec:specific.v.s.N}
	Based on the table shown in Subsection~\ref{sec:configs.s.N} we computed explicit expressions for
$v_s^{}(N)$ in all situations where we did not have to numerically optimize any of the degrees of freedom. 
	Thus, for $N\in\{2,3,4,6\}$ and $s\in(-2,\infty)$ one has
\begin{align}\label{eq:vs.2.3.4}
	{v}_s^{}(2)
&= 
	\textstyle\frac{1}{s}\left( \big(\frac{1}{2}\big)^{s} -1 \right),\\
	{v}_s^{}(3)
&= 
	\textstyle\frac{1}{s}\left( \big(\frac{1}{3}\big)^{\frac{s}{2}}	-1 \right),\\
	{v}_s^{}(4)
&= 
	\textstyle\frac{1}{s}\left( \big(\frac{3}{8}\big)^{\frac{s}{2}}-1 \right),
\end{align}
and
\begin{equation}\label{eq:vs.6}
	{v}_s^{}(6)
= 
	\textstyle\frac{1}{s}\left(\frac{1}{5} \big(\frac{1}{2}\big)^{s}
				+ \frac{4}{5}  \big(\frac{1}{2}\big)^{\frac{s}{2}}
				-1 \right),
\end{equation}
while for $v_s^{}(5)$ an explicit expression is available for $s\in(-2,15.04807...]$, viz.
\begin{equation}
	{v}_s^{}(5)
= 
	\textstyle\frac{1}{s}\left(\frac{1}{10} \big(\frac{1}{2}\big)^{s}
				+  \frac{3}{10} \big(\frac{1}{3}\big)^{\frac{s}{2}}
				+  \frac{3}{5}  \big(\frac{1}{2}\big)^{\frac{s}{2}}
				-1 \right),
\end{equation}
and for $v_s^{}(7)$ the following expression is valid for $s\in (0,2)$:             
\begin{equation}
	{v}_s^{}(7)
= 
	\textstyle\frac{1}{s}\left(\frac{1}{21}\big(\frac{1}{2}\big)^{s}
				+ \frac{5}{21}\!
		\left(\!\left(2\sin\frac{2\pi}{5}\right)^{-s}+\left(2\sin\frac{\pi}{5}\right)^{-s}\right) 
				+  \frac{10}{21}\big(\frac{1}{2}\big)^{\frac{s}{2}}
				-1 \right).
\end{equation}

	\subsection{Bounds on $s_*^{}$ from $s\mapsto\ddot{v}_s^{}(3)$ for $s\in(-2,\infty)$} \label{sec:specific.DDvs.3}
	Only for $N=3$ can we compute $\ddot{v}_s^{}(N)$ for all $s>-2$,  viz.
\begin{equation}
	\ddot{v}_s^{}(3)
=\label{eq:DDvs.3}
	\textstyle\frac{1}{s}\left(\big(\frac{1}{2}\big)^{s} +
 \big(\frac{3}{8}\big)^{\frac{s}{2}} - 2\big(\frac{1}{3}\big)^{\frac{s}{2}}\right).
\end{equation}
	The elementary function $s\mapsto\ddot{v}_s^{}(3)$ is sufficiently simple to allow a thoroughly rigorous discussion.
	Thus, inserting $s=-2$ we obtain $\ddot{v}_{-2}^{}(3)= -1/3$, but we already knew that
any function $s\mapsto\ddot{v}_s^{}(N)$ had to be strictly negative at $s=-2$.
	On the other hand, rewriting $\ddot{v}_s^{}(3)$ as
\begin{equation}
	\ddot{v}_s^{}(3)
=\label{eq:DDvs.3.asymp}
	\textstyle\frac{1}{s}\big(\frac{3}{8}\big)^{\frac{s}{2}}	
		\left(1 -\big(\frac{8}{9}\big)^{\frac{s}{2}}\left(1 - \big(\frac{3}{4}\big)^{\frac{s}{2}} \right)\right),
\end{equation}
we can extract the large-$s$ asymptotic behavior 
$\ddot{v}_s^{}(3)\asymp\frac{1}{s}\Big(\frac{3}{8}\Big)^{\frac{s}{2}}>0$.
	As a continuous function, $\ddot{v}_s^{}(3)$ therefore has to have an odd number of zeros in $(-2,\infty)$.
	Let $s_1^{(3)}$ denote the smallest zero of $\ddot{v}_s^{}(3)$ in $(-2,\infty)$.
	Then $s_*^{} \leq s_1^{(3)}$, rigorously.

        Of course, the bound $s_*^{}\leq s_1^{(3)}$ is not explicit.
	An explicit upper bound is obtained by inserting $s=10$ into (\ref{eq:DDvs.3}), or (\ref{eq:DDvs.3.asymp}),
which yields the strictly positive (tiny) value $\ddot{v}_{10}^{}(3)= 1289/79626240$. 
        Thus we have rigorously proved:
\begin{proposition}\label{prop:rig.UP.bound.on.s.STAR}
	The critical $s_*^{}$ satisfies the upper bound $s^{}_* \leq s_1^{(3)} <10$.
\end{proposition}

	With a little extra effort one can show that $\ddot{v}_s^{}(3)$ has exactly one zero in  $(-2,\infty)$.
	Namely, starting at $s=-2$ with a strictly negative value, $\ddot{v}_s^{}(3)$ then increases monotonically 
to a single --- strictly positive --- maximum, after which it decays monotonically to zero as $s\to\infty$.
	We illustrate this behavior with a {\sc{mathematica}} plot of $s\mapsto\ddot{v}_s^{}(3)$, see Fig.~\ref{bif-ddot-v3vsS}.
\begin{figure}[H]
\centering 
\includegraphics[scale=.8]{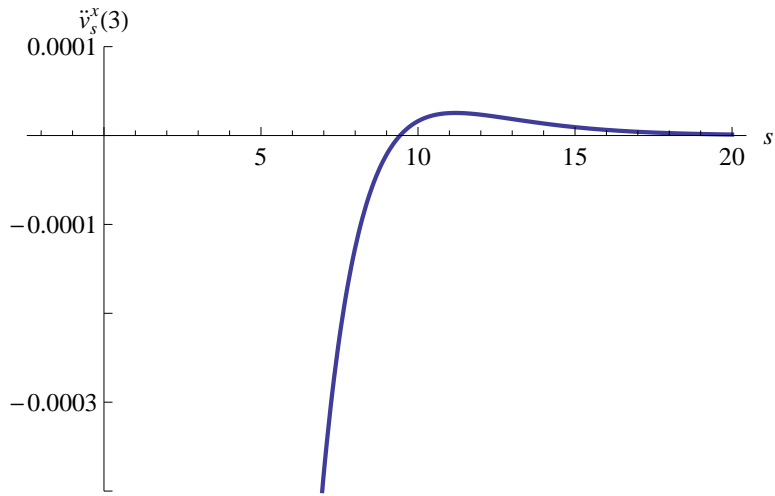}
\caption{\label{bif-ddot-v3vsS} \footnotesize{Behavior of $\ddot{v}_{s}(3)$ as a function of $s$.}}
\end{figure}
	Numerically we get $s_1^{(3)}\approx 9.4$, so that our explicit upper bound on $s_*^{}$ is not much
worse than this numerically computed bound on $s_*^{}$.

        The bound stated in Proposition \ref{prop:rig.UP.bound.on.s.STAR} is the only upper  bound on $s_*^{}$ which 
we were able to establish with complete rigor by analyzing explicitly known finite-$N$ optimizers.
	Quantitatively this bound is lousy. 
        In particular, as mentioned in the introduction, Hardin and Saff \cite{HardinSaffTWO} rigorously showed that
$N\mapsto {v}_s^{}(N)$ is asymptotically, for large $N$, strictly convex when $s>4$, which implies that $s_*^{}\leq 4$.

	In the ensuing subsections we will obtain several (much) better upper bounds, alas aided by a computer. 
	Some of these are stated as (conditional) propositions, which includes Proposition \ref{prop:N6BOUNDonSstar}.

	\subsection{Bounds on $s_*^{}$ via $s\mapsto\ddot{v}_s^{}(4)$, $-2< s\leq 15.048...$} \label{sec:specific.DDvs.4}
	For $N=4$ we can explicitly compute  $\ddot{v}_s^{}(N)$ when $-2< s\leq 15.04807...$, namely:
\begin{align}
	\ddot{v}_s^{}(4)
&=\label{eq:DDvs.4}
	\textstyle\frac{1}{s}\left(\frac{1}{10} \big(\frac{1}{2}\big)^{s}
				+  \frac{13}{10} \big(\frac{1}{3}\big)^{\frac{s}{2}}
				+  \frac{3}{5} \big(\frac{1}{2}\big)^{\frac{s}{2}}
				- 2 \big(\frac{3}{8}\big)^{\frac{s}{2}}	\right).
\end{align}
	The $s$-dependence of this elementary function is also simple enough to allow a thoroughly rigorous discussion, 
which reveals the following. 
	Thus, $\ddot{v}_{-2}^{}(4)$ is strictly negative, again as we know already it must.
	Also the large-$s$ asymptotics of r.h.s.(\ref{eq:DDvs.4}) can be worked out immediately,
r.h.s.(\ref{eq:DDvs.4})$\,\asymp\frac{1}{s}\frac{3}{5} \big(\frac{1}{2}\big)^{\frac{s}{2}}>0$,
yielding the information that r.h.s.(\ref{eq:DDvs.4}) has an odd number of zeros in $(-2,\infty)$.
	However, since the validity of the expression (\ref{eq:DDvs.4}) is restricted to 
$s\in (-2,15.04807...]$, the asymptotic result does not yield relevant information which ---
in concert with the strict negativity of $\ddot{v}_{-2}^{}(4)$ --- would allow us
to draw any conclusions about $s\mapsto\ddot{v}_s^{}(4)$ in $s\in (-2,15.04807...]$. 
	Instead we need to take a closer look.

	Inserting $s=2$, we easily find $\ddot{v}_2(4)=\frac{41}{120}>0$.
	As a continuous function, $\ddot{v}_s^{}(4)$ therefore has to have an odd number of zeros in $(-2,2)$.
	Let $s_1^{(4)}$ denote the smallest zero of $\ddot{v}_s^{}(4)$ in $(-2,2)$.
	Then $s_1^{(4)}$ is a (quasi-)rigorous\footnote{The reason for the prefix ``quasi'' is the absence of a
		rigorous proof that the triangular bi-pyramid is the $N=5$ optimizer for $s\in(-2,2+\epsilon)$.
	For such a small $N$-value it is reasonable, though, to  take the numerically found optimizers for granted.}
upper bound on $s^{}_*$.
        Also the bound $s^{}_*\leq s_1^{(4)}$ is not explicit, but $s_1^{(4)}< 2$ gives the
explicit (quasi-)rigorous upper bound $s^{}_* < 2$ which improves over $s^{}_* <10$, and also over 
$s_*^{}\leq 4$ from the asymptotic analysis \cite{HardinSaffTWO}.

	Also, with the aid of {\sc{maple}} (or such) the explicit quasi-rigorous upper bound can be improved to
$s^{}_* < 1$, for numerically ${\ddot{v}_1(4)= 0.0000745467...>0}$ (in remarkable agreement with 
the numerical data extracted from the computer experiments), so that $s_1^{(4)}< 1$.
        In fact, this can also be made quasi-rigorous.
\begin{proposition}\label{prop:rig.UP.bound.on.s.STAR.ALT}
	The critical $s_*^{}$ satisfies the upper bound $s^{}_* < 1$ as
\begin{equation}
	\ddot{v}_{1}^{}(4)
=\label{eq:DDv1.4}
	\textstyle \frac{1}{20} 
				+  \frac{13}{10} \big(\frac{1}{3}\big)^{\frac{1}{2}}
				+  \frac{3}{5} \big(\frac{1}{2}\big)^{\frac{1}{2}}
				-  \big(\frac{3}{2}\big)^{\frac{1}{2}} > 0.
\end{equation}
\end{proposition}

\begin{proof}
  First, we derive a sufficiently tight lower bound of $(1-x)^{1/2}$ in terms of a polynomial that gives a rational number if $x$ is rational.
 (The usual Taylor polynomial with Lagrange remainder term does not provide good enough bounds.) 
  Using the binomial formula, Pochhammer symbols, and Gauss hypergeometric functions, we get for every positive integer $K$
\begin{align*}
\left( 1 - x \right)^{1/2} 
&= \sum_{k=0}^\infty \binom{1/2}{k} ( - x)^k =
 \sum_{k=0}^{K-1} \frac{\Pochhsymb{-1/2}{k}}{k!} x^k + \sum_{k=K}^{\infty} \frac{\Pochhsymb{-1/2}{k}}{k!} x^k \\
&= \sum_{k=0}^{K-1} \frac{\Pochhsymb{-1/2}{k}}{k!} x^k + \frac{\Pochhsymb{-1/2}{K}}{K!} x^K \Hypergeom{2}{1}{K-1/2,1}{K+1}{x}.
\end{align*}
  Since $\Pochhsymb{-1/2}{K} = \Gamma(K-1/2) / \Gamma(-1/2) < 0$, and since the Gauss hypergeometric function is strictly increasing 
in $x$ on $[0,1]$, assuming the value $2K$ at $x = 1$ (cf. DLMF~Eq.~15.4.20\footnote{Digital Library of Mathematical Functions. 
$\backslash$url{http://dlmf.nist.gov/15.4.E20}}), we arrive at
\begin{equation*}
\left( 1 - x \right)^{1/2} > \sum_{k=0}^{K-1} \frac{\Pochhsymb{-1/2}{k}}{k!} x^k - \frac{\Pochhsymb{1/2}{K-1}}{(K-1)!} x^K \qquad \text{for $0 < x < 1$.}
\end{equation*}
 Thus, we estimate
\begin{align*}
\left( \frac{1}{3} \right)^{1/2} 
&= \left( 1 - \frac{2}{3} \right)^{1/2} >
 \sum_{k=0}^{K-1} \frac{\Pochhsymb{-1/2}{k}}{k!} \left( \frac{2}{3} \right)^k - \frac{\Pochhsymb{1/2}{K-1}}{(K-1)!} \left( \frac{2}{3} \right)^K, \\
\left( \frac{1}{2} \right)^{1/2} 
&= \left( 1 - \frac{1}{2} \right)^{1/2} > 
\sum_{k=0}^{K-1} \frac{\Pochhsymb{-1/2}{k}}{k!} \left( \frac{1}{2} \right)^k - \frac{\Pochhsymb{1/2}{K-1}}{(K-1)!} \left( \frac{1}{2} \right)^K
\end{align*}
for any positive integer $K$, whereas for even positive integers $K$ 
\begin{equation*}
\left( \frac{3}{2} \right)^{1/2} = 
\left( 1 + \frac{1}{2} \right)^{1/2} < 1 + \sum_{k=1}^{K-1} \frac{\Pochhsymb{-1/2}{k}}{k!} (-1)^k \left( \frac{1}{2} \right)^k,
\end{equation*}
because the binomial expansion of $(1+x)^{1/2}-1$ is an alternating series. 

Let $K = 20$. 
Combining everything, we get 
\begin{align*}
\ddot{v}_{1}^{}(4) 
&= \frac{1}{20} + 
\frac{13}{10} \left( \frac{1}{3} \right)^{1/2} + \frac{3}{5} \left( \frac{1}{2} \right)^{1/2} - \left( \frac{3}{2} \right)^{1/2} \\
&> \frac{5764409437417341241721}{209374412387531441339105280} > 0.
\end{align*}
For the computation of the rational bounds one can use, e.g., {\sc{mathematica}}. 
Hence, by the definition of $s^{}_*$ we obtain the desired result.
\end{proof}

        With some extra energy one should be able to make the above bounds totally rigorous, but this would require
to prove that the triangular bi-pyramid is the $N=5$ optimizer for $s\in [0,2]$, say. 

        For the sake of completeness, in Fig.~\ref{bif-ddot-v4vsS} we display a numerical plot of $s\mapsto\ddot{v}_s^{}(4)$ 
for $s\in(-2,15.04807...]$.
\begin{figure}[H]
\centering 
\includegraphics[scale=.7]{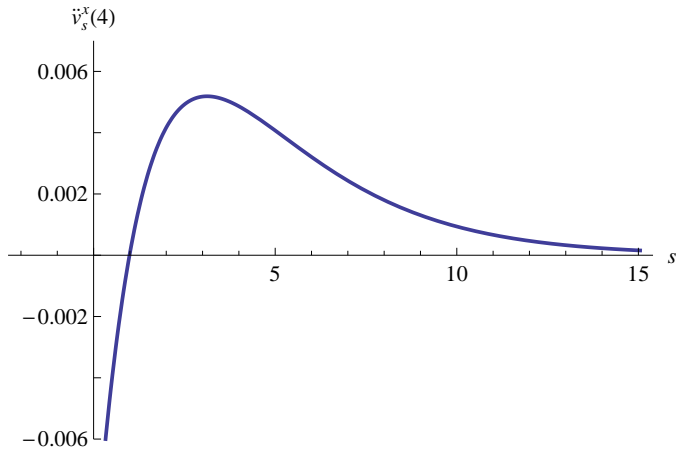}
\caption{\label{bif-ddot-v4vsS} \footnotesize{Behavior of $\ddot{v}_{s}(4)$ as a function of $s$ for $s\in(-2,15.04807...]$.}}
\end{figure}

	\subsection{Bounds on $s_*^{}$ from $s\mapsto\ddot{v}_s^{}(6)$ for $0 < s< 2$} \label{sec:specific.DDvs.6}  
        We now give a computer-assisted proof of Proposition \ref{prop:N6BOUNDonSstar}.
\begin{proof}
For $0 < s< 2$ we find          
\begin{equation}
	\ddot{v}_s^{}(6)
=\label{eq:DDvs.6}
	\textstyle\frac{1}{s}\!\left(\!\frac{5}{21}\!
	\left(\!\left(2\sin\frac{2\pi}{5}\right)^{-s}\!+\!\left(2\sin\frac{\pi}{5}\right)^{-s}\right)
			\!+\!\frac{3}{10}\big(\frac{1}{3}\big)^{\frac{s}{2}}
			\!-\!\frac{11}{21}\big(\frac{1}{2}\big)^{\frac{s}{2}}
			\!-\!\frac{53}{210}\big(\frac{1}{2}\big)^{s}\right)\!.\!\!
\end{equation}
	The $s$-dependence of this elementary function is less simple.
        Moreover, the bifurcation of the $N=7$ minimizers at $s=0$ is known to us only through our numerical bifurcation analysis, 
so that it seems prudent for now to be satisfied with a computer-assisted evaluation. 
        Fig.~\ref{bif-ddot-v6vsS} shows a numerical rendering of r.h.s.(\ref{eq:DDvs.6}).
\begin{figure}[H]
\centering 
\includegraphics[scale=.65]{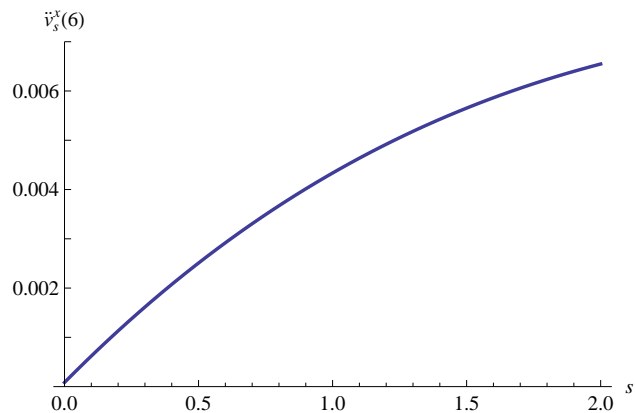}
\caption{\label{bif-ddot-v6vsS} \footnotesize{Graph of $\ddot{v}_{s}(6)$ for $s\in(0,2)$.}}
\end{figure}

        Numerically, $\lim_{s\to 0}\ddot{v}_s^{}(6)= 0.000084098... >0$, and since $\ddot{v}_{-2}^{}(6)<0$, 
we conclude that the continuous map $s\mapsto\ddot{v}_s^{}(6)$ has an odd number of zeros in $(-2,0)$.
        Let $s_1^{(6)}$ be the smallest one. 
        Then $s^{}_*\leq s_1^{(6)}<0$.
        This proves Proposition \ref{prop:N6BOUNDonSstar}.
\end{proof}

	\subsection{Remarks on $s\mapsto\ddot{v}_s^{}(12)$ for $-2< s< 2$} \label{sec:specific.DDvs.12}
	Fig.s~\ref{ddotv0vsN} and \ref{ddotv1vsN} suggest that the 
$N=6$ peak increases more rapidly with $s$ than the $N=12$ peak when $s$ runs from $0$ to $1$.
	Since the $N=12$ peak appears to be the highest peak for $s=0$ while the $N=6$ peak is barely above zero then,
we suspect that by lowering $s$ below $0$ the $N=12$ peak will vanish at a lower $s$-value than the 
$N=6$ peak, so that the first zero of $s\mapsto\ddot{v}_s^{}(12)$ will yield a better upper bound on 
$s^{}_*$ than the one we have found with $N=6$.

	Unfortunately, to evaluate $s\mapsto\ddot{v}_s^{}(12)$ one needs to know the optimal configurations for $N\in\{11,13\}$.
	We are planning a computer-assisted evaluation of $s\mapsto\ddot{v}_s^{}(12)$, which
seems to be the best one can do right now.

	\section{Rigorous bounds on $\ddot{v}_s^{}(N)$} \label{sec:rigorous.upper.bound}

\subsection{Generic  $\mathcal{O}(N^{-2})$ upper bounds on $\ddot{v}_s^{}(N)$  for $s<0$}

	We now derive the following positive $\mathcal{O}(N^{-2})$ upper bound on $\ddot{v}_s^{}(N)$ for $s<0$.
\begin{proposition}\label{prop:dd.v.s.UPPER.bound}
	For $s<0$ the second derivative of ${v}_s^{}(N)$ is bounded above by
\begin{equation}
\ddot{v}_s^{}(N) \leq -\frac{2}{(N+1)N} \left( v_s^{}(N-1) + \frac{1}{s} \right).
\end{equation}
\end{proposition}

\noindent
\begin{proof}
	We rewrite 
\begin{equation}
2v_s^{}(N) 
=\label{TWOvs}
(1+\delta)v_s^{}(N)  +(1-\delta) v_s^{}(N),
\end{equation}
where $\delta = 2/(N+1)$. 
	Next, since 
\begin{equation*}
v_s^{}(N) 
= 
\langle V_s\rangle(\omega_N^s), 
\end{equation*}
using \Ref{aveRIESZpairENERGY} we now rewrite \Ref{TWOvs} further as
\begin{align*}
2v_s^{}(N) 
&= \left( 1 + \delta \right) \frac{2}{N(N-1)} \mathop{\sum\sum}_{1\leq i < j\leq N-1} V_s(|\qV_i^s-\qV_j^s|) \\
&\phantom{=}+ \left( 1 + \delta \right) \frac{2}{N(N-1)} \sum_{1\leq i\leq N-1} V_s(|\qV_i^s-\qV_N^s|)  \\
&\phantom{=}+ \left( 1 - \delta \right) \frac{2}{N(N-1)} \mathop{\sum\sum}_{1\leq i < j\leq N} V_s(|\qV_i^s-\qV_j^s|),
\end{align*}
%
and so, using that  $\delta = 2/(N+1)$, 
\begin{equation}
\begin{split} \label{TWOvsNEUagain}
2v_s^{}(N) 
&= \frac{2(N+3)}{(N+1)N(N-1)} \mathop{\sum\sum}_{1\leq i < j\leq N-1} V_s(|\qV_i^s-\qV_j^s|) \\
&\phantom{=}+ \frac{2(N+3)}{(N+1)N(N-1)} \sum_{1\leq i\leq N-1} V_s(|\qV_i^s-\qV_N^s|)  \\
&\phantom{=}+ \frac{2}{(N+1)N} \mathop{\sum\sum}_{1\leq i < j\leq N} V_s(|\qV_i^s-\qV_j^s|).
\end{split}
\end{equation}
%

	For the term in the second line we use that $N+3 = (N-1)+4$ and rewrite
\begin{equation}
\frac{2(N+3)}{(N+1)N(N-1)}
=\label{MIDDLEsumFAKTOR}
\frac{2}{(N+1)N}+\frac{8}{(N+1)N(N-1)}.
\end{equation}
	In the single sum receiving the factor $2/(N+1)N$ we now rename the point $\qV_N^s$ into $\qV_{N+1}'$
and pool this single sum together with the double sum in the last line in \Ref{TWOvsNEUagain}, obtaining
\begin{equation}
\sum_{1\leq i\leq N-1} V_s(|\qV_i^s-\qV_{N+1}'|)
+
\mathop{\sum\sum}_{1\leq i < j\leq N} V_s(|\qV_i^s-\qV_j^s|)
=\label{MIDDLEsumDOUBLEsumMERGE}
 \frac{1}{s}+ \mathop{\sum\sum}_{1\leq i < j\leq N+1} V_s(|\qV_i'-\qV_j'|),
\end{equation}
where $\qV_k'=\qV_k^s$ for $1\leq k\leq N$ and $\qV_{N+1}'=\qV_N^s$, and where we used
that $V_s(0)=-\frac1s$ for $s<0$;
now multiplying by $2\big/(N+1)N$ and recalling \Ref{aveRIESZpairENERGY} yields
\begin{equation}
\begin{split} \label{MIDDLEsumDOUBLEsumMERGErewrite}
&\frac{2}{(N+1)N} \left( \sum_{1\leq i\leq N-1} V_s(|\qV_i^s-\qV_{N+1}'|) + \mathop{\sum\sum}_{1\leq i < j\leq N} V_s(|\qV_i^s-\qV_j^s|) \right) \\
&\phantom{equals}= \frac{2}{(N+1)N} \, \frac{1}{s} + \langle V_s\rangle(\omega_{N+1}'),
\end{split}
\end{equation}
where $\omega_{N+1}'=\{\qV_1',...,\qV_{N+1}'\}$.

	For the remaining single sum, with factor $8\big/(N+1)N(N-1)$, we note that without loss of generality we may assume that
\begin{equation*}
\sum_{1 \leq i \leq N-1} V_s(|\qV_i^s-\qV_N^s|) \geq \sum_{\substack{1 \leq i \leq N, \\ i \neq j}} V_s(|\qV_i^s-\qV_j^s|), \qquad 
\forall j = 1, \dots, N;
\end{equation*}
but then
\begin{equation*}
\frac{1}{N-1}\sum_{1\leq i \leq N-1} V_s(|\qV_i^s-\qV_N^s|) \geq \frac{2}{N(N-1)} \mathop{\sum\sum}_{1\leq i < j\leq N} V_s(|\qV_i^s-\qV_j^s|),
\end{equation*}
and therefore also
\begin{equation}
\frac{1}{N-1} \sum_{1\leq i \leq N-1} V_s(|\qV_i^s-\qV_N^s|) 
\geq \label{MIDDLEsumDOMINATESsubave}
\frac{2}{(N-1)(N-2)} \mathop{\sum\sum}_{1\leq i < j\leq N - 1} V_s(|\qV_i^s-\qV_j^s|). 
\end{equation}
	This estimate of the single sum from the middle line of \Ref{TWOvsNEUagain}
in terms of the  double sum in the first line of \Ref{TWOvsNEUagain} gives us
\begin{align*} 
&\frac{2(N+3)}{(N+1)N(N-1)} \mathop{\sum\sum}_{1\leq i < j\leq N-1} V_s(|\qV_i^s-\qV_j^s|) \\
&\phantom{equals=}+ \frac{8}{(N+1)N(N-1)} \sum_{1\leq i \leq N-1} V_s(|\qV_i^s-\qV_N^s|) \\
&\phantom{equals}\geq 
 \frac{(N+3)(N-2)+8}{(N+1)N} \, \frac{2}{(N-1)(N-2)} \mathop{\sum\sum}_{1\leq i < j\leq N - 1} V_s(|\qV_i^s-\qV_j^s|) \\
&\phantom{equals}= \left( 1 + \frac{2}{(N+1)N} \right) \langle V_s\rangle(\omega^s_{N}\setdiff \{ \qV^s_N \}).
\end{align*}

	In total, we therefore have the estimate
\begin{equation}
2v_s^{}(N) 
\geq\label{TWOvsESTIMATE}
\frac{2}{(N+1)N} \, \frac{1}{s} + \langle V_s\rangle(\omega_{N+1}')
+
\left(1+\frac{2}{(N+1)N}\right) 
\langle V_s\rangle(\omega^s_{N}\setdiff \{ \qV^s_N \}),
\end{equation}
which can be estimated again with the help of \Ref{MINaveRIESZpairENERGY} to get
\begin{equation}
2v_s^{}(N) 
\geq\label{TWOvsESTIMATEmin}
{{\frac{2}{(N+1)N}}}\, \frac{1}{s}+ v_s^{}({N+1})
+
\left(1+{{\frac{2}{(N+1)N}}}\right) 
v_s^{}({N-1}).
\end{equation}
        Regrouping terms in \Ref{TWOvsESTIMATEmin} we get
\begin{equation}
 v_s^{}(N-1) - 2v_s^{}(N) + v_s^{}(N+1)
\leq \label{DDOTvsBOUND}
-{{\frac{2}{(N+1)N}}}\left(\frac{1}{s}+ v_s^{}({N-1})\right) .
\end{equation}

The proof is complete.
\end{proof}

\begin{remark} 
Step \Ref{MIDDLEsumDOUBLEsumMERGE} requires $V_s(0)<\infty$, which is true only for $s<0$.
\end{remark} 

\begin{remark}
	For $s<0$ we have $\frac{1}{s}+ v_s^{}({N-1})<0$ so that r.h.s.\Ref{DDOTvsBOUND} $>0$.
\end{remark}

\begin{remark}
	For the special value $s=-2$ we have the exact result \Ref{vSUBminusTWO},
and a direct computation shows that $\ddot{v}_{-2}(N)= -2/(N-1)^3 <0$.
	Therefore, at least for $s=-2$, the upper estimate is not optimal.

        Also, for $s<-2$ and even $N$ we have the exact result \Ref{vsFORsBELOWminusTWO},
which converges to $v_{-2}(N)$ when taking the limit $s\uparrow -2$ of \Ref{vsFORsBELOWminusTWO}.
	Again, a simple computation shows that $\ddot{v}_{s}(N)<0$ for $s<-2$ and $N$ even.
	Modulo confirming that the situation doesn't change when $N$ is odd, the upper bound on
$\ddot{v}_s^{}(N)$ for $s<0$ would be non-optimal for all $s\leq -2$.
\end{remark} 

\begin{remark}
	The ${1}/{s}$ at r.h.s.\Ref{DDOTvsBOUND} is a consequence of our definition of $V_s(r)$. 
	By working instead with $\widetilde{V}_s(r)=r^{-s}/s$, the $1/s$ term will not show up in 
\Ref{MIDDLEsumDOUBLEsumMERGE} and \Ref{MIDDLEsumDOUBLEsumMERGErewrite}
and hence be absent in \Ref{DDOTvsBOUND}, which would read
\begin{equation}
\widetilde{v}_s^{}(N-1) - 2\,\widetilde{v}_s^{}(N) + \widetilde{v}_s^{}(N+1)
\leq \label{DDOTtildevsBOUND}
-{{\frac{2}{(N+1)N}}}\, \widetilde{v}_s^{}({N-1}).
\end{equation}
	Note that indeed, the $1/s$ at r.h.s.\Ref{DDOTvsBOUND} cancels against the $-1/s$ implicit in $v_s^{}({N-1})$, 
and l.h.s.\Ref{DDOTvsBOUND} is invariant under the change $V_s(r) \to  \widetilde{V}_s(r)$ so that
our \Ref{DDOTvsBOUND} is quantitatively identical to \Ref{DDOTtildevsBOUND}.
\end{remark}

\begin{remark}
	If one could improve the factor 2 at r.h.s.\Ref{MIDDLEsumDOMINATESsubave}
into a $3/2$, this would prove concavity.
	However, the inequality \Ref{MIDDLEsumDOMINATESsubave} turns out to be sometimes an equality
(e.g. when $s=-1$ and $N=3$ or $N=4$).
	Also \Ref{TWOvsESTIMATE}, rewritten as
\begin{equation} 
\begin{split} \label{TWOvsESTIMATEneu}
v_s^{}(N-1) - 2v_s^{}(N) + v_s^{}(N+1) 
&\leq v_s^{}(N+1) - \langle V_s\rangle(\omega_{N+1}') \\
&\phantom{=}+ v_s^{}(N-1) -\langle V_s\rangle(\omega^s_{N}\setdiff \{ \qV^s_N \}) \\
&\phantom{=}- \frac{2}{(N+1)N} \left( \frac{1}{s} + \langle V_s\rangle(\omega^s_{N}\setdiff \{ \qV^s_N \}) \right),
\end{split}
\end{equation}
%
is still sometimes an equality!
	For instance, take $s=-1$, then 
$v_{-1}(2) = 1-2$ and $v_{-1}(3) = 1-\sqrt{3}$ and $v_{-1}(4) = 1- 2\sqrt{2/3}$, so for  $N=3$ then 
$\ddot{v}_{-1}(3) = v_{-1}(2) - 2v_s^{}(3) + v_s^{}(4)= 2(-1 +\sqrt{3} -\sqrt{2/3})=-0.168891546...$,
and the exact evaluation of r.h.s.\Ref{TWOvsESTIMATEneu} yields the same answer. 
\end{remark}

\begin{remark}
	The only true inequality is in the step from \Ref{TWOvsESTIMATE} to \Ref{TWOvsESTIMATEmin}.
	Thus, to prove concavity at least for some $s$-values one has to improve this estimate.

	We have checked that for $s=-1$ and $N=3$ the sum of the second line at r.h.s.\Ref{TWOvsESTIMATEneu} 
alone does not dominate the term in the third line; namely the sum of these two terms is $(7/6)\sqrt{3}-2>0$.
	Also, the term in the first line at r.h.s.\Ref{TWOvsESTIMATEneu} alone does not dominate the one
in the third line in this case;
namely, their sum equals $\sqrt{3}(1-2\sqrt{2}/3)>0$, too.
	Therefore, to prove concavity (for some negative $s$-values, say $s\leq -1$), one will need to prove that 
the first two lines at the r.h.s.\Ref{TWOvsESTIMATEneu} \emph{together} are more negative than the last 
line is positive, at least for a certain range of negative $s$-values.
\end{remark}

\subsection{\hskip-10pt Point-energies: better $\mathcal{O}(N^{-2})$ bounds on $\!\ddot{v}_s^{}(N),\,s<0$}
\label{sec:rigorous.upper.lower.bounds}

%

        So far our analysis has only used the concept of an ``energy of a configuration of $N$ points,'' which is based on the
interpretation of $V_s(|\qV-\qV'|)$ as a \emph{pair-energy}. 
	We now bring in the well-known concept of an ``energy of an individual point,'' or \emph{point-energy} for short,
which is based on the interpretation of $V_s(|\qV-\qV'|)$ as the \emph{(potential) energy of a particle located at $\qV$ in the 
potential field of a particle at $\qV'$}, a notion which is clearly reflexive.

We recall that every (unlabeled) $N$-point configuration $\omega_N^{}$ on $\Sset^2$ can be assigned a unique compatible family 
of normalized, so-called empirical $n$-point  measures, with $1\leq n\leq N$, and vice versa.
	For instance, with the help of any labeling, the first two of these can be written as\footnote{Note that the expressions at the
		r.h.s.s are invariant under the permutation group $S_N$, which is why the mapping 
	${\omega_N^{}}\leftrightarrow \{\uli\Delta^{(n)}_{\omega_N^{}}\}_{n=1}^N$ is one-to-one only for unlabeled configurations.}
\begin{equation}\label{normalEMPmeasONE}
\uli\Delta^{(1)}_{\omega_N^{}}(d\qV)
 = 
{{\frac{1}{N}}} \sum_{1\leq j\leq N} \delta_{\qV_j}(d\qV),
\end{equation}
\begin{equation}\label{normalEMPmeasTWO}
\uli\Delta^{(2)}_{\omega_N^{}}(d\qV d\qV')
 = 
{{\frac{2}{N(N-1)}}} \sum\!\sum_{\hskip-.7truecm 1\leq j < k\leq N} \delta_{\qV_j}(d\qV)\delta_{\qV_k}(d\qV'),
\end{equation} 
and analogously one writes the empirical measures of higher order $n =3,...,N$; here, with a mild abuse of notation,
$\delta_{\qV_j}(d\qV)$ denotes the Dirac measure on $\Sset^2$ concentrated at $\qV_j$, i.e. for any Borel subset $\Lambda$ of $\Sset^2$
we have $\int_\Lambda  \delta_{\qV_j}(d\qV)=0$ if $\qV_j\not\in\Lambda$, and $\int_\Lambda  \delta_{\qV_j}(d\qV)=1$ 
if $\qV_j\in\Lambda$.
	With the help of $\uli\Delta^{(2)}_{\omega_N^{}}$ we can rewrite the average standardized Riesz pair-energy of an
$N$-point configuration $\omega_N^{}$ as
\begin{equation} \label{eq:}
 \langle V_s \rangle(\omega_N^{}) 
=
 \iint_{\Sset^2\times\Sset^2} V_s\left(\left| \qV - \qV' \right|\right) \uli\Delta^{(2)}_{\omega_N^{}}(d\qV d\qV').
\end{equation}

Moreover, with the help of $\uli\Delta^{(1)}_{\omega_N^{}}$ we can now associate with each $N$-point configuration $\omega_N^{}$ 
a standardized Riesz $s$-potential function on $\Sset^2$, given by 
\begin{equation}
\big(V_s\star\uli\Delta^{(1)}_{\omega_N^{}}\big)(\qV)
 \equiv  \label{eq:V.potential.function}
 \int_{\Sset^2} V_s\left(\left| \qV - \qV' \right|\right) \uli\Delta^{(1)}_{\omega_N^{}}(d\qV'), \qquad \qV \in \Sset^2,\quad s<0.
\end{equation}

\begin{remark}
        Since $s<0$, the potential function \Ref{eq:V.potential.function} is well-defined and finite on all of $\Sset^2$ because $V_s:\Rset_+\to\Rset$
is continuous for $s<0$.
	Of course, provided one restricts $\qV$ to $\Sset^2\!\setminus\!\omega_N^{}$, one can also extend \Ref{eq:V.potential.function} 
continuously to the regime $s\geq 0$.
	Moreover, when $s\in[0,2)$ the function $\qV'\mapsto V_s\left(\left| \qV - \qV' \right|\right)$ is weakly lower semi-continuous 
(i.e. $V_s$ is the pointwise limit of an increasing sequence of continuous functions), and therefore the potential function 
\Ref{eq:V.potential.function} is well-defined (in the sense that it may be positive infinite for certain $\qV$) also for $N\to\infty$, 
whenever $\uli\Delta^{(1)}_{\omega_N^{}}(d\qV')\rightharpoonup \mu(d\qV')$ in the weak$^*$ sense, 
where $\mu(d\qV')$ is a regular Borel measure on $\Sset^2$.
	Note that $\big(V_s\star\mu\big)(\qV)\leq \liminf\big(V_s\star\uli\Delta^{(1)}_{\omega_N^{}}\big)(\qV)$.
\end{remark}

Of special interest to us are the standardized Riesz $s$-potentials of $N$-point configurations $\omega_N^{}$ 
on $\Sset^2$ obtained from ${(N+1)}$-point configurations by removing any particular point --- or rather: their analogues with $N$ replaced by $N-1$.
	Every $\omega_N^{}$ defines a set of $N$ such ${(N-1)}$-point configurations.
	After introducing any convenient labeling of the points in $\omega_N^{}$, this set of ${(N-1)}$-point configurations reads
$\big\{{\omega_{N}^{}\setdiff \{ \qV_\ell \}}\big\}_{\ell=1}^{_N}$. 
	For every ${(N-1)}$-point configuration there is a standardized Riesz $s$-potential function on $\Sset^2$, given by 
$\big(V_s\star\uli\Delta^{(1)}_{\omega_N^{}\setdiff \{ \qV_\ell \}}\big)(\qV)$.

For every point $\qV_\ell\in \omega^{}_N,\, \ell \in\{ 1, \dots, N\}$, its \emph{average standardized Riesz point-energy 
w.r.t. the reduced ${(N-1)}$-point configuration} $\omega^{}_N \setdiff \{ \qV_\ell \}$ is simply the 
standardized Riesz $s$-potential of $\omega^{}_N \setdiff \{ \qV_\ell \}$ evaluated  at $\qV=\qV_\ell$, viz.
\begin{equation} \label{eq:average.point.energy}
\frac{1}{N-1} \sum_{\substack{1 \leq j \leq N, \\ j \neq \ell}} V_s( | \qV_j - \qV_\ell | )
=
\big(V_s\star\uli\Delta^{(1)}_{\omega_{N}^{}\setdiff \{ \qV_\ell \}}\big)(\qV_\ell),
 \quad \ell = 1, \dots, N.
\end{equation}
	Thus every $\omega_N^{}$ defines a set of $N$ point-energies 
$\big\{\big(V_s\star\uli\Delta^{(1)}_{\omega_{N}^{}\setdiff \{ \qV_\ell \}}\big)(\qV_\ell)\big\}_{\ell=1}^{_N}$. 
	The average standardized Riesz pair-energy of $\omega^{}_N$ is the mean of these:
\begin{equation}
\langle V_s \rangle( \omega_N ) 
= \label{eq:Vs.of.omegaN.as.ave.of.pt.energies} 
\frac{1}{N} \sum_{1 \leq \ell \leq N} \big(V_s\star\uli\Delta^{(1)}_{\omega_{N}^{}\setdiff \{ \qV_\ell \}}\big)(\qV_\ell).
\end{equation}
\begin{remark}
Given a minimizing configuration $\omega_N^s$, it can easily be seen that
\begin{equation}
\big(V_s\star\uli\Delta^{(1)}_{\omega_N^s\setdiff \{ \qV^s_\ell \}}\big)(\qV)
\geq 
\big(V_s\star\uli\Delta^{(1)}_{\omega_N^s\setdiff \{ \qV^s_\ell \}}\big)(\qV^s_\ell)\qquad \text{for all $\qV \in \Sset^2$},
\end{equation}
which simply says that each (generalized) unit point charge in a minimal-energy configuration occupies a point of
minimal potential energy in the potential field generated by the remaining (generalized) unit point charges.
\end{remark}


	With the help of the so-defined potential functions and point energies, we are ready to prove Proposition \ref{prop:dd.v.s.bounds.1},
which improves the upper bound in Proposition \ref{prop:dd.v.s.UPPER.bound} and also supplies a lower bound of the same type.

\begin{proof} (of Proposition \ref{prop:dd.v.s.bounds.1})
	Using the definitions \eqref{aveRIESZpairENERGY} and \eqref{eq:average.point.energy},  for each $\ell = 1, \dots, N$ we have
\begin{equation} \label{eq:master.identity}
\langle V_s \rangle( \omega_N ) 
= 
\frac{N - 2}{N} \, \langle V_s \rangle( \omega_N \setdiff \{ \qV_\ell \} ) + 
\frac{2}{N} \big(V_s\star\uli\Delta^{(1)}_{\omega_{N}^{}\setdiff \{ \qV_\ell \}}\big)(\qV_\ell).
\end{equation}
	This will be our ``master identity.''

	First, averaging \eqref{eq:master.identity} over all $\qV_\ell$, $\ell = 1, \dots, N$, and
recalling \Ref{eq:Vs.of.omegaN.as.ave.of.pt.energies}, gives
\begin{equation}
 \label{eq:identity.1a}
\langle V_s \rangle( \omega_N^{} )= \frac{1}{N} \sum_{1 \leq \ell \leq N} \langle V_s \rangle( \omega_N^{} \setdiff \{ \qV_\ell \} ).
\end{equation}
	Now replacing $\omega_N^{}$ by $\omega_N^s = \{ \qV_1^s, \dots, \qV_N^s \}$, a minimizing $N$-point configuration,
and using \Ref{MINaveRIESZpairENERGY} for ${(N-1)}$-point configurations at the r.h.s., we recover the monotonicity relation 
\begin{equation}
v_s^{}( N ) 
\geq \label{eq:monotonicity}
v_s^{}( N - 1 ) \qquad \text{for all integers $N > 2$;} 
\end{equation}
cf. \cite{Landkof,KieRMP,KieJSPeNULL}.

	Second, replacing $\omega_N^{}$ in \eqref{eq:master.identity} with $\omega_N^{} \!\cup\! \{ \qV \}$, $\qV \in \Sset^2$, and 
$\qV_\ell$ by $\qV$, yields
\begin{equation*}
\langle V_s \rangle( \omega_N \!\cup\! \{ \qV \} ) 
=
 \frac{N-1}{N+1} \langle V_s \rangle( \omega_N ) +  \frac{2}{(N +1)}\big(V_s\star\uli\Delta^{(1)}_{\omega_{N}^{}}\big)(\qV),
\end{equation*}
or, equivalently,
\begin{equation} \label{eq:identity.2}
\langle V_s \rangle( \omega_N ) 
= 
\frac{N+1}{N-1} \langle V_s \rangle( \omega_N \!\cup\! \{ \qV \} ) 
- \frac{2}{(N -1)}\big(V_s\star\uli\Delta^{(1)}_{\omega_{N}^{}}\big)(\qV).
\end{equation}
	Now setting $\qV =\qV_\ell\in\omega_N^{}$ and 
averaging \eqref{eq:identity.2} over all $\qV_\ell$, $\ell = 1, \dots, N$, and
recalling \Ref{eq:average.point.energy} and \Ref{eq:Vs.of.omegaN.as.ave.of.pt.energies}, gives
\begin{equation*}
\langle V_s \rangle( \omega_N ) 
= 
\frac{N}{N+2} \, \frac{N+1}{N-1} 
\left\{ \frac{1}{N} \sum_{1 \leq \ell \leq N} \langle V_s \rangle( \omega_N \!\cup\! \{ \qV_\ell \} ) - \frac{2V_s( 0 )}{(N+1)N} \right\}.
\end{equation*}
	Once again replacing $\omega_N^{}$ by $\omega_N^s$, then using \Ref{MINaveRIESZpairENERGY} with $N$ replaced by $N+1$
at the r.h.s., we find an estimate in the opposite direction to \Ref{eq:monotonicity},
\begin{equation}
v_s^{}( N ) 
\geq  \label{eq:monotonicityCOMPLEMENT}
\frac{(N+1)N}{(N+2)(N-1)} \, v_s^{}( N + 1) - \frac{2v_s^{}( 0 )}{(N+2)(N-1)}, \quad N > 2;
\end{equation}
note that $V_s(0) = - 1 / s$ for $s<0$.

	Next, inserting a minimizing $N$-point configuration $\omega_N^s = \{ \qV_1^s, \dots, \qV_N^s \}$ directly 
into \eqref{eq:master.identity}, and also into \eqref{eq:identity.2} with $\qV = \qV_\ell^s$, then using \Ref{MINaveRIESZpairENERGY},
yields
\begin{align}
v_s^{}( N ) 
&\geq \label{eq:estimate.A} 
\frac{N-2}{N} v_s^{}( N - 1 ) +
\frac{2}{N} \big(V_s\star\uli\Delta^{(1)}_{\omega_{N}^s\setdiff \{ \qV^s_\ell \}}\big)(\qV^s_\ell), \\
v_s^{}( N ) 
&\geq \label{eq:estimate.B}
\frac{N+1}{N-1} v_s^{}( N + 1 ) - 
\frac{2}{N} \big(V_s\star\uli\Delta^{(1)}_{\omega_{N}^s\setdiff \{ \qV^s_\ell \}}\big)(\qV^s_\ell)
- \frac{2 V_s( 0 )}{N (N - 1)} 
\end{align}
 for each $\ell = 1, \dots, N$.
	Recalling definition \eqref{DDOTvs} of $\ddot{v}_s^{}(N)$, we
split $2 v_s^{}(N)$ in \eqref{DDOTvs} into $(1-\delta) v_s^{}(N) + (1+\delta) v_s^{}(N)$, then use inequality \Ref{eq:estimate.A} 
to estimate $(1-\delta) v_s^{}(N)$ and inequality \Ref{eq:estimate.B} to estimate $(1+\delta) v_s^{}(N)$, arriving (after simplifications) at
\begin{equation}
\begin{split} \label{eq:master.inequality}
\ddot{v}_s^{}(N) 
&\leq \frac{2 + ( N - 2 ) \delta}{N} v_s^{}( N - 1 ) - \frac{2 + ( N + 1 ) \delta}{N-1} v_s^{}( N + 1 ) \\
&\phantom{=}+ 4 \delta 
	 \frac{1}{N}\big(V_s\star\uli\Delta^{(1)}_{\omega_{N}^s\setdiff \{ \qV^s_L \}}\big)(\qV^s_\ell)
- \frac{2 ( 1 + \delta )}{s} \frac{1}{N (N-1)}.
\end{split}
\end{equation}
	This relation holds for all $\ell = 1,\dots, N$ and parameters $\delta$ with 
$-1 \leq \delta \leq 1$. 

	Let $0 < \delta \leq 1$, and pick $\ell$ such that
\begin{equation} 
	\big(V_s\star\uli\Delta^{(1)}_{\omega_{N}^s\setdiff \{ \qV^s_\ell \}}\big)(\qV^s_\ell)
\leq \label{eq:potential.estimate}
	v_s^{}( N ); 
\end{equation}
that such a choice of $\ell$ is possible follows from \Ref{eq:Vs.of.omegaN.as.ave.of.pt.energies}.
	With these choices of $\delta$ and $\ell$, inequality \Ref{eq:master.inequality}
and the monotonicity relation \eqref{eq:monotonicity} together give
\begin{align*}
\ddot{v}_s^{}(N) 
&\leq \frac{2 + ( N - 2 ) \delta}{N}   v_s^{}( N ) 
    - \frac{2 + ( N + 1 ) \delta}{N-1} v_s^{}( N ) 
    + \frac{4 \delta}{N}               v_s^{}( N ) 
\\
&\qquad - \frac{2 ( 1 + \delta )}{s} \frac{1}{N (N-1)} \\
&=- \frac{2( 1 + \delta )}{N (N - 1)} \left( v_s^{}( N ) + \frac{1}{s} \right).
\end{align*}
	So nothing can be gained here. 

	Let $-2/(N+1) \leq \delta < 0$, and now pick $\ell$ such that
\begin{equation}
	\big(V_s\star\uli\Delta^{(1)}_{\omega_{N}^s\setdiff \{ \qV^s_\ell \}}\big)(\qV^s_{\ell})
\geq \label{eq:potential.estimate.II}
	v_s^{}( N ); 
\end{equation}
that such a choice of $\ell$ is possible follows again from \Ref{eq:Vs.of.omegaN.as.ave.of.pt.energies}.
	With these choices of $\delta$ and~$\ell$, inequality \Ref{eq:master.inequality}
and the estimate \eqref{eq:monotonicityCOMPLEMENT} together give
\begin{equation*}
	\ddot{v}_s^{}(N) 
\leq 
	- \frac{2( 1 + \delta )}{N (N - 1)} \left( v_s^{}( N ) + \frac{1}{s} \right), \qquad - \frac{2}{N+1} 
\leq \delta < 0,
\end{equation*}
for $2 + ( N - 2 ) \delta \geq 0$ and $2 + ( N + 1 ) \delta \geq 0$.
   Since $v_s^{}( N ) + 1 / s < 0$, we obtain for $\delta = - 2/(N + 1)$ the upper estimate exhibited in Proposition \ref{prop:dd.v.s.bounds.1},
\begin{equation*}
\ddot{v}_s^{}(N) \leq - \frac{2}{(N + 1) N} \left( v_s^{}( N ) + \frac{1}{s} \right).
\end{equation*}

	On the other hand, the estimates \eqref{eq:estimate.B} and \eqref{eq:potential.estimate} can be used to get a lower bound for 
$v_s^{}( N - 1 )$,
\begin{equation*}
v_s^{}( N - 1) \geq \frac{N}{N+1} \frac{N-1}{N-2} v_s^{}( N ) - \frac{2V_s(0)}{( N + 1 ) ( N - 2 )}.
\end{equation*}
	This estimate and the estimate $v_s^{}( N + 1 ) \geq v_s^{}( N )$ allow us to estimate the r.h.s. in the definition
\eqref{DDOTvs} of $\ddot{v}_s^{}(N)$ to  obtain the lower estimate exhibited in Proposition \ref{prop:dd.v.s.bounds.1},
\begin{align*}
\ddot{v}_s^{}( N ) 
&\geq \frac{N}{N+1} \frac{N-1}{N-2} v_s^{}( N ) - \frac{2V_s(0)}{( N + 1 ) ( N - 2 )} - 2 v_s^{}( N ) + v_s^{}( N ) \\
&= \frac{2}{( N + 1 ) ( N - 2 )} \left( v_s^{}( N ) - V_s(0) \right).
\end{align*}
%
%
%
\end{proof}

\begin{remark}
	For $\delta \leq - 2 / ( N + 1 )$ and $\ell$ as in \Ref{eq:potential.estimate.II}
we obtain from  \Ref{eq:master.inequality} (using $v_s^{}( N ) \geq v_s^{}( N - 1 )$) 
\begin{align*}
\ddot{v}_s^{}(N) 
&\leq \frac{2 + (N-2) \delta}{N} v_s^{}(N-1) 
  - \frac{2 + (N+1) \delta}{N-1} v_s^{}(N+1) 
	    + \frac{4 \delta}{N} v_s^{}( N ) \\
&\phantom{=}- \frac{2 ( 1 + \delta )}{s} \frac{1}{N (N-1)} \\
&\leq \frac{2 + (N+2) \delta}{N} v_s^{}(N-1) 
  - \frac{2 + (N+1) \delta}{N-1} v_s^{}(N+1)
\\
&\phantom{=}  - \frac{2 ( 1 + \delta )}{s} \frac{1}{N (N-1)} \\
&= \frac{2}{N} v_s^{}( N - 1 ) - \frac{2}{N-1} v_s^{}( N + 1 ) - \frac{2}{s} \frac{1}{N (N-1)} \\
&\phantom{=}+ \left( \frac{N + 2}{N} v_s^{}( N - 1 ) - \frac{N + 1}{N-1} v_s^{}( N + 1 )- \frac{2}{s} \frac{1}{N (N-1)} \right) \delta.
\end{align*}
	Both factors $2 + ( N + 2 ) \delta$ and $2 + ( N + 1 ) \delta$ are $\leq 0$. 
	As special cases one has
\begin{align*}
\ddot{v}_s^{}(N) &\leq - \frac{2}{N (N+1)} \left( v_s^{}( N - 1 ) + \frac{1}{s} \right) & &\text{if $\delta = - \frac{2}{N+1}$,} \\
\ddot{v}_s^{}(N) &\leq v_s^{}( N + 1 ) - v_s^{}( N - 1 ) & &\text{if $\delta = -1$;}
\end{align*}
the first of these inequalities is simply Proposition \ref{prop:dd.v.s.UPPER.bound}, while
the second one also follows trivially from the definition \eqref{DDOTvs} and \eqref{eq:monotonicity}.
\end{remark}

\begin{remark}
	The concept of the Riesz potential of an $N$-point configuration also allows one to obtain identities 
relating point energies and the average pair-energy when a single point is removed or added in again, as follows.

	Again inserting a minimizing $N$-point configuration $\omega_N^s = \{ \qV_1^s, \dots, \qV_N^s \}$
into \eqref{eq:master.identity}, and also into \eqref{eq:identity.2} with $\qV = \qV^s_{\ell}$,
but this time \emph{not} using \Ref{MINaveRIESZpairENERGY}, yields 
\begin{align}
	v_s^{}(N) 
&= \label{eq:v.s_ID.a}
	\frac{N - 2}{N} \, \langle V_s \rangle( \omega_N^s \setdiff \{ \qV_\ell \} ) +
	 \frac{2}{N} \big(V_s\star\uli\Delta^{(1)}_{\omega_{N}^s\setdiff \{ \qV^s_\ell \}}\big)(\qV^s_\ell),\\
	v_s^{}(N)
&= \label{eq:v.s_ID.b}
	\frac{N+1}{N-1} \, \langle V_s \rangle( \omega_N^s \!\cup\! \{ \qV^s_{\ell'} \} ) - 
	\frac{2}{N}  \big(V_s\star\uli\Delta^{(1)}_{\omega_{N}^s\setdiff \{ \qV^s_{\ell'} \}}\big)(\qV^s_{\ell'})
\end{align}
for all $\ell = 1, \dots, N$, meaning we have $2N$ (not necessarily all different) expressions for $v_s^{}(N)$.
	
	Now subtracting \Ref{eq:v.s_ID.a} ``from itself,'' with two different $\ell$-values in place, 
and the same for \Ref{eq:v.s_ID.b}, and the same for the identity r.h.s.\Ref{eq:v.s_ID.a} = r.h.s.\Ref{eq:v.s_ID.b},
and resorting, yields for all $ \ell, \ell' = 1, \dots, N$ the identities\footnote{In fact, 
			one can show these hold for all general $N$-point sets $\omega_N$.}
\begin{align*}
&\frac{N-2}{N} \left[ \langle V_s \rangle( \omega_N^s \setdiff \{ \qV_\ell \} ) 
	- \langle V_s \rangle( \omega_N^s \setdiff \{ \qV_{\ell'} \} ) \right] \\
&\phantom{equals}= \frac{2}{N} \left[ \big(V_s\star\uli\Delta^{(1)}_{\omega_{N}^s\setdiff \{ \qV^s_{\ell'} \}}\big)(\qV^s_{\ell'}) 
	- \big(V_s\star\uli\Delta^{(1)}_{\omega_{N}^s\setdiff \{ \qV^s_\ell \}}\big)(\qV^s_\ell) \right] \\
&\phantom{equals}= \frac{N+1}{N-1} \left[ \langle V_s \rangle( \omega_N^s \!\cup\! \{ \qV^s_{\ell'} \} ) 
- \langle V_s \rangle( \omega_N^s \!\cup\! \{ \qV^s_{\ell} \} ) \right].
\end{align*}

	In this vein, alternate exact representations of the discrete second derivative of $v_s^{}(N)$ follow; for instance,
we offer
\begin{equation*}
\begin{split}
\ddot{v}_s (N)
&= v_s^{}( N - 1 ) - \langle V_s \rangle( \omega_N^s \setdiff \{ \qV_\ell \} ) 
+ v_s^{}( N + 1 ) - \langle V_s \rangle( \omega_N^s \!\cup\! \{\qV^s_{\ell'} \} ) \\
&\phantom{=}
- \frac{2}{N} \left[ \big(V_s\star\uli\Delta^{(1)}_{\omega_{N}^s\setdiff \{ \qV^s_\ell \}}\big)(\qV^s_\ell) 
-  \langle V_s \rangle( \omega_N^s \setdiff \{ \qV_\ell \} ) \right.\\
&\phantom{=}\qquad \qquad  + \left.\langle V_s \rangle( \omega_N^s \!\cup\! \{ \qV^s_{\ell'} \} ) -
\big(V_s\star\uli\Delta^{(1)}_{\omega_{N}^s\setdiff \{ \qV^s_{\ell'} \}}\big)(\qV^s_{\ell'}) 
	\right]  \\
&\phantom{=}+ \frac{2}{N(N-1)} \left[ \langle V_s \rangle( \omega_N^s \!\cup\! \{ \qV^s_{\ell'} \} ) - V_s(0) \right].
\end{split}
\end{equation*}
	Note that from the definition \eqref{aveRIESZpairENERGY} it follows that the first line at the r.h.s. $<0$.
\end{remark}
	By the last line in our remark we obtain an upper bound on $\ddot{v}_s (N)$ entirely in terms of expressions 
involving only the optimal $N$-point configuration, viz.
\begin{proposition} \label{prop:dd.v.s.UPPER.bound.2}
	For $s < 0$ the map $N\mapsto \ddot{v}_s^{}(N)$ is bounded above by
\begin{equation*}
\begin{split}
\ddot{v}_s (N)
&\leq
- \frac{2}{N} \left[ \big(V_s\star\uli\Delta^{(1)}_{\omega_{N}^s\setdiff \{ \qV^s_\ell \}}\big)(\qV^s_\ell) 
- \big(V_s\star\uli\Delta^{(1)}_{\omega_{N}^s\setdiff \{ \qV^s_{\ell'} \}}\big)(\qV^s_{\ell'}) \right.\\
&\left.\phantom{\uli\Delta^{(1)}_{\omega_{N}^s\setdiff \{ \qV^s_{\ell'} \}}\big)}
  + \langle V_s \rangle( \omega_N^s \!\cup\! \{ \qV^s_{\ell'} \} ) 
	-  \langle V_s \rangle( \omega_N^s \setdiff \{ \qV_\ell \} ) \right]  \\
&\phantom{\leq }\ + \frac{2}{N(N-1)} \left[ \langle V_s \rangle( \omega_N^s \!\cup\! \{ \qV^s_{\ell'} \} ) - V_s(0) \right].
\end{split}
\end{equation*}
\end{proposition}

%
%

	\subsection{Upper bounds on $\ddot{v}_s^{}(N)$ for $N\in\{4,6,12\}$, $s\in(-2,\infty)$} \label{sec:6and12bounds}

	The upper and lower bounds on  $\ddot{v}_s^{}(N)$ presented so far are valid for general $N> 2$ and $s<0$.
	No structural information about any optimizer was used. 

	For the $N$-values of the universal configurations, viz.\footnote{Of course, the $N=2$ configuration is also universally
		optimal, but `` $\ddot{v}_s^{}(2)$'' is ill-defined.}
$N = 3$, $4$, $6$, and $12$, one can easily get better upper bounds on $\ddot{v}_s^{}(N)$ for any $s$;
though for $N=3$ this is a pointless exercise, because \Ref{eq:DDvs.3} gives the exact expression 
of $\ddot{v}_s^{}(3)$ for all $s>-2$; cf. Subsection~\ref{sec:specific.DDvs.3}, 
where we found that $\ddot{v}_s^{}(3)<0$ for $s<9.4$ (approximately).
	This leaves the cases $N\in\{4,6,12\}$.

	We begin by noting the obvious inequality
\begin{equation}\label{eq:ddvs.UP.inequ}
\ddot{v}_s^{} (N)
\leq
 \langle V_s \rangle( \omega_{N-1}^{}) -2{v}_s^{} (N) + \langle V_s \rangle( \omega_{N+1}^{}) ;
\end{equation}
here, $\omega_{N-1}^{}$ and $\omega_{N+1}^{}$ are any convenient ${(N\mp 1)}$-point configurations.
	Inequality \Ref{eq:ddvs.UP.inequ} allows us to state an immediate corollary to the results of Section~\ref{sec:rigorsSTAR}.
\begin{corollary}
	Since the optimizers for $N\!\in\!\{3,4\}$  are universal, but not the one for $N\!=\!5$, 
r.h.s.(\ref{eq:DDvs.4}) is a \emph{rigorous upper bound} to $\ddot{v}_{s}^{}(4)$ for all $s\in(-2,\infty)$.
        As a consequence, $\ddot{v}_{s}^{}(4)<0$ for all $s\in(-2,s_1^{(4)})$, with $s_1^{(4)}<1$ (yet $\approx 1$).

	Since the $N=6$ optimizer is universal, but not those for $N\in\{5,7\}$, 
r.h.s.(\ref{eq:DDvs.6}) is a \emph{rigorous upper bound} to $\ddot{v}_{s}(6)$ for all $s\in(-2,\infty)$.
        As a consequence, $\ddot{v}_{s}^{}(6)<0$ for all $s\in(-2,s_1^{(6)})$, with $s_1^{(6)}<0$ (yet $\approx 0$).
\end{corollary}

	Similarly, based on the fact that the $N=12$ optimizer is universal 
while those for $N\in\{11,13\}$ are not, we can obtain rigorous upper bounds on $\ddot{v}_s^{}(12)$.
\begin{proposition}	
	The map $s\mapsto \ddot{v}_s^{}(12)$ is bounded above by
\begin{equation}
\ddot{v}_s^{} (12)
< \label{eq:DDvs.12upperBOUND}
 \langle V_s \rangle(\omega_{12}^{s}\setdiff\{\qV_1^s\}) -2{v}_s^{} (12) + \langle V_s \rangle( \omega_{12}^{s}\!\cup\!\{\pV\}) .
\end{equation}
	In particular, choosing the epi-center of a face of the icosahedron for $\pV$ yields
\begin{equation}
\begin{split}
	\ddot{v}_s^{}(12)
< \label{eq:DDvs.12upperBOUND.expl}
& \textstyle\frac{1}{s}\!\left(\frac{1}{26}\left[
  \Big(\sqrt{1 - \frac{3+\sqrt{5}}{2\sqrt{3}\sin\frac{2\pi}{5}} +\frac{5+\sqrt{5}}{8\sin^2\frac{2\pi}{5}}}\Big)^{\,-s}
\!+ \Big(\sqrt{1 + \frac{3+\sqrt{5}}{2\sqrt{3}\sin\frac{2\pi}{5}} +\frac{5+\sqrt{5}}{8\sin^2\frac{2\pi}{5}}}\Big)^{\,-s}
\right.\right. \\
&\left.\left.
\hskip1truecm
\textstyle
+ \Big(2\sin\Big[\frac{1}{2}\arccos\frac{1}{\sqrt{3}\sin\frac{2\pi}{5}}+\frac{1}{2}\arctan\frac{1}{2}\Big]\Big)^{-s}
\right.\right. \\
&\left.\left.
\hskip1truecm
\textstyle
+ \Big(2\sin\Big[\frac{1}{2}\arcsin\frac{1}{\sqrt{3}\sin\frac{2\pi}{5}}+\frac{1}{2}\arccot\frac{1}{2}\Big]\Big)^{-s}\right]\right.\\ 
&\qquad -\left.\textstyle\frac{2}{143}\left[
  5\Big(\frac{1}{\sin\frac{2\pi}{5}}\Big)^{-s}
+ 5\Big(\frac{\sqrt{4\sin^2\frac{2\pi}{5}-1}}{\sin\frac{2\pi}{5}}\Big)^{-s} 
+ 2^{-s}\right]\right).
\end{split}
\end{equation}
\end{proposition}
	Note that the two trial configurations used to obtain \Ref{eq:DDvs.12upperBOUND.expl} are \emph{not} local energy 
minimizers (not even mechanical equilibrium configurations).
	Unfortunately, we pay a high price for having chosen these configurations which allowed us to compute an upper estimate explicitly:
r.h.s.\Ref{eq:DDvs.12upperBOUND.expl}$\,>0$ for all $s$; its minimum $\approx 0.014$ at $s\approx -1.8$.
	Presumably good upper bounds can only be obtained with the aid of a computer, by optimizing the parameters in a
well-chosen multi-parameter family of configurations.

	\subsection{ Upper and lower bounds on $\ddot{v}_s^{}(5)$ for $-2<s<\infty$}\label{sec:specific.DDvs.5}

         We begin with the rigorous lower bound.
\begin{proposition}\label{prop:rig.LOW.bound.on.ddvs.FIVE}
         For all $s\in(-2,\infty)$, we have
\begin{align}
	\ddot{v}_s^{}(5)
&\geq\label{eq:DDvs.5}
	\textstyle\frac{1}{s}\left(\big(\frac{3}{8}\big)^{\frac{s}{2}}	
				- \frac{2}{5} \big(\frac{1}{2}\big)^{\frac{s}{2}}	
				- \frac{3}{5} \big(\frac{1}{3}\big)^{\frac{s}{2}}\right).
\end{align}
\end{proposition}
\begin{proof}
	Since both $\omega_4^s$ and $\omega_{6}^s$ are rigorously known to be universal optimizers for all $s\geq -2$, 
and therefore do not depend on $s$ when $s\in(-2,\infty)$, whereas $\omega_{5}^s$ does, it follows from the definition
(\ref{MINaveRIESZpairENERGY}) that 
\begin{align}
	\ddot{v}_s^{}(5)
&\geq\label{eq:DDvs.5.INEQU}
        {v}_s^{}(4) -2 	\langle V_s\rangle (\omega_5^{\mathrm{trial}})  +{v}_s^{}(6),
\end{align}
where $\omega_5^{\mathrm{trial}}$ is any trial $5$-point configuration. 
        Picking the regular triangular bi-pyramid, and using (\ref{eq:vs.2.3.4}) and (\ref{eq:vs.6}) yields (\ref{eq:DDvs.5}).
\end{proof}

        We next vindicate Conjecture \ref{prop:vsN5isCONCAVE}.
        Thus we work under the hypothesis that the putative $5$-point minimizers listed in our table 
in Subsection~\ref{sec:configs.s.N} are actual minimizers.

\begin{proposition}\label{prop:rig.UP.bound.on.ddvs.FIVE}
         For all $s\in(-2,\infty)$, we have
\begin{align}
	\textstyle\frac{1}{s}\left(\big(\frac{3}{8}\big)^{\frac{s}{2}}	
				- \frac{2}{5} \big(\frac{1}{2}\big)^{\frac{s}{2}}	
				- \frac{3}{5} \big(\frac{1}{3}\big)^{\frac{s}{2}}\right)
&<\label{eq:DDvs.5.UPbound}
0.
\end{align}
\end{proposition}
\begin{proof}
        We break the proof down into several parts.
        First we show that l.h.s.(\ref{eq:DDvs.5.UPbound})$\,<0$ for all $s>0$.  
        We rewrite
\begin{align}\textstyle
\big(\frac{3}{8}\big)^{\frac{s}{2}}
				- \frac{2}{5} \big(\frac{1}{2}\big)^{\frac{s}{2}}	
				- \frac{3}{5} \big(\frac{1}{3}\big)^{\frac{s}{2}}
=
\big(\frac{3}{8}\big)^{\frac{s}{2}}\left(1
				- \frac{2}{5} \big(\frac{4}{3}\big)^{\frac{s}{2}}	
				- \frac{3}{5} \big(\frac{8}{9}\big)^{\frac{s}{2}}\right),
\end{align}
now use the inequality between arithmetic and geometric means to estimate
\begin{align}\textstyle
\frac{2}{5} \big(\frac{4}{3}\big)^{\frac{s}{2}}	+ \frac{3}{5} \big(\frac{8}{9}\big)^{\frac{s}{2}}
\geq\label{eq:arithmGEOMinequ}
 \big(\frac{4}{3}\big)^{\frac{s}{2}\frac{2}{5}} \big(\frac{8}{9}\big)^{\frac{s}{2}\frac{3}{5}}
=  \big(\frac{2^{13}}{3^8}\big)^{\frac{s}{10}}
>1,
\end{align}
where the last estimate follows from ${2^{13}}/{3^8}=\frac{8192}{6561}>1$ and the strict increase of the 
map $x\mapsto x^{{s}/{10}}$ for $s>0$.
        This proves that l.h.s.(\ref{eq:DDvs.5.UPbound})$\,<0$ when $0< s\leq 15.04807...$.

\begin{remark}
        For $s\geq 2$ the map $x\mapsto x^{{s}/{2}}$ is convex, so we can alternatively use Jensen's inequality to get
\begin{align}\textstyle
       \frac{2}{5} \big(\frac{1}{2}\big)^{\frac{s}{2}} + \frac{3}{5} \big(\frac{1}{3}\big)^{\frac{s}{2}}
\geq\label{eq:JensenUNDincreaseINEQU}
       \big(\frac{2}{5} \frac{1}{2} + \frac{3}{5} \frac{1}{3}\big)^{\frac{s}{2}} = 
       \big(\frac{2}{5}\big)^{\frac{s}{2}} 
>
       \big(\frac{3}{8}\big)^{\frac{s}{2}},
\end{align}
where the second inequality follows from the strict increase of the map $x\mapsto x^{{s}/{2}}$, noting that
$\frac{2}{5}=\frac{16}{40}>\frac{15}{40}=\frac{3}{8}$.
\end{remark}

        Next we show that l.h.s.(\ref{eq:DDvs.5.UPbound})$\,<0$ for all $-2\leq s<0$.  
        For $-2\leq s<0$ the map $x\mapsto x^{{|s|}/{2}}$ is concave, so we use Jensen's inequality to get
\begin{align}\textstyle
\big(\frac{8}{3}\big)^{\frac{|s|}{2}}
				- \frac{2}{5} 2^{\frac{|s|}{2}}	
				- \frac{3}{5} 3^{\frac{|s|}{2}}
\geq\label{eq:JensenUNDincreaseINEQUagain}
\big(\frac{8}{3}\big)^{\frac{|s|}{2}}
				- \big(\frac{2}{5} 2 + \frac{3}{5} 3\big)^{\frac{|s|}{2}}
=
\big(\frac{8}{3}\big)^{\frac{|s|}{2}}
				- \big(\frac{13}{5}\big)^{\frac{|s|}{2}}
>
0,
\end{align}
where the second inequality follows from $\frac{8}{3} = \frac{40}{15}>\frac{39}{15}=\frac{13}{5}$ and 
the strict increase of the map $x\mapsto x^{{|s|}/{2}}$.
       Since $s<0$, this proves that l.h.s.(\ref{eq:DDvs.5.UPbound})$\,<0$ when $-2\leq s< 0$.

       It remains to show that l.h.s.(\ref{eq:DDvs.5.UPbound})$\,<0$ when $s= 0$.
       Indeed, what we proved so far does not rule out that  l.h.s.(\ref{eq:DDvs.5.UPbound})$\,=0$ when $s= 0$.
       However, l'Hospital's rule yields
\begin{align}\textstyle
\lim_{s\to 0}	\textstyle\frac{1}{s}\left(\big(\frac{3}{8}\big)^{\frac{s}{2}}	
				- \frac{2}{5} \big(\frac{1}{2}\big)^{\frac{s}{2}}	
				- \frac{3}{5} \big(\frac{1}{3}\big)^{\frac{s}{2}}\right)
= 
\frac{1}{10}\ln \frac{3^8}{2^{13}} <0,
\end{align}
where the inequality follows from ${3^8}/{2^{13}} =\frac{6561}{8192}<1$, already used earlier.

       This completes the proof that l.h.s.(\ref{eq:DDvs.5.UPbound})$\,<0$ when $s\in(-2,\infty)$.
\end{proof}

   By our hypotheses, according to our table in Subsection~\ref{sec:configs.s.N} the regular triangular bi-pyramid is the optimizer
if $-2< s\leq 15.04807...$, thus $\ddot{v}_s^{}(5)$ = r.h.s.(\ref{eq:DDvs.5}) = l.h.s.(\ref{eq:DDvs.5.UPbound}) in this range of $s$-values.
        Since l.h.s.(\ref{eq:DDvs.5.UPbound})$\,<0$ for $s\in(-2,\infty)$, we obtain
\begin{corollary}
We have that $\ddot{v}_s^{}(5)<0$ when $-2< s\leq 15.04807...$; see Fig.~\ref{bif-ddot-v5vsS}.
\begin{figure}[H]
\centering 
\includegraphics[scale=.7]{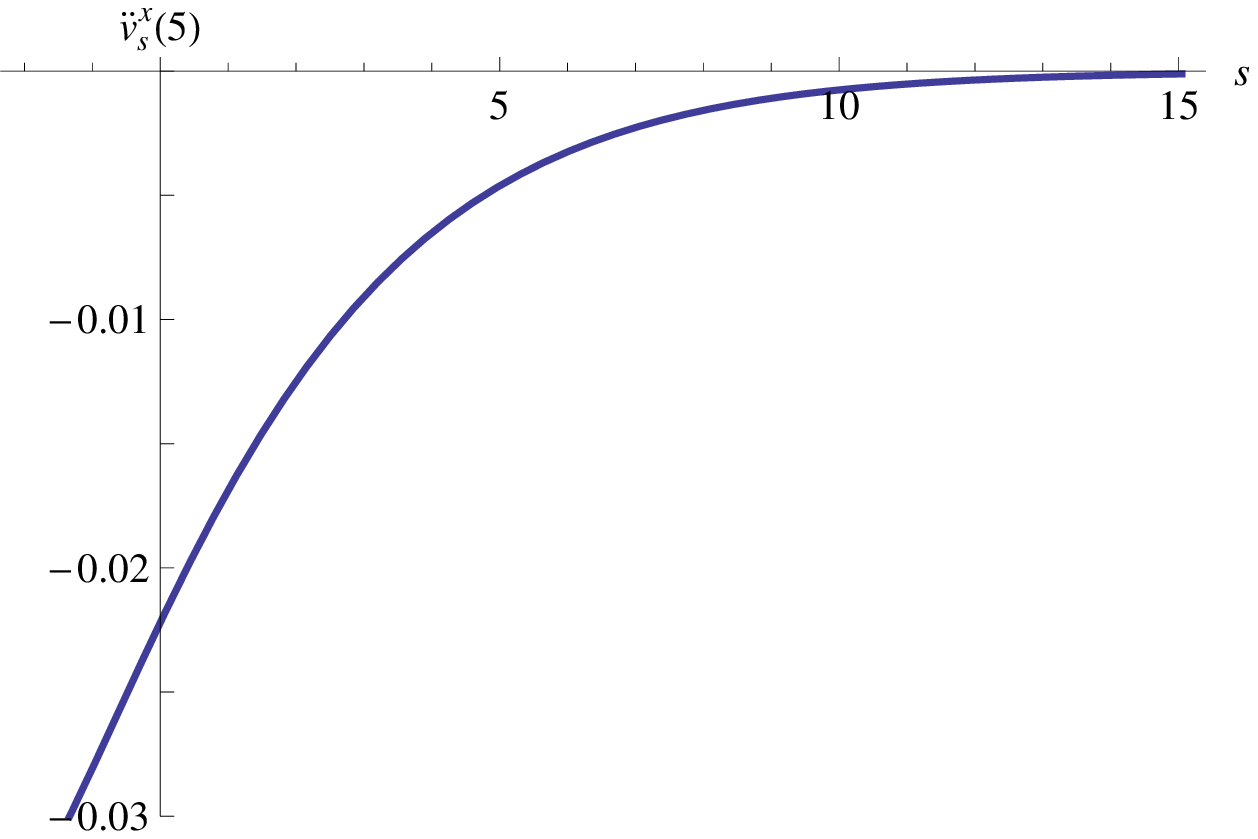}
\caption{\label{bif-ddot-v5vsS} \footnotesize{Behavior of $\ddot{v}_{s}(5)$ as a function of $s$ for $s\in(-2,15.08407...)$.}}
\end{figure}
\end{corollary}

	Unfortunately, empirically we know that r.h.s.(\ref{eq:DDvs.5}) is not the correct formula for 
$\ddot{v}_s^{}(5)$ when $s\in(15.04807...,\infty)$.
        Therefore our estimate (\ref{eq:arithmGEOMinequ}) does not rule out that $\ddot{v}_s^{}(5)$ 
may become positive for \emph{some} $s\in(15.04807...,\infty)$.
        We now summarize our computer-generated evidence that $\ddot{v}_s^{}(5)<0$ for $s\in(15.04807...,\infty)$.

        By our hypotheses, the optimal $5$-point configuration for $s\in(15.04807...,\infty)$ is the
square pyramid with adjusted height. 
        So for $s\in(15.04807...,\infty)$ the optimal $5$-point average Riesz $s$-pair-energy is given by
${v}_s^{}(5) = \langle V_s\rangle (\omega_{\mathrm{sqpyr}}^s)$, with
\begin{align}
\hskip-10pt
\langle V_s\rangle (\omega_{\mathrm{sqpyr}}^s)
=
	\textstyle\frac{1}{s}\left(\frac{2^{1-{s}/{2}}}{5} \big(1-z\big)^{-\frac{s}{2}} 
	+ \frac{2^{1-{s}/{2}}}{5} \big(1-z^2\big)^{-\frac{s}{2}} 
        + \frac{2^{-s}}{5} \big(1-z^2\big)^{-\frac{s}{2}} -1 \right),
\end{align}
where $z$ is the unique solution in $(-1,0)$ of the equation
\begin{align}\label{eq:sqpyrFIXptEQ}
 \big(1+z\big)^{1+\frac{s}{2}} +\big(2+2^{-{s}/{2}}\big)z =0;
\end{align}
the solution of (\ref{eq:sqpyrFIXptEQ}) is the ``$z$-coordinate'' of the base of the pyramid if its tip is at $z=1$.
For $s>15$, equation (\ref{eq:sqpyrFIXptEQ}) has no explicit solution in terms of elementary functions, 
but l.h.s.(\ref{eq:sqpyrFIXptEQ}) is convex and a Newton scheme very rapidly converges downward from $z_0^{}=0$,
for all $s>15$.
        With the help of the Newton scheme we found that
\begin{align}
        {v}_s^{}(4) -2 \langle V_s\rangle (\omega_{\mathrm{sqpyr}}^s)  +{v}_s^{}(6) 
<\label{eq:DDvs.5.sINTERMED}
0
\end{align}
for $15<s<10^6$.       
    To make this completely rigorous we need to show that asymptotically $\ddot{v}_s^{}(5) <0$, with good error bounds,
which we hope to supply in future work.

       The vindication of Conjecture \ref{prop:vsN5isCONCAVE} is complete.

	\section{An asymptotic point of view} \label{sec:asymptotics}
	A rigorous proof of the concavity of $N\mapsto v_s^{}(N)$ for some regime of $s$-values (possibly
for all $s<s_*^{}$, with $s_*^{}$ conjectured to be in $(-1,0)$) 
most likely has to be based on a combination of different strategies.
	Attempts to rigorously identify all the optimizers and to check for the negativity of $\ddot{v}_s^{}(N)$ by ``explicit computation'' 
may be feasible for sufficiently small $N$-values, but clearly are bound to fail even for moderately large $N$-values. 
	On the other hand, the very-large-$N$ regime is to some extent accessible by asymptotic analysis.
	Ultimately the goal is to determine the complete asymptotic large-$N$ expansion of $N\mapsto v_s^{}(N)$, but at least as
many terms as possible. 
	With enough hard work one may be able to extend the asymptotic control down to sufficiently 
small $N$-values to establish an overlap with some explicitly controlled small-$N$ regime. 
	In the previous section we presented the type of analysis suitable for the small-$N$ regime.
	This section is devoted to asymptotic analysis.

	It has already been mentioned in the introduction that the limit $\displaystyle{\lim_{N\to\infty}}v_s^{}(N)$ is given by
the variational principle \Ref{eq:limit.VP}.
	Explicitly
\begin{equation}\label{eq:limit.VP.repeat}
\lim_{N\to\infty}v_s^{}(N)
=
\inf_{\mu\in \Psp(\Sset^2)} \iint_{\!\!\!\Sset^2\times\Sset^2} \frac{1}{s}\left(\frac{1}{|\pV-\qV|^s}-1\right)\mu(d\pV)\mu(d\qV),
\end{equation}
where $\Psp(\Sset^2)$ is the set of all Borel probability measures supported on $\Sset^2$. 
	By classical potential theory (\cite{Bjorck} for $s < 0$, and \cite{Landkof} for $s > 0$) 
one knows that for $s < 2$, 
\begin{equation} \label{eq:energy.intgral}
\mathcal{I}_s[\mu] \equiv \int_{\Sset^2} \int_{\Sset^2} \frac{1}{\left| \pV - \qV \right|^s} \mu( d \pV) \mu( d \qV)
\end{equation}
is well-defined and finite for any $\mu\in\Psp(\Sset^2)$.
	On $\Psp(\Sset^2)$ the Riesz $s$-energy integral has a degenerate maximum with value $2^{|s|-1}$ when $s\leq -2$, and it has a 
unique maximum when $-2<s<0$, respectively a unique minimum when $0<s<2$, achieved at the normalized surface area measure $\sigma$ 
on $\Sset^2\subset\mathbb{R}^3$, with value 
\begin{equation}
\mathcal{I}_s[\sigma] = \frac{2^{1-s}}{2-s} 
\equiv \label{eq:W.s}
W_s;
\end{equation}
this expression is also valid at $s=0$, with $\mathcal{I}_0[\mu]\equiv 1$ for all $\mu\in\Psp(\Sset^2)$. 
	For $s \geq 2$ the energy integral \Ref{eq:energy.intgral} is $+\infty$ for any $\mu\in\Psp(\Sset^2)$.
	Altogether, therefore,
\begin{equation}\label{eq:limit.NtoINFTY.vs.N}
\lim_{N\to\infty}v_s^{}(N)
=
\begin{cases}
\frac{1}{s} \left( 2^{-s-1}-1 \right) &\quad\mbox{if}\quad s\in (-\infty,-2], \\
\frac{1}{s}\left(W_s -1\right) &\quad\mbox{if}\quad s\in [-2,2), \\
\infty &\quad\mbox{if}\quad s\in [2,\infty) ,
\end{cases}
\end{equation}
with the special case $s=0$ included as  
$\lim\limits_{s\to 0} \frac{1}{s}\left(W_s -1\right) = \frac12+\ln\frac12$.
	We need to know how these limiting ``values'' are approached by $v_s^{}(N)$ as $N\to\infty$.

	For the regime $s< -2$ we already have an exact formula for $v_s^{}(N)$ valid for all \emph{even} $N$, namely \Ref{vsFORsBELOWminusTWO}.
	Yet, as mentioned in Appendix \ref{sec:appdx.A}, the $s<-2$ regime is largely unsettled for odd $N$, and we have nothing to add to this 
here. 
	Moreover, the case $s=-2$ is completely solved, with $v_{-2}^{}(N)$ given by \Ref{vSUBminusTWO}.
	Thus, henceforth we will discuss the regime $s>-2$. 
	Clearly, we need to further make the distinction between the subregimes $s\in (-2,2)$ and $s\geq 2$.
	
	As to the regime $s\in(-2,2)$, to find the large-$N$ asymptotic expansion of $v_s^{}(N)$ one seeks the powers
$\alpha_s^{}\in\Rset$ and $\beta_s^{}\in\Rset$ for which a nontrivial limit
\begin{equation}\label{eq:limit.vs.NEXT.term}
\lim_{N\to\infty}
N^{\alpha_s} (\ln{N})^{\beta_s} 
\Big( v_s^{}(N) - \textstyle\frac{1}{s}\left(W_s -1\right)\Big)
=
\widehat{C}_s
\end{equation}
exists; then subtracts $\widehat{C}_s$ from the expression under the limit in \Ref{eq:limit.vs.NEXT.term}, 
multiplies by $N^{\alpha_s^\prime} (\ln{N})^{\beta_s^\prime}$ with new powers $\alpha_s^\prime$ and $\beta_s^\prime$, and repeats the procedure; etc.
	This strategy has to some extent been carried out in the literature, see \cite{BrHaSa2012} for a recent account.
	We will call upon the results of \cite{BrHaSa2012} in a moment. 
	To pave the way for our discussion we first recall the definition \Ref{MINaveRIESZpairENERGY} of $v_s^{}(N)$, 
involving our standardized Riesz pair-energy $V_s(r)$ given in \Ref{RieszVs}, and note that the additive term
$\frac1s$ in the difference $ v_s^{}(N) - \frac{1}{s}\left(W_s -1\right)$
cancels out. 
          This motivates the definition of the {\em \centered Riesz pair-energy} 
\begin{equation}  \label{eq:U.s}
U_s( r ) \equiv s^{-1} \left( r^{-s} - W_s \right), \qquad s < 2,\quad s \neq 0,
\end{equation}
and its $s\to 0$ limit
\begin{equation}  \label{eq:U.log}
U_0( r) \equiv -\log{r} - W_{\mathrm{log}},
\end{equation}
where
\begin{align*}
W_{\mathrm{log}} 
&\equiv 
\inf \Big\{ \int_{\Sset^2} \int_{\Sset^2} \log \frac{1}{\left| \pV - \qV \right|} \mu( d \pV ) \mu( d \qV ) :
 \mu \in \Psp( \Sset^2) \Big\}
\end{align*}
is the $s$-derivative of
\begin{equation*}
\textstyle W_s = 1 + \left( \frac{1}{2} - \log 2 \right) s + \mathcal{O}(s^2) \qquad {\mathrm{as}}\quad s \to 0,
\end{equation*}
evaluated at $s=0$, with value $W_{\mathrm{log}} (= W_0^\prime) = \frac{1}{2} + \log \frac{1}{2}  = - 0.193147... < 0$.

	We define the \emph{average \centered Riesz pair-energy of a configuration}~by
\begin{equation}\label{eq:mean.U.s.OF.omega}
\langle U_s \rangle( \omega_N ) 
\equiv 
\frac{2}{N \left( N - 1 \right)} \mathop{\sum \sum}_{1 \leq j < k \leq N} U_s( \left| \qV_j - \qV_k \right| ),\quad s\in (-2,2),
\end{equation}
and the \emph{minimal average \centered Riesz pair-energy} by\footnote{The energy functionals $\langle U_s \rangle( \omega_N )$ 
		(without the normalization $1 / (N (N-1))$) were studied 
		by Wagner~\cite{Wa1990, Wa1992} who first derived two-sided bounds for optimal $N$-point configurations in terms of the 
		correct order of decay of $N$ for the complete range $-2 < s < 2$.}
\begin{equation}\label{eq:minimal.mean.U.s.OF.omega}
u_s^{}( N ) \equiv \inf_{\omega_N \subset \Sset^2} \langle U_s \rangle( \omega_N ),\quad s\in (-2,2).
\end{equation}
	We note that for each $s\in(-2,2)$ the map $N\mapsto u_s^{}(N)$ is monotonically increasing and that, by construction, the limit is $0$. 
	Hence $u_s^{}(N) \leq 0$~for~all~$N$.

	As for the regime $s \geq 2$, since $v_s^{}(N)\to\infty$ as $N\to\infty$, it would seem that the asymptotic analysis has to be set up 
in a somewhat different manner. 
	In fact, this is true for $s=2$ (see below).
	However, it is worth noting that r.h.s.\Ref{eq:W.s} is defined on the complex $s$-plane except 
at the single and simple pole at $s = 2$ and thus gives the analytic continuation of $W_s$ to the complex $s$-plane. 
	We will denote this meromorphic function by the same symbol, $W_s$.  
	Understood in this analytically extended way, \Ref{eq:limit.vs.NEXT.term} and the ensuing description,
the definition \Ref{eq:U.s}, as well as \Ref{eq:mean.U.s.OF.omega} and \Ref{eq:minimal.mean.U.s.OF.omega}, 
all make sense also for the regime $s>2$. 
	Note, though, that for $s >2$ the map $N\mapsto u_s^{}(N)$ diverges monotonically to $+\infty$ as $N\to\infty$
(hence $u_s^{}(N) \not\leq 0$~for~most~$N$ when $s>2$), so that the power $\alpha_s^{}\leq 0$ for $s>2$ (and if $\alpha_s^{}=0$,
then $\beta_s^{}<0$), while for $s<2$ the power $\alpha_s^{}\geq 0$ (and if $\alpha_s^{}=0$, then $\beta_s^{}>0$).
	While all this may seem just like a convenient coincidence, we shall see in Subsection~\ref{sec:hyper.sing}
that the analytic continuation to $s>2$ actually seems to have some deeper significance for the aysmptotic problem; 
cf. \cite{BrHaSa2012}. 

	As to the singular case $s=2$,
it obviously makes no sense to subtract the infinite term ``$\frac{1}{2}W_2$'' from $v_2^{}(N)$, or from $V_2(r)$.
	Yet, if one replaces \Ref{eq:limit.vs.NEXT.term} by
\begin{equation}\label{eq:limit.vTWO.NEXT.term}
\lim_{N\to\infty}
N^{\alpha_2^{}} (\ln{N})^{\beta_2^{}} 
\big( v_2^{}(N) +\textstyle\frac12\big)
=
\widehat{C}_2;
\end{equation}
the ensuing description remains valid, with the definitions \Ref{eq:mean.U.s.OF.omega} and \Ref{eq:minimal.mean.U.s.OF.omega}
extended to $s=2$ with the help of the definition $U_2(r) \equiv \frac12 \frac{1}{r^2}$.
	Note that $N\mapsto u_2^{}(N)$ diverges monotonically to $+\infty$ as $N\to\infty$, so the same remarks apply as for the hypersingular
regime $s>2$ regarding the powers $\alpha_s^{}$ and $\beta_s^{}$.

	Clearly, for each $s\in(-2,\infty),\, s\neq 2$, the \centered pair-energy $U_s(r)$ on $\Sset^2$ and the standardized pair-energy $V_s(r)$ 
differ only by a constant, viz. 
\begin{equation}  \label{eq:connecting.formula}
U_s(r) = V_s(r) + s^{-1}({1-W_s}), \quad r> 0,
\end{equation}
with the case $s=0$ understood as limit $s\to 0$, viz. 
$U_0(r) = V_0(r) + W_{\rm{log}}$.
	As a consequence, all their discrete $N$-derivatives coincide; in particular, we have
\begin{lemma}\label{lemma:DDu.is.DDv}
	For all $s>-2$ ($s\neq 2$) we have 
\begin{align}
\dot{u}_s^{\pm}(N) 
&=\label{eq:Du.is.Dv}
\dot{v}_s^{\pm}(N), \\
\ddot{u}_s^{}(N) 
&=\label{eq:DDu.is.DDv}
\ddot{v}_s^{}(N) .
\end{align}
\end{lemma}

	In the remaining subsections we will elaborate on the asymptotic expansion of $N\mapsto u_s^{}(N)$ 
and its implications for the asymptotic expansion of $N\mapsto \ddot{v}_s^{}(N)$.
	Subsection~\ref{sec:pot.asymp} is concerned with the potential regime $s\in (-2,2)$, Subsection~\ref{sec:hyper.sing} with the hypersingular regime $s>2$, and Subsection~\ref{sec:sing} with the singular case $s=2$.
	As mentioned earlier, a discussion of the ``degeneracy regime'' $s\leq -2$ has to be left for some future work. 


\subsection{The potential-theoretical regime $-2 < s < 2$\label{sec:pot.asymp}} 

\subsubsection{The non-logarithmic cases}
	In the non-logarithmic cases one has the following bounds for $u_s^{}(N)$.
\begin{proposition} \label{prop:bounds.for.u.s.N}
	Let $-2 < s < 2$ with $s \neq 0$. 
	Then there exist positive $s$-dependent constants $C > c > 0$ such that for all sufficiently large $N \geq 2$ 
\begin{align*}
\frac{1}{s} \, \frac{W_s - c \, N^{s/2}}{N-1} &\leq u_s^{}( N ) \leq \frac{1}{s} \, \frac{W_s - C \, N^{s/2}}{N-1} & &\text{if $-2 < s < 0$,} \\
\frac{1}{s} \, \frac{W_s - C \, N^{s/2}}{N-1} &\leq u_s^{}( N ) \leq \frac{1}{s} \, \frac{W_s - c \, N^{s/2}}{N-1} & &\text{if $0 < s < 2$.}
\end{align*}
\end{proposition}
\begin{proof}
	We will call upon the results of \cite{BrHaSa2012} and references cited therein.
	To facilitate the identification of the relevant results in the pertinent literature, we introduce 
the {\em optimal Riesz $s$-energy of $N$ points on $\Sset^2$}, defined for $s\neq 0$~by
\begin{equation*}
\pzcE_s(\Sset^2, N ) 
\equiv 
\sign(s)\inf\Big\{2 \mathop{\sum \sum}_{1\leq j<k\leq N}\frac{\sign(s)}{\left|\qV_{j}-\qV_{k}\right|^s}:\{\qV_1, \dots, \qV_N\}\subset\Sset^2 \Big\},
\end{equation*}
cf. \cite{HardinSaffONE}.
	Then it is known\footnote{It is furthermore well-known that for a sequence $\{\omega_{N}^{s}\}_{N \geq 2}$ of optimal 
		$N$-point configurations with $-2 < s < 2$ ($s \neq 0$) one has the limit relation
\begin{equation} \label{eq:limit.relation}
\lim_{N \to \infty} 
\frac{2}{N \left( N - 1 \right)} \mathop{\sum \sum}_{1 \leq j < k \leq N} \left| \qV_{j,N}^{s} - \qV_{k,N}^{s} \right|^{-s} = W_s.
\end{equation}
		The reciprocal of the quantity under the limit symbol is also known as the \emph{$N$-th generalized diameter of $\Sset^2$}. 
		It converges monotonically to the \emph{generalized transfinite diameter of $\Sset^2$}, introduced by 
	P{\'o}lya and Szeg{\H{o}} in \cite{PoSz1931}, which equals the so-called \emph{$s$-capacity} $1/W_s$ of $\Sset^2$.
		Incidentally, it should also be noted that $- | \pV - \qV |^{-s}$ is a conditionally positive definite function 
		of order $1$ for $-2 < s < 0$; 
	cf. \cite{Sch1938}.}
(cf. \cite{BrHaSa2012}) that there are $s$-dependent constants $C > c > 0$ such that 
\begin{equation*}
-C \, N^{1+s/2} \leq \pzcE_s(\Sset^2, N )  - W_s \, N^2 \leq - c \, N^{1+s/2} 
\end{equation*}
for all sufficiently large $N \geq 2$. 
	Hence
\begin{equation*}
-C \, \frac{N^{1+s/2}}{N (N - 1)} 
\leq 
\frac{\pzcE_s(\Sset^2, N ) }{N (N - 1)} - W_s + \left( 1 - \frac{N^2}{N (N - 1)} \right) W_s  \leq - c \, \frac{N^{1+s/2}}{N (N - 1)}, 
\end{equation*}
or, equivalently,
\begin{equation*}
\frac{W_s-C \, N^{s/2}}{N - 1} \leq \frac{\pzcE_s(\Sset^2, N ) }{N (N - 1)} - W_s \leq \frac{W_s- c \, N^{s/2}}{N - 1}. 
\end{equation*}
	The desired relations follow by multiplying the last relations with $1/s$ and using the definition of the \centered Riesz pair-energy.
\end{proof}
	From the bounds in Proposition~\ref{prop:bounds.for.u.s.N} we get estimates for the discrete second derivative of $u_s^{}( N )$ 
(and by means of Lemma \ref{lemma:DDu.is.DDv} for $v_s^{}( N )$).
\begin{proposition} \label{prop:ddot.u.s.bounds.1}
	Let $s\in(-2,2),\, s \neq 0$. 
	Let $c$ and $C$ be the constants from Proposition~\ref{prop:bounds.for.u.s.N}.
	For $s\in(-2,0)$ the discrete second derivative of $u_s^{}(N)$ satisfies
\begin{equation*}
\ddot{u}_s( N ) 
\leq 
- \frac{2(C-c)}{s} \, \frac{N^{s/2}}{N-1} + \frac{W_s}{s} \, \frac{1}{N(N-1)(N-2)} - \frac{C}{s} \frac{N^{s/2} H_s(N) }{N(N-1)(N-2)},
\end{equation*}
where $H_s(N) \to (2-s)(4-s) > 0$ as $N \to \infty$. 
	The right-hand side becomes a lower bound by interchanging $c$ and $C$. 

	The bounds for $0<s<2$ are obtained by interchanging $c$ and $C$ in the bounds for $\ddot{u}_s(N)$ with $-2<s<0$. 
\end{proposition}

\begin{proof}
	Let $-2 < s < 0$. 
	First, we consider the upper bound. 
	By the definition of $\ddot{u}_s( N )$ and Proposition~\ref{prop:bounds.for.u.s.N}
\begin{align*}
\ddot{u}_s( N ) 
&\leq \frac{1}{s} \, \frac{W_s - C \, (N-1)^{s/2}}{N-2} - \frac{2}{s} \, \frac{W_s - c \, N^{s/2}}{N-1} + \frac{1}{s} \, \frac{W_s - C \, (N+1)^{s/2}}{N} \\
&= \frac{W_s}{s} \left[ \frac{1}{N-2} - \frac{2}{N-1} + \frac{1}{N} \right] \\
&\quad - \frac{1}{s} \left\{ \frac{C (N-1)^{s/2}}{N-2} - \frac{2c N^{s/2}}{N-1} + \frac{C (N+1)^{s/2}}{N} \right\} \\
&= - \frac{2(C-c)}{s} \, \frac{N^{s/2}}{N-1} + \frac{W_s}{s} \, \frac{1}{N(N-1)(N-2)}\\
&\quad - \frac{C}{s} \left\{ \frac{(N-1)^{s/2}}{N-2} - \frac{2 N^{s/2}}{N-1} + \frac{ (N+1)^{s/2}}{N} \right\}.
\end{align*}
	Since the function (using the integral representation of the gamma function)
\begin{equation*}
f( x ) \equiv \frac{x^{s/2}}{x-1} 
= 
\sum_{n=0}^\infty \frac{1}{x^{n+1-s/2}} = \int_0^\infty e^{-x t} \sum_{n=0}^\infty \frac{t^{n-s/2}}{\Gamma(n+1-s/2)} \, d t, \quad x > 1,
\end{equation*}
is strictly monotonically decreasing and convex, the last expression in braces is strictly positive. 
	Series expansion (assisted by {\sc{mathematica}}) reveals
\begin{equation*}
\frac{(N-1)^{s/2}}{N-2} - \frac{2 N^{s/2}}{N-1} + \frac{ (N+1)^{s/2}}{N} 
= \frac{N^{s/2} \left\{ (2-s)(4-s) - \frac{s(s-4)}{2} N^{-1} + \cdots \right\}}{N(N-1)(N-2)}.
\end{equation*}
Thus
\begin{equation*}
\ddot{u}_s( N ) 
\leq - \frac{2(C-c)}{s} \, \frac{N^{s/2}}{N-1} + \frac{W_s}{s} \, \frac{1}{N(N-1)(N-2)} - \frac{C}{s} \frac{N^{s/2} H_s(N) }{N(N-1)(N-2)},
\end{equation*}
where $H_s(N) \to (2-s)(4-s) > 0$ as $N \to \infty$. 

	For a lower bound one has to interchange $C$ and $c$. 

	For $s\in (0,2)$, the above computations hold with $c$ and $C$ interchanged.~\end{proof}

	Proposition~\ref{prop:ddot.u.s.bounds.1}, although much weaker than Proposition~\ref{prop:dd.v.s.bounds.1}, 
clearly shows that one needs more information about the asymptotic behavior of $u_s^{}(N)$ for large $N$. 

	The investigation of the asymptotic behavior of $\pzcE_s(\Sset^2, N ) $ for large $N$ yields the following fundamental conjecture. 
	(We refer the interested reader to \cite{BrHaSa2012} and papers cited therein.)

\begin{conjecture} \label{conj:fundamental.conjecture}
	Let $-2 < s < 2$ and $s \neq 0$. 
	Then there exists a constant $C_s$ and a function $\Omega_s( N )$ such that
\begin{equation*}
u_s^{}( N ) 
= \frac{1}{s} \frac{C_s}{(4\pi)^{s/2}} \, N^{s/2-1} + \frac{W_s}{s} N^{-1} + \Omega_s( N ) \quad \text{for all $N \geq 2$}
\end{equation*}
and $N^{1-s/2} \Omega_s( N ) \to 0$ as $N \to \infty$.
\end{conjecture}

It should be noted that the $N^{-1}$-term may be dominated by the function $\Omega_s(N)$ for $0 < s < 2$. 
Indeed, at presence it is unclear how fast $N^{1-s/2} \Omega_s( N )$ tends to zero as $N \to \infty$. 
One suggestion is that $N^{(s-1)/2-1}$ (or even $N^{s/2-2}$) is the correct order of $N$ for the next term. 
The numerical evidence is inconclusive in this regard. 
Worse, the properly normalized $\Omega_s( N )$ may be oscillating with bounded non-zero amplitude as $N$ becomes large. 

\begin{proof}[Motivation of Conjecture~\ref{conj:fundamental.conjecture}] 
	For $-2 < s < 2$ and $s \neq 0$, the following conjecture for the large-$N$ behavior of $\pzcE_s(\Sset^2, N ) $ is known (cf. 
\cite{BrHaSa2012} for a most recent account)
\begin{equation*}
\pzcE_s(\Sset^2, N )  = W_s \, N^2 + \frac{C_s}{(4\pi)^{s/2}} \, N^{1+s/2} + \mathcal{R}_s( N ),
\end{equation*}
where $\mathcal{R}_s( N ) \big/ N^{1+s/2} \to 0$ as $N \to \infty$. 
	Hence,
\begin{align*}
u_s^{}( N ) 
&= \frac{1}{s} \left\{ \frac{\pzcE_s(\Sset^2, N ) }{N (N - 1)} - W_s \right\} \\
&= \frac{1}{s} \left\{ \left( \frac{N}{N-1} - 1 \right) W_s + \frac{C_s}{(4\pi)^{s/2}} \frac{N^{1+s/2}}{N (N - 1)} + \frac{\mathcal{R}_s( N )}{N (N - 1)} \right\} \\
&= \frac{1}{s} \left\{ \frac{C_s}{(4\pi)^{s/2}} \, N^{s/2-1} + W_s \, N^{-1} + \frac{W_s + \frac{C_s}{(4\pi)^{s/2}} \, N^{s/2} + \mathcal{R}_s( N )}{N (N - 1)} \right\} \\
&= \frac{1}{s} \frac{C_s}{(4\pi)^{s/2}} \, N^{s/2-1} + \frac{W_s}{s} \, N^{-1} + \Omega_s( N ),
\end{align*}
where
\begin{equation*}
\Omega_s( N ) = \frac{1}{N (N - 1)} \left\{ \frac{W_s}{s} + \frac{1}{s} \frac{C_s}{(4\pi)^{s/2}} \, N^{s/2} + \frac{1}{s} \frac{\mathcal{R}_s( N )}{N^{1+s/2}} \, N^{1+s/2} \right\}.
\end{equation*}
	By the assumption on $\mathcal{R}_s( N )$ it follows that $N^{1-s/2} \Omega_s( N ) \to 0$ as $N \to \infty$.
\end{proof}

	Conjecture~\ref{conj:fundamental.conjecture} imposes the following large-$N$ behavior for the discrete second derivative of $u_s^{}(N)$.
\begin{corollary}
	Let $-2 < s < 2$ with $s \neq 0$. 
	Under the assumption that Conjecture~\ref{conj:fundamental.conjecture} is true, there holds
\begin{equation} \label{eq:asymptotic.1}
\ddot{u}_s( N ) = \frac{(1-s/2)(2-s/2)}{s} \, \frac{C_s}{(4 \pi)^{s/2}} \, N^{s/2-3} + \frac{2}{s} W_s \, N^{-3} + \ddot{\Omega}_s( N ) + F_s( N ) 
\end{equation}
with $N^{4-s/2} F_s( N ) \to 0$ as $N \to \infty$.
\end{corollary}

\begin{proof}
	By the definition of the discrete second derivative (cf. \eqref{DDOTvs})
\begin{align*}
 \ddot{u}_s( N )
&= \frac{1}{s} \, \frac{C_s}{(4 \pi)^{s/2}} \, 
 N^{s/2-1} \left\{ \left( 1 - \frac{1}{N} \right)^{s/2-1} - 2 - \left( 1 + \frac{1}{N} \right)^{s/2-1} \right\} \\
&\phantom{=}+ \frac{W_s}{s} \, 
 N^{-1} \left\{ \left( 1 - \frac{1}{N} \right)^{-1} - 2 - \left( 1 + \frac{1}{N} \right)^{-1}  \right\} + \ddot{\Omega}_s( N ). 
\end{align*}
	Series expansion gives (here $\Pochhsymb{a}{n}$ is the \emph{Pochhammer symbol} or rising factorial)
\begin{equation*}
\Big( 1 - \frac{1}{N} \Big)^{s/2-1} - 2 - \Big( 1 + \frac{1}{N} \Big)^{s/2-1} 
= 
\sum_{n=1}^\infty \frac{\Pochhsymb{1-s/2}{n}}{n!} \, \frac{1+(-1)^n}{N^n} 
\end{equation*}
and simplification yields
\begin{equation*}
\Big( 1 - \frac{1}{N} \Big)^{-1} - 2 - \Big( 1 + \frac{1}{N} \Big)^{-1} 
= \frac{2}{N^2-1}. 
\end{equation*}
	Hence
\begin{equation*}
 \ddot{u}_s( N ) 
= \frac{\Pochhsymb{1-s/2}{2}}{s} \, 
 \frac{C_s}{(4 \pi)^{s/2}} \, N^{s/2-3} + \frac{2}{s} W_s \, N^{-3} + \ddot{\Omega}_s( N ) + F_s( N ),
\end{equation*}
where 
\begin{equation*}
\begin{split}
 F_s( N ) 
&\equiv \frac{1}{s} \, \frac{C_s}{(4 \pi)^{s/2}} \, N^{s/2-1} \Big[ \Big(
 1 - \frac{1}{N} \Big)^{s/2-1}\!\! - 2 - \Big( 1 + \frac{1}{N} \Big)^{s/2-1}\!\! - \frac{\Pochhsymb{1-s/2}{2}}{N^2} \Big] \\
&\phantom{=}+ \frac{W_s}{s} \, \frac{2}{N^3 (N^2 - 1)}
\end{split}
\end{equation*}
and $N^{4-s/2} F_s( N ) \to 0$ as $N \to \infty$, since $-2<s<2$.
\end{proof}
\begin{remark}
	For $-2 < s < 0$ the dominant term in \eqref{eq:asymptotic.1} is $( 2 W_s / s ) N^{-3}$ or possibly $\ddot{\Omega}_s(N)$. 
	Thus, for sufficiently large $N$ the sign of $\ddot{u}_s(N)$ is negative or interference from higher order terms in the conjectured 
asymptotic expansion of $u_s^{}(N)$ forces a positive sign. 
The numerical evidence for $s=-1$ (no exceptional $N<200$ with non-negative discrete 
second derivative of $v_{-1}^{}(N)$) discussed earlier supports a negative sign of $\ddot{v}_{-1}^{}(N)$,
viz. $\ddot{u}_{-1}^{}(N)$.
	In fact, if $\ddot{u}_s(N) \geq 0$ for infinitely many growing $N_k$, then \eqref{eq:asymptotic.1} would imply that 
\begin{equation*}
\ddot{\Omega}_s( N_k ) = \Omega_s( N_k - 1 ) - 2 \Omega_s( N_k ) + \Omega_s( N_k + 1 ) \geq - \textstyle\frac{2}{s} W_s N_k^{-3} > 0.
\end{equation*}
\end{remark}
\begin{remark}
	For $0 < s < 2$ the dominant term in \eqref{eq:asymptotic.1} is the $C_s$-term or possibly $\ddot{\Omega}_s(N)$. 
	It is not even known that the constant $C_s$ appearing in Conjecture~\ref{conj:fundamental.conjecture} and thus 
in \eqref{eq:asymptotic.1} exists. 
	Results for the hypersingular case discussed below suggest that $C_s$ is related to the Epstein zeta function for 
the hexagonal lattice in the plane. 
	An inspection of the graph of the conjectured value of $C_s$, analytically continued, shows that it is negative in the interval $(0,2)$! 
Asymptotically seen, higher-order terms in the large-$N$ expansion of $u_s^{}(N)$, $0 < s < 2$ cause 
the appearance of ``large'' magic numbers $N$. 
	Indeed, analysis of the putatively minimal average \centered Riesz pair-energy up to $N=200$ for $s=1$ gives a sequence of $N$'s 
for which $\ddot{u}_s^{}(N)$ is not negative. 
\end{remark}
\subsubsection{The logarithmic case}

	In the logarithmic case one has the following bounds for $u_0(N)$.
\begin{proposition} \label{prop:log.bounds}
	There exist positive constants $C > c > 0$ such that for all sufficiently large $N$
\begin{equation*}
\frac{W_{\mathrm{log}} - \frac{1}{2} \log N - C \, N}{N (N - 1)} \leq u_0( N ) \leq \frac{W_{\mathrm{log}} - \frac{1}{2} \log N - c \, N}{N (N - 1)}.
\end{equation*}
\end{proposition}
\begin{proof}
	Let $\pzcE_{\mathrm{log}}(\Sset^2,N)$ be defined by
\begin{equation*}
\pzcE_{\mathrm{log}}(\Sset^2,N)
\equiv 
\inf\Big\{2\textstyle\mathop{\sum \sum}\limits_{1\leq j<k\leq N}\log\frac{1}{\left|\qV_{j}-\qV_{k}\right|}:\{\qV_1,\dots,\qV_N\}\subset\Sset^2\Big\}.
\end{equation*}
	Then it is known\footnote{It is also well-known that, for a sequence $\{\omega_{N}^{\mathrm{log}}\}_{N \geq 2}$ of optimal  
		$N$-point configurations,
\begin{equation*} 
\textstyle\lim\limits_{N \to \infty} \frac{2}{N \left( N - 1 \right)} 
\mathop{\sum \sum}\limits_{1 \leq j < k \leq N}\log\frac{1}{\big| \qV_{j,N}^{\mathrm{log}}-\qV_{k,N}^{\mathrm{log}} \big|} = W_{\mathrm{log}}.
\end{equation*}
		The reciprocal of the exponential function of the quantity under the limit symbol is also known as the 
\emph{$N$-th diameter of $\Sset^2$ in the logarithmic case}. It converges monotonically to the 
\emph{transfinite diameter of $\Sset^2$ (in the logarithmic case)}, introduced in \cite{PoSz1931}, which equals the \emph{logarithmic capacity}  
$\exp(-W_{\mathrm{log}})$ of $\Sset^2$; see \cite{Pritsker} for a recent account.}
that there are constants $C > c > 0$ such that (cf. \cite{BrHaSa2012})
\begin{equation*}
-C \, N \leq \pzcE_{\mathrm{log}}(\Sset^2,N) - W_{\mathrm{log}} \, N^2 + \textstyle\frac{1}{2} N \log N \leq - c \, N 
\end{equation*}
for all sufficiently large $N \geq 2$. 
	Hence
\begin{equation*}
\frac{- \frac{1}{2} \log N - C \, N}{N (N - 1)} 
\leq 
\frac{\pzcE_{\mathrm{log}}(\Sset^2,N)}{N(N-1)} - W_{\mathrm{log}} - \frac{1}{N - 1} W_{\mathrm{log}} 
 \leq 
\frac{- \frac{1}{2} \log N - c \, N}{N (N - 1)}. 
\end{equation*}
	The desired bounds follow. 
\end{proof}

	Proposition~\ref{prop:log.bounds} provides the following weak bounds for the discrete second derivative of $u_0(N)$
(and by means of Lemma \ref{lemma:DDu.is.DDv} for $v_s^{}( N )$).
\begin{proposition}
	Let $c$ and $C$ be the constants from Proposition~\ref{prop:log.bounds}. 
	Then 
\begin{equation*}
\begin{split}
\ddot{u}_0(N) 
&\leq \frac{2(C-c)}{N-1} - \frac{2 c}{N (N - 1) (N - 2)} \\
&\phantom{=}- \frac{3 \log N}{(N + 1) N (N - 1) (N - 2)} + \frac{6 W_{\mathrm{log}} + H_0(N)}{(N + 1) N (N - 1) (N - 2)},
\end{split}
\end{equation*}
where $H_0(N) \to 5/2$ as $N \to \infty$. 

	For the corresponding lower bound one has to interchange $c$ and $C$.
\end{proposition}

\begin{proof}
	First we consider the upper bound. 
	By the definition of $\ddot{u}_0(N)$ and Proposition~\ref{prop:log.bounds}
\begin{align*}
\ddot{u}_0(N) 
&\leq \frac{W_{\mathrm{log}}-\frac{1}{2}\log(N-1)-c\,(N-1)}{(N - 1)(N - 2)} - 2 \frac{W_{\mathrm{log}} - \frac{1}{2} \log N - C \, N}{N(N-1)}\\
&\phantom{=}+ \frac{W_{\mathrm{log}} - \frac{1}{2} \log( N + 1 ) - c \, ( N + 1)}{(N + 1) N} \\
&= W_{\mathrm{log}} \left( \frac{1}{(N - 1)(N - 2)} - 2 \frac{1}{N (N - 1)} + \frac{1}{(N + 1) N} \right) \\
&\phantom{=}- \frac{1}{2} \left[ \frac{\log( N - 1 )}{(N - 1)(N - 2)} - 2 \frac{\log N}{N (N - 1)} + \frac{\log( N + 1 )}{(N + 1) N} \right] \\
&\phantom{=}- \left\{ \frac{c}{N - 2} - 2 \frac{C}{N - 1} + \frac{c}{N} \right\}.
\end{align*}
	Simplification (assisted by {\sc{mathematica}}) gives for the expression in parenthesis
\begin{equation*}
\frac{6}{(N + 1) N (N - 1) (N - 2)},
\end{equation*}
for the expression in braces
\begin{equation*}
-2 \frac{C-c}{N-1} + c \left( \frac{1}{N-2} - \frac{2}{N-1} + \frac{1}{N} \right) = - 2 \frac{C-c}{N-1} + \frac{2c}{N (N - 1) (N - 2)}
\end{equation*}
and for the square-bracketed expression (using series expansion)
\begin{equation*}
\begin{split}
&\frac{6 \log N}{(N + 1) N (N - 1) (N - 2)} + \frac{\log(1-\frac{1}{N})}{(N - 1)(N - 2)} + \frac{\log(1+\frac{1}{N})}{(N + 1) N} \\
&\phantom{equals}= \frac{6 \log N}{(N + 1) N (N - 1) (N - 2)} + \frac{-\frac{5}{N^2} + \frac{8}{N^3} + \cdots}{(N - 1) (N - 2)}.
\end{split}
\end{equation*}
	Putting everything together, we arrive at
\begin{equation*}
\begin{split}
\ddot{u}_0(N) 
&\leq \frac{2(C-c)}{N-1} - \frac{2 c}{N (N - 1) (N - 2)} \\
&\phantom{=}- \frac{3 \log N}{(N + 1) N (N - 1) (N - 2)} + \frac{6 W_{\mathrm{log}} + \frac{5}{2} - \frac{3}{2N} + \dots}{(N + 1) N (N - 1) (N - 2)}.
\end{split}
\end{equation*}

	For the corresponding lower bound one has to interchange $c$ and $C$.
\end{proof}
	The investigation of the asymptotic behavior of $\pzcE_{\mathrm{log}}(\Sset^2,N)$ provides the following conjecture.
\begin{conjecture} \label{conj:log.conjecture}
	There exists a constant $C_{\mathrm{log}}$ and a function $\Omega_0( N )$ such that
\begin{equation*}
u_0( N ) 
= - \frac{1}{2} \frac{\log N}{N} + \frac{W_{\mathrm{log}} + C_{\mathrm{log}}}{N} 
- \frac{1}{2} \frac{\log N}{N^2} + \frac{W_{\mathrm{log}} + C_{\mathrm{log}}}{N^2} + \Omega_0( N ) \quad \forall N \geq 2,
\end{equation*}
where $N \Omega_0( N ) \to 0$ (or in its stronger form $N^2 \Omega_0( N ) \to c_0 \neq 0$) as $N \to \infty$.

	The constant $C_{\mathrm{log}}$ is given by
\begin{equation*}
C_{\mathrm{log}} = 2 \log 2 + \frac{1}{2} \log \frac{2}{3} + 3 \log \frac{\sqrt{\pi}}{\Gamma(1/3)} = -0.05560530494339251850\ldots < 0.
\end{equation*}
\end{conjecture}

\begin{proof}[Motivation of Conjecture~\ref{conj:log.conjecture}]
	The following conjecture for the large-$N$ behavior of $\pzcE_{\mathrm{log}}(\Sset^2,N)$ is known (cf. \cite{BrHaSa2012} 
for a most recent account)
\begin{equation*}
\pzcE_{\mathrm{log}}(\Sset^2,N) = W_{\mathrm{log}} \, N^2 - \frac{1}{2} N \log N + C_{\mathrm{log}} \, N + \mathcal{R}_{\mathrm{log}}( N ),
\end{equation*}
where $\mathcal{R}_{\mathrm{log}}( N ) \big/ N \to 0$ as $N \to \infty$. 
	A stronger form states that $\mathcal{R}_{\mathrm{log}}( N )$ converges to a non-zero constant. 
	Hence
\begin{align*}
u_0( N )
&= \frac{\pzcE_{\mathrm{log}}(\Sset^2,N)}{N(N-1)} - W_{\mathrm{log}} \\
&= \underbrace{\left( \frac{N^2}{N (N - 1)} - 1 \right)}_{1/(N-1)} W_{\mathrm{log}} - \frac{1}{2} \frac{\log N}{N-1} + \frac{C_{\mathrm{log}}}{N-1} + \frac{\mathcal{R}_{\mathrm{log}}( N )}{N(N-1)} \\
&= - \frac{1}{2} \frac{\log N}{N} + \frac{W_{\mathrm{log}} + C_{\mathrm{log}}}{N} - \frac{1}{2} \frac{\log N}{N^2} + \frac{W_{\mathrm{log}} + C_{\mathrm{log}}}{N^2} + \Omega_0( N ), 
\end{align*}
where we used that
\begin{equation*}
\frac{1}{N-1} = \frac{1}{N} + \frac{1}{N (N - 1)} = \frac{1}{N} + \frac{1}{N^2} + \frac{1}{N^2 (N - 1)} 
\end{equation*}
and
\begin{equation*}
\Omega_0( N ) 
\equiv 
\frac{1}{N (N - 1)} \left\{ \mathcal{R}_{\mathrm{log}}( N ) - \frac{1}{2} \frac{\log N}{N} + \frac{W_{\mathrm{log}}}{N}  \right\}.
\end{equation*}
	From the assumptions on $\mathcal{R}_{\mathrm{log}}( N )$ it follows that $N \Omega_0(N) \to 0$ 
(or in its stronger form $N^2 \Omega_0(N) \to c_0 \neq 0$) as $N \to \infty$. 
\end{proof}

	Conjecture~\ref{conj:log.conjecture} implies the following large-$N$ behavior for the discrete second derivative of $u_0(N)$.
\begin{corollary}
	Under the assumption that Conjecture~\ref{conj:log.conjecture} is true, there holds
\begin{equation} \label{eq:u.0.N.asymptotic}
\ddot{u}_0( N ) = - \frac{\log N}{N^3} + \frac{3/2 + W_{\mathrm{log}} + C_{\mathrm{log}}}{N^3} + \ddot{\Omega}_0( N ) + F_0( N ),
\end{equation}
with $N^{4-\eps} F_0( N ) \to 0$ as $N \to \infty$ for any $\eps > 0$.
\end{corollary}
	Note that $3/2 + W_{\mathrm{log}} + C_{\mathrm{log}} > 0$.

\begin{proof}
	By the definition of the discrete second derivative (cf. \eqref{DDOTvs})
\begin{align*}
\ddot{u}_0( N ) 
&= - \frac{1}{2} \left( \frac{\log(N - 1)}{N-1} - 2 \frac{\log N}{N} + \frac{\log(N + 1)}{N+1} \right) \\
&\phantom{=}+ \left( W_{\mathrm{log}} + C_{\mathrm{log}} \right) \left[ \frac{1}{N-1} - 2 \frac{1}{N} + \frac{1}{N+1} \right] \\
&\phantom{=}- \frac{1}{2} \left( \frac{\log(N - 1)}{(N - 1)^2} - 2 \frac{\log N}{N^2} + \frac{\log(N + 1)}{(N + 1)^2} \right) \\
&\phantom{=}+ \left( W_{\mathrm{log}} + C_{\mathrm{log}} \right) \left[ \frac{1}{(N - 1)^2} - 2 \frac{1}{N^2} + \frac{1}{(N + 1)^2} \right]
+ \ddot{\Omega}_0( N ).
\end{align*}
	The series expansions
\begin{equation*}
\Big( 1 - \frac{1}{N} \Big)^{-\alpha} = \sum_{n=0}^{\infty} \frac{\Pochhsymb{\alpha}{n}}{n!} \frac{1}{N^{n}}
\end{equation*}
and
\begin{align*}
\log\Big( 1 - \frac{1}{N} \Big) \Big( 1 - \frac{1}{N} \Big)^{-\alpha} 
& = - \sum_{k=1}^\infty \sum_{n=0}^{\infty} \frac{1}{k} \frac{\Pochhsymb{\alpha}{n}}{n!} \frac{1}{N^{k+n}} 
\\
& = - \sum_{\ell=1}^\infty 
\left\{ \sum_{n=0}^{\ell-1} \frac{1}{\ell-n} \frac{\Pochhsymb{\alpha}{n}}{n!} \right\} \frac{1}{N^\ell}
\end{align*}
yield
\begin{align*}
\frac{1}{N-1} - \frac{2}{N} + \frac{1}{N+1} 
& = \frac{1}{N} \left( \frac{1}{1-\frac{1}{N}} - 2 + \frac{1}{1+\frac{1}{N}} \right) 
 = \frac{1}{N} \sum_{n=2}^\infty \frac{1+(-1)^n}{N^n}
\end{align*}
and thus
\begin{align*}
&\frac{\log(N - 1)}{N-1} - 2 \frac{\log N}{N} + \frac{\log(N + 1)}{N+1} \\
&\phantom{XX}= \log N\! \left[\! \frac{1}{N-1} - 2 \frac{1}{N} + \frac{1}{N+1}\! \right] 
+ \frac{1}{N} \left[\! \frac{\log(1 - \frac{1}{N})}{1-\frac{1}{N}} + \frac{\log(1 + \frac{1}{N})}{1+\frac{1}{N}}\! \right] \\
&\phantom{XX}= \frac{\log N}{N^3} \sum_{n=0}^\infty \frac{1+(-1)^n}{N^n} - \sum_{\ell=2}^\infty H_\ell \frac{1+(-1)^\ell}{N^{\ell+1}},
\end{align*}
where $H_\ell$ denotes the $\ell$-th Harmonic number $H_\ell \equiv \sum_{k=1}^\ell \frac{1}{k}$; furthermore
\begin{align*}
&\frac{1}{(N-1)^2}-2\frac{1}{N^2}+\frac{1}{(N + 1)^2} = \frac{1}{N^2} \left( \frac{1}{(1-\frac{1}{N})^2} - 2 + \frac{1}{(1+\frac{1}{N})^2} \right) \\
&\phantom{equals}= \frac{1}{N^2} \sum_{n=2}^{\infty} \frac{\Pochhsymb{2}{n}}{n!} \frac{1+(-1)^n}{N^{n}} =
 \sum_{n=2}^{\infty} (n + 1) \frac{1+(-1)^n}{N^{n+2}}
\end{align*}
and thus
\begin{align*}
&\frac{\log(N - 1)}{(N - 1)^2} - 2 \frac{\log N}{N^2} + \frac{\log(N + 1)}{(N + 1)^2} 
\\
&\phantom{eq}= \log N \left( \frac{1}{(N - 1)^2} - 2 \frac{1}{N^2} + \frac{1}{(N + 1)^2} \right)
\\
&\phantom{equa} + \frac{1}{N^2} \left( \frac{\log(1 - \frac{1}{N})}{(1-\frac{1}{N})^2} +
 \frac{\log(1 + \frac{1}{N})}{(1+\frac{1}{N})^2} \right) 
\\
&\phantom{eq} = \frac{\log N}{N^4} \sum_{n=0}^{\infty} (n + 3) \frac{1+(-1)^n}{N^{n}} 
- \sum_{\ell=2}^\infty 
\left\{ \sum_{n=0}^{\ell-1} \frac{1}{\ell-n} \frac{\Pochhsymb{2}{n}}{n!} \right\} 
\frac{1+(-1)^\ell}{N^{\ell+2}}.
\end{align*}

	Combining everything (and shifting indices of summations), we arrive at
\begin{align*}
\ddot{u}_0( N ) 
&= - \frac{1}{2} \left( \frac{\log N}{N^3} \sum_{n=0}^\infty \frac{1+(-1)^n}{N^n} -
 \frac{1}{N^3} \sum_{\ell=0}^\infty H_{\ell+2} \frac{1+(-1)^\ell}{N^{\ell}} \right) \\
&\phantom{=}+ \frac{W_{\mathrm{log}} + C_{\mathrm{log}}}{N^3} \sum_{n=0}^\infty \frac{1+(-1)^n}{N^{n}} \\
&\phantom{=}- \frac{\log N}{2N^4} \sum_{n=0}^{\infty} (n + 3) \frac{1+(-1)^n}{N^{n}} \\
&\phantom{=}
+ \frac{1}{2N^4} \sum_{\ell=0}^\infty \left\{ \sum_{n=0}^{\ell+1} \frac{n+1}{\ell+2-n}  \right\} \frac{1+(-1)^\ell}{N^{\ell}} \\
&\phantom{=}+ \frac{W_{\mathrm{log}} + C_{\mathrm{log}}}{N^4} \sum_{n=0}^{\infty} (n + 3) \frac{1+(-1)^n}{N^{n}} \\
&\phantom{=}+ \ddot{\Omega}_0( N ).
\end{align*}
	Rearranging the terms gives
\begin{equation*}
\ddot{u}_0( N ) = - \frac{\log N}{N^3} + \frac{3/2 + W_{\mathrm{log}} + C_{\mathrm{log}}}{N^3} + \ddot{\Omega}_0( N ) + F_0( N ),
\end{equation*}
where
\begin{align*}
F_0( N ) 
&\equiv - \frac{1}{2} \frac{\log N}{N^3} \sum_{n=2}^\infty \frac{1+(-1)^n}{N^n} + \frac{1}{2} \frac{1}{N^3} \sum_{\ell=2}^\infty H_{\ell+2} \frac{1+(-1)^\ell}{N^{\ell}} \\
&\phantom{=}+ \frac{W_{\mathrm{log}} + C_{\mathrm{log}}}{N^3} \sum_{n=2}^\infty \frac{1+(-1)^n}{N^{n}} \\
&\phantom{=}- \frac{\log N}{2N^4} \sum_{n=0}^{\infty} (n + 3) \frac{1+(-1)^n}{N^{n}} \\
&\phantom{=} + \frac{1}{2N^4} \sum_{\ell=0}^\infty \left\{ \sum_{n=0}^{\ell+1} \frac{n+1}{\ell+2-n} 
 \right\} \frac{1+(-1)^\ell}{N^{\ell}} \\
&\phantom{=}+ \frac{W_{\mathrm{log}} + C_{\mathrm{log}}}{N^4} \sum_{n=0}^{\infty} (n + 3) \frac{1+(-1)^n}{N^{n}} 
\end{align*}
with $N^{4-\eps} F_0( N ) \to 0$ as $N \to \infty$ for any $\eps > 0$.
\end{proof}
\begin{remark}
	The dominant term in \eqref{eq:u.0.N.asymptotic} is the negative $- ( \log N ) / N^3$ term or possibly $\ddot{\Omega}_0(N)$. 
	An increasing infinite sequence of magic numbers $(N_k)$ with $\ddot{u}_0(N_k) \geq 0$ would be caused by higher-order terms 
in the conjectured asymptotic expansion of $u_0(N)$ and thus would, for example, exclude the hypothetical expansion 
$\Omega_0( N ) = c_0 + \widetilde{\Omega}_0( N )$ with $\widetilde{\Omega}_0( N ) / \big[ ( \log N ) / N  \big] \to 0$ as $N \to \infty$. 
\end{remark}

\subsection{The hypersingular regime $s > 2$\label{sec:hyper.sing}} 
%
	For $s > 2$ it is proved in Hardin and Saff~\cite{HardinSaffTWO} that
\begin{equation*}
\lim_{N \to \infty} \frac{\pzcE_s(\Sset^2, N ) }{N^{1+s/2}} = \frac{C_{s}}{(4 \pi)^{s/2}}
\end{equation*}
for some constant $C_{s}$ depending on $s$.
	In \cite{KuijlaarsSaff} it is shown that, for $s > 2$,
\begin{equation*}
\limsup_{N \to \infty} \frac{\pzcE_s(\Sset^2, N ) }{N^{1+s/2}} 
\leq \frac{\left( \sqrt{3} / 2 \right)^{s/2} \zetafcn_{\Lambda}(s)}{\left( 4 \pi \right)^{s/2}},
\end{equation*}
and it is conjectured in \cite{KuijlaarsSaff} that for $s > 2$, 
\begin{equation} \label{eq:C.s}
C_{s} = \left( \sqrt{3} / 2 \right)^{s/2} \zetafcn_{\Lambda}(s),
\end{equation}
where $\zetafcn_\Lambda$ is the zeta function associated with the hexagonal lattice.
	The zeta function associated with the hexagonal lattice admits the factorization
\begin{equation*} 
\textstyle
\zetafcn_\Lambda(s) = 6 \zetafcn\big(\frac{s}{2}\big) \DirichletL_{-3}\big(\frac{s}{2}), \qquad \re s > 2,
\end{equation*}
where $\zeta(s)$ is the Riemann zeta function and  $\DirichletL_{-3}(s)$ a Dirichlet $\DirichletL$-Series, viz.
\begin{align*}
	\zetafcn(s) 
&=
 1 + \frac{1}{2^s} + \frac{1}{3^s} + \frac{1}{4^s} + \frac{1}{5^s} + \frac{1}{6^s} + \cdots, \qquad \re s > 1, \\
	\DirichletL_{-3}(s) 
&= 1 - \frac{1}{2^s} + \frac{1}{4^s} - \frac{1}{5^s} + \frac{1}{7^s} - \frac{1}{8^s} + \cdots, \qquad \re s > 1.
\end{align*}
	For computational purposes it is more convenient to 
express this Dirichlet $\DirichletL$-series in terms of the Hurwitz zeta function
\begin{equation*}
	\zeta(s,a) \equiv \sum_{k=0}^\infty \frac{1}{\left( k + a \right)^s}, \quad \re s > 1,
\end{equation*} 
by means of
\begin{equation*} 
	\DirichletL_{-3}(s) = 3^{-s} \textstyle{\left[ \zeta\big(s,\frac{1}{3}\big) - \zeta\big(s,\frac{2}{3}\big) \right].} 
\end{equation*} 
	Using these representations we computed the graph of r.h.s.\Ref{eq:C.s}, see Fig.~\ref{fig:Cs}.
\begin{figure}[H]
\centering
\includegraphics[scale=.8]{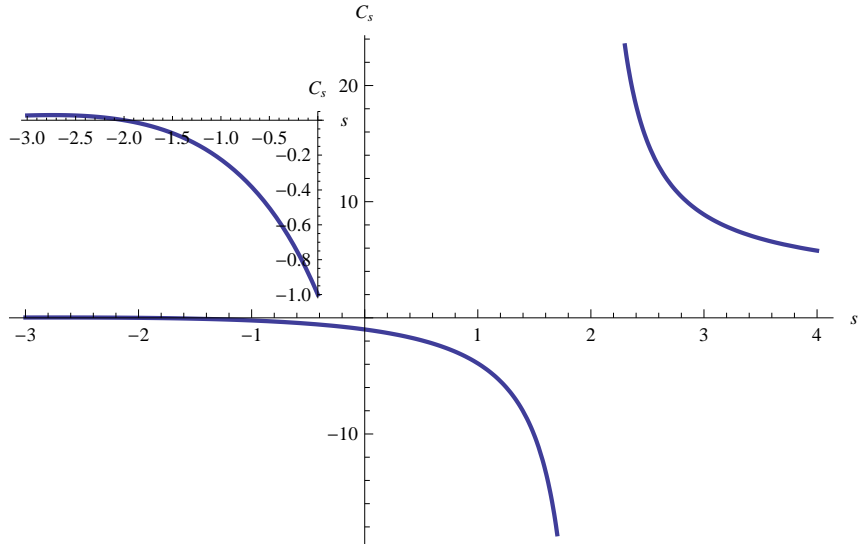}
\caption{\footnotesize{The graph of the 
			right-hand side of \eqref{eq:C.s}.}} \label{fig:Cs}
\end{figure}

	The fundamental conjecture for the asymptotic expansion of $\pzcE_s(\Sset^2, N )$, $s\in(2,4)$, as $N$ becomes large states that
\begin{conjecture}
	For $s\in(2,4)$ the asymptotic expansion of $\pzcE_s(\Sset^2, N )$ reads
\begin{equation*}
\pzcE_s(\Sset^2, N )  
\sim 
{\big( \sqrt{3} /8 \pi  \big)^{s/2} \zetafcn_{\Lambda}(s)} \, N^{1+s/2} + W_s^{} \, N^2 + o(N^2)
\quad \text{as $N \to \infty$.}
\end{equation*}
\end{conjecture}
	Note the appearance of (the analytically continued) $W_s^{}$ as coefficient of the $\mathcal{O}(N^2)$-term, which now is the
next-to-leading-order term.
\newpage

	The fundamental conjecture motivates the introduction of the following \emph{\recentered pair-energy}
\begin{equation}  \label{eq:tildeU.s}
\widetilde{U}_s( r )
\equiv 
s^{-1} \left( r^{-s} - W_s - \big(\sqrt{3} / 8\pi \big)^{s/2} \zetafcn_{\Lambda}(s) \, N^{s/2-1} \right), \quad s \in(2,4),
\end{equation}
the associated average \recentered pair-energy of a configuration $\omega_N^{}$,
\begin{equation}\label{eq:mean.tildeU.s.OF.omega}
\langle \widetilde{U}_s \rangle( \omega_N^{} ) 
\equiv 
\frac{2}{N \left( N - 1 \right)} \mathop{\sum \sum}_{1 \leq j < k \leq N} \widetilde{U}_s( \left| \qV_j - \qV_k \right| ),\quad s\in (2,4),
\end{equation}
and the \emph{minimal average \recentered Riesz pair-energy}, given by
\begin{equation}\label{eq:minimal.mean.tildeU.s.OF.omega}
\tilde{u}_s^{}( N ) 
\equiv 
\inf_{\omega_N \subset \Sset^2} \langle \widetilde{U}_s \rangle( \omega_N ),\quad s\in (2,4).
\end{equation}
	According to the fundamental conjecture, $\tilde{u}_s^{}(N)$ would be bounded above and
tend to $0$ when $N\to\infty$, for all $2 < s < 4$. 

	However, since $N\mapsto N^{s/2-1}$ is a concave, increasing function for $s\in(2,4)$, it is neither clear whether
$N\mapsto \tilde{u}_s^{}( N ) = {u}_s^{}( N ) - \big(\sqrt{3}/8\pi \big){}^{\!\!s/2} \zetafcn_{\Lambda}(s)N^{s/2-1}$ 
is increasing (whereas $N\mapsto {u}_s^{}( N )$ is), nor whether $N\mapsto \tilde{u}_s^{}( N )$ is concave whenever 
$N\!\mapsto\! {u}_s^{}( N )$ is.
	These are interesting open problems for future study.

\subsection{The singular case $s = 2$ \label{sec:sing}} 
	Kuijlaars and Saff~\cite{KuijlaarsSaff} showed that
\begin{equation*}
\lim_{N \to \infty} \frac{{E}_2( N )}{N^2 \log N} = \frac{1}{4}.
\end{equation*}
	A conjecture of Brauchart, Hardin and Saff \cite{BrHaSa2012} is that
\begin{equation*}
\textstyle {E}_2( N) \sim \frac{1}{4} \, N^2 \log N + C_{2} N^2 + o(N^2) \quad \text{as $N \to \infty$,}
\end{equation*}
where
\begin{equation*}
\hskip-5pt
\textstyle
C_{2} 
= \frac{1}{4} \left[ \gamma - \log ( 2 \sqrt{3} \pi ) \right] 
	+ \frac{\sqrt{3}}{4 \pi} \left[ \gamma_1\big(\frac{2}{3}\big) - \gamma_1\big(\frac{1}{3}\big) \right] 
= -0.08576841030090... < 0.
\end{equation*}
	This motivates the introduction of the following \emph{\recentered pair-energy},
\begin{equation}  \label{eq:tildeU.TWO}
\widetilde{U}_2( r )
\equiv 
\textstyle\frac{1}{2} \left( r^{-2} - C_2 - \frac{1}{4} \log N\right),
\end{equation}
its associated average \recentered pair-energy of a configuration $\omega_N^{}$,
\begin{equation}\label{eq:mean.tildeU.TWO.OF.omega}
\langle \widetilde{U}_2 \rangle( \omega_N^{} ) 
\equiv 
\frac{2}{N \left( N - 1 \right)} \mathop{\sum \sum}_{1 \leq j < k \leq N} \widetilde{U}_2( \left| \qV_j - \qV_k \right| ),
\end{equation}
and the \emph{minimal average \recentered Riesz pair-energy} at $s=2$, given by
\begin{equation}\label{eq:minimal.mean.tildeU.TWO.OF.omega}
\tilde{u}_2^{}( N ) 
\equiv 
\inf_{\omega_N \subset \Sset^2} \langle \widetilde{U}_2 \rangle( \omega_N ).
\end{equation}
	By the above conjecture, $\tilde{u}_2^{}(N)$ would be bounded and
tend to $0$ when $N\to\infty$. 

	However, since $N\mapsto \log N$ is a concave, increasing function, it is neither clear whether
$N\mapsto \tilde{u}_2^{}( N ) = {u}_2^{}( N )  - C_2 - \frac{1}{4} \log N$ is increasing (whereas $N\mapsto {u}_2^{}( N )$ is) 
nor whether $N\mapsto \tilde{u}_2^{}( N )$ is concave whenever $N\!\mapsto\! {u}_2^{}( N )$ is.
	Also these are interesting open problems for future study.

\vskip-.5truecm
	\section{Summary and Outlook} \label{sec:summary}

	In this paper we have inquired into the local concavity properties of the map $N\mapsto v_s^{}(N)$,
where $v_s^{}(N)$ is the minimal average standardized Riesz pair-energy for $N$-point configurations
on the unit 2-sphere $\mathbb{S}^{2}\subset\mathbb{R}^{3}$.
	By ``standardized'' Riesz pair-energy we mean $V_{s}(r)=s^{-1}\left(r^{-s}-1\right)$, with $s\in\mathbb{R}$,
where $r$ is the chordal distance between the points of the pair.
	The map $s\mapsto V_{s}(r)$ defines a real analytical family of increasing pair-energies; in particular, 
it includes the logarithmic interaction $-\ln r =\lim_{s\to 0}V_s(r)$.

	Given the limited amount of knowledge about true minimizers (see Appendix~\ref{sec:appdx.A}), we have studied 
mostly the $N$-dependence 
of putatively minimal average standardized Riesz pair-energies $v^{x}_{s}(N)$, obtained numerically in computer experiments.
	Our empirical findings indicate that for $s=-1$ the minimal average standardized Riesz pair-energy could 
be locally  strictly concave function of $N$, without any ``convex anomalies.''
	However, when $s\in\{0,1,2,3\}$ we have found that $N\mapsto v^{x}_{s}(N)$ is not strictly concave. 
	Based on our empirical findings we have conjectured that there exists an $s_*^{}\in(-1,0)$ such that $N\mapsto v^{x}_{s}(N)$ 
is locally strictly concave for all $s<s_*^{}$, while local strict concavity is violated at some $N$-values whenever $s\geq s_*^{}$.
	We presented some rigorous (but rough), and some quasi-rigorous (yet more promising) upper bounds on $s_*^{}$.
	Also, in Appendix \ref{sec:appdx.A}, we shall readily see analytically that $N\mapsto v_{-2}^{}(N)$ is strictly locally 
concave, and restricted to even numbers $N=2n$, the results of \cite{Bjorck} imply the strict local 
concavity of $N\mapsto v_s^{}(N)$ also for $s<-2$. 

	We also presented various rigorous bounds on the second discrete derivative, $\ddot{v}^{}_{s}(N)$,
of $N\mapsto v^{}_{s}(N)$, and we related concavity to an asymptotic analysis of the large-$N$ regime.
	Our control of $\ddot{v}_s^{}(N)$ is not good enough to prove strict local concavity for any $s$ other
than $s=-2$, yet we expect that our analysis will serve as a stepping stone along the way to a concavity proof for $s<s_*^{}$.
        In any event, our rigorous bounds can serve as test criteria for optimality.

	We emphasize that an a-priori knowledge of any concavity properties of the
actual maps $N\mapsto v^{}_{s}(N)$ will furnish valuable test criteria for the accuracy of empirical maps $N\mapsto v^{x}_{s}(N)$. 
	In particular, given the practical importance of the $s=-1$ problem in numerical integration schemes as
mentioned in Appendix \ref{sec:appdx.A}, it is of special interest to prove, or disprove, the empirically suggested local strict concavity 
of $N\mapsto v^{}_{-1}(N)$. 
	Since the optimizers of the $s=-1$ Fekete problem are related to spherical digital nets, it is of interest to raise our
concavity questions also for these average standardized Riesz pair-energies computed for such nets; see our Appendix~\ref{sec:appdx.D}.

	For each studied $s$-value, the $N$-values at which the map $N\mapsto v^{x}_{s}(N)$ is strictly convex
were collected into a set $\cC_{+}^{x}(s)$.  
	Inspection of the empirical map $s\mapsto \cC_{+}^{x}(s)$ suggested to us 
the conjecture that the actual map $s\mapsto \cC_{+}^{}(s)$ is set-theoretically monotonic increasing.
	Inspired by this conjecture we suspected, and then verified (by finding lower-energy configuations), that the $N=177$ and $N=197$ 
data points in the computer-experimental tables of putatively minimal Riesz pair-energies for $s\in\{2,3\}$ in \cite{Ca2009} are non-optimal.
	This makes it plain that a-priori knowledge of any monotonicity property of the map $s\mapsto \cC_{+}^{}(s)$ will also furnish 
valuable test criteria for the accuracy of empirical maps $N\mapsto v^{x}_{s}(N)$. 

        We have also discovered yet another empirical monotonicity: the percentage of odd numbers in $\cC_{+}^{x}(s)$ increases monotonically 
with $s\in\{0,1,2,3\}$. 
        Based on this finding it is reasonable to conjecture that the percentage of odd numbers in $\cC_{+}^{}(s)$ increases monotonically
with $s>s_*^{}$, if such $s_*^{}$ exists.

	We have not been able to detect any algebraic generating rule of $\cC_{+}^{x}(s)$ for any $s\in\{0,1,2,3\}$,
and with increasing $s$ it seems increasingly unlikely (another monotonicity property!) to detect any such rule, because
the sets $\cC_{+}^{x}(s)$ appear more and more random as $s$ increases from $0$ to $3$.
	On the other hand, the set $\cC_{+}^{x}(0)$ exhibits some intriguing quasi-regular patterns which reminded us\footnote{By 
          this we only mean a vague qualitative reminiscence. 
        Of course, Thomson \cite{Thomson} may have hoped to find the actual quantitative pattern of the periodic table of the 
        chemists; for a most recent inquiry in this spirit, see \cite{LaFave}.}
of the periodic table of the chemists, or the ``magic'' numbers in nuclear physics; incidentally, note that
the $s=0$ problem can be viewed as the electric ground state energy problem for a system of $N$ ``two-dimensional'' 
point charges on $\mathbb{S}^{2}$.
	Thus we decided to call the $N$-values in $\cC_{+}^{x}(0)$ the ``Magic Numbers of Smale's $7$th Problem.'' 
	We have speculated that those ``magic'' numbers are perhaps associated with ``optimally symmetric'' endpoints of 
families of more-and-more symmetric configurations, based on our observation that the first few 
configurations associated with $\cC_{+}^{x}(0)$ are highly symmetric.
        In order to switch from a ``magic-$N$'' optimal configuration to an optimal $N+1$ configuration by adding a particle 
would then require an unusually large amount of energy to break up the ``optimally symmetric (magic-$N$)'' configuration ---
hence the convexity. 

	We hope that our paper triggers further research into the concavity properties of the minimial average standardized 
Riesz pair-energies on $\Sset^2$, in particular the structure of its sets of convexity defects as functions of $s$. 
	As a preliminary guide into such future inquiries, we have formulated a list of 14 interesting questions, yet surely there
are many more!
	It would also be interesting to try to answer analogues of our questions formulated
for the $s$-Fekete problem on other compact manifolds $\cM$.
	In particular, we expect that the analogue of our list of questions Q~1 -- Q14 for 
$\cM=\Sset^1$ can be answered \emph{in complete detail}; cf. our Remark~\ref{rem:S1concavity}. 

	Furthermore, even though the paper \cite{MeKnSm1977} contains many numerically computed bifurcation diagrams, we noticed 
the absence of any proper bifurcation analyses in the literature on the $s$-Fekete problems.
	To assist our inquiries we have begun such a bifurcation analysis, so far numerically, for a few small $N$-values.
	In the process we discovered a previously undetected bifurcation in the $N=7$-point problem at $s=0$,         
where the local minimality is exchanged between the pentagonal bi-pyramid and a $C_2(1^12^3)(f=5)$ configuration.
        We have \emph{not} confirmed that this family is truly globally minimizing for the interval $-2<s<0$, \emph{neither}
did we confirm that the pentagonal bi-pyramid is globally minimizing for $0<s<1$; yet we strongly suspect that all this is true.
	We are planning a more detailed bifurcation analysis, with special attention given to the $7$-point problem, 
to be reported on in a future publication.
        We emphasize that a bifurcation analysis will be practically feasible only for moderate $N$ values because the
eigenvalue problems involve nontrivial $2N\times2N$ matrices which will become difficult to handle when $N$ gets too large.

        We mentioned that  computer-experimental evidence suggests that the number of non-globally 
minimizing configurations (modulo rotations on $\Sset^2$), is growing exponentially with $N$; cf. \cite{ErberHockneyTWO}.
        It would be good to have a rigorous proof, together with a determination of the growth rate.
        But to the best of our knowledge, there is no rigorous estimate even of their number being finite!
        This reminds one of Smale's 6-th problem, viz. the celestial-mechanical counterpart of such a  finiteness proof.

        An exponential growth rate of the number of locally minimizing configurations
(or perhaps of the number of equilibrium configurations) is reminiscent of ``the \emph{complexity} of the energy landscape,'' 
see \cite{WalesBOOK}.
        As far as we know, not much is known about the Riesz $s$-energy landscape for $N$-point configurations on $\Sset^2$.
        With the help of catalogs of non-globally minimizing configurations and their energies it should be 
feasible to determine the experimental number counts of the local minimizers \emph{below a certain energy $E$};
see \cite{CalefETal2013} for a most recent study and additional references.

\noindent
{\bf Acknowledgement:} J.B. and M.K thank Ed Saff for his role as ``matchmaker'' that started our collaboration
at the conference OPTIMAL 2010 at Vanderbilt University, Nashville, Tennessee, and for interesting comments.
  R.N. and M.K. thank Michael Kastner for providing the opportunity to get our collaboration started at the 2011 STIAS workshop on 
``Equilibrium and Equilibration'' in Stellenbosch, South Africa, and Lapo Casetti for his support of our collaboration and his 
interesting comments on energy landscapes.
	We would like to thank Rob Womersley (UNSW) for numerical data and helpful discussions regarding the numerics.
	J.B. greatfully acknowledges partial support by an APART-Fellowship of the Austrian Academy of Sciences and the 
hospitality of the School of Mathematics and Statistics at UNSW and the support of the Australian Research Council;
	R.N. and M.K. greatfully acknowledge partial financial support by the NSF through grant DMS 0807705, and by 
the INFN, Sezione di Firenze.
	Lastly, we thank the referee for the careful reading of the original manuscript and for constructive criticism. 


\newpage

\begin{appendices}
\section{}\label{sec:appdx.A}
\subsection*{Optimal Riesz energy configurations on $\Sset^2$: A brief survey} 
\noindent
	The problem to determine $v_s^{}(N)$ together with the minimizing configuration(s) $\omega_N^s$ 
has been solved completely for all $N\geq 2$ only at the distinguished value $s=-2$, by explicit calculation.
        Namely, $s=-2$ yields the energy law for the completely integrable Newtonian $N$-body problem with 
repulsive harmonic forces.
	Any $N$-point configuration satisfying $\sum_{i=1}^N \qV_i = \mathbf{0}$ is a 
minimizing configuration of $\langle V_s\rangle(\omega_N)$, and only such are. 
	The minimal energy~reads

\vskip-.4truecm
\begin{equation} 
	v_{-2}(N)
=\label{vSUBminusTWO}
	-\frac{1}{2}\frac{N+1}{N-1}.
\end{equation}		

	The (presumably) next-simplest parameter regime is $s<-2$.
	Here one is confronted with the possibly startling observation that for large $N$ the $N$-tuple 
Fekete points accumulate around two opposite points, and the localization sharpens as $N$ is getting larger; 
this is a consequence of Theorem 7 in \cite{Bjorck}. 
	In particular, it follows right away from Theorem 7 in \cite{Bjorck} that for even $N$ the infimum 
$v_s^{}(N),\, s<-2$, is achieved\footnote{By Theorem 7 of \cite{Bjorck}, the infimum is 
		\emph{not} achieved by a proper $N$-point configuration.}
if and only if half of the particles each are placed at two antipodal points,\footnote{This was already noted by
	          Rachmanov, Saff, and Zhou \cite{RSZa}.}
yielding
\begin{equation} 
	v_s^{}(N)
=\label{vsFORsBELOWminusTWO}
-\frac{1}{|s|}\,\frac{\left( 2^{|s|-1}-1 \right) N+1}{N-1},\qquad s< -2,\qquad N=2n,
\end{equation}
which converges to $v_{-2}(N)$ when taking the limit $s\uparrow -2$ of \Ref{vsFORsBELOWminusTWO}.
        When $N$ is odd the situation is already more tricky, and more interesting!
	For instance, for the smallest allowed odd $N=3$ it is suggestive to conjecture that the
minimizing configuration consists of the corners of an equilateral triangle in an arbitrary equatorial 
plane; yet comparison with an antipodal ``configuration'' (arrangement) with two 
labeled points in the North and one in the South Pole 
reveals that the equilateral configuration yields a lower average standardized Riesz pair-energy only for 
$s_3 < s < -2$, where $s_3 \equiv \ln(4/9)/\ln(4/3)$, while for $s < s_3$
the antipodal arrangement yields the lower average standardized Riesz pair-energy; in this case one can easily show
rigorously that the antipodal arrangement is in fact optimal: namely, the equilateral triangle and
the antipodal arrangement are the only equilibrium arrangements of 3 labeled points.
	When comparing the average standardized Riesz pair-energy for antipodal and equilateral arrangements for other odd $N$, 
this changeover happens only if $N$ is a multiple of $3$.
	The critical $s_{3(2n-1)}$ tends monotonically to $-2$ as $N =3(2n-1)\to \infty$.
	Of course, this does not prove that either arrangement is optimal in the respective range of~$s$. 
	To the best of our knowledge, the optimal arrangement of odd-$N$ points as a function of $s<-2$ 
is far from being settled.

	The concentration of the minimizing ``$N$-point configuration'' for $s< -2$
at a few distinct points indicates that the optimization problem  is incorrectly posed 
in the set of proper $N$-point configurations.
	(The deeper reason is that the Riesz pair-energy ceases to be positive definite in the sense of 
Schoenberg~\cite{Sch1938} for $s < -2$.)
	Interestingly, the sum of distance problem for $s < -2$ plays a central role in the theory of 
Quasi-Monte Carlo integration schemes\footnote{An important special case of Quasi-Monte Carlo schemes 
  are the so-called $t$-designs, which can be characterized by polynomial energy functionals \cite{SlWo2009}.} 
for functions in smooth enough function spaces over $\Sset^2$;
	we refer the interested reader to \cite{BrSaSlWo14,BrDi2013a,BrDi2013b} and papers cited therein.

        When $s> -2$ the problem becomes drastically more complicated.
	One needs to distinguish the cases $-2<s<2$, $s=2$, $s>2$, and the limit ${s\to\infty}$. 

The interval $-2<s<2$ is known as the \emph{potential-theoretical regime},
since concepts and methods of potential theory can be applied to study both the discrete and the continuous (i.e. $N\to\infty$) optimization 
problems.\footnote{The potential function $\qV \mapsto V_{s}( \left| \qV \right| )$ is \emph{strictly superharmonic} 
		(i.e. $\Delta V_{s}( \left| \qV \right| ) < 0$ in $\dot\Rset^3$) for $s\in(-2,1)$, \emph{harmonic} 
  ($\Delta V_{s}( \left| \qV \right| ) = 0$) for $s = 1$, and \emph{strictly subharmonic} 
  ($\Delta V_{s}( \left| \qV \right| ) > 0$) for $s\in(1,2)$; here, $\Delta$ is the Laplacian in the ambient 
		space $\Rset^3$, and $\dot{\Rset}^3$ is $\Rset^3$ with its origin removed.
			One consequence is that $N$-point configurations with minimal average standardized Riesz pair-energy 
		in the closed unit ball in $\Rset^3$ live on its  boundary $\Sset^2$ 
		in the superharmonic case, but extend ``into the solid'' in the subharmonic case as 
                some points need to move into the volume to lower the energy (``charge injection'') (cf. \cite{Landkof,Berezin}).} 
Within this regime the integer values $s=-1$, $s=0$, and $s=1$ are of particular interest.
	When $s=-1$ the minimal average standardized Riesz pair-energy problem is equivalent to the maximal 
average pairwise chordal distance problem; see \cite{FejToth,Stol,Beck}. 
        In \cite{BrDi2013b} it is shown that maximum-sum-of-distance configurations are ideal integration nodes 
for a certain optimal-order Quasi-Monte Carlo integration scheme on $\Sset^2$; we will come back to this in Appendix~\ref{sec:appdx.D}.
        The case ``$s=0$,'' i.e. the limit $s\to 0$, which yields the logarithmic pair-energy~\Ref{limiteNULL}
(also known as the Coulomb energy for a pair of ``two-dimensional unit point charges'' on $\Sset^2$,
respectively the Kirchhoff energy \cite{KirchhoffBOOK} of a pair of unit point vortices on $\Sset^2$),
occurs in a stunning variety of problems (on $\Sset^2$ and other manifolds) in the sciences and mathematics;
see, e.g.~\cite{ForresterJancoviciMadore,ShubSmale,ChaKieCMP,SaffTotik,BCNTlett,KBCDKNV,NeCh2009,ForresterBOOK,KieWang}.
	Originally Smale's 7th problem for the 21st century \cite{Smale} was formulated for the logarithmic energy, see below.
	 Lastly, the value $s=1$ yields the Coulomb pair-energy of ``three-dimensional unit point charges''
associated with the so-called \emph{Thomson problem} (see \cite{Thomson,Whyte,ErberHockneyTWO,AWRTSDW,PerezGetal,BCM,LaFave}).

	Amongst the values $s\geq 2$, the borderline value $s=2$ is special in the sense that the finite-$N$ behavior 
is qualitatively different from both, the regime $-2<s<2$, and the regime $s>2$.
	Yet it can be understood by considering a certain limit process $s \to 2$; cf. \cite{CaHa2009}
for the limit process $s\to d$ in analogous optimization problems formulated on $d$-dimensional manifolds.
	The Riesz pair interaction for $s=2$, in physics considered as correction term to Newton's gravity \cite{Manev},
is also special in the sense that it yields a Newtonian $N$-body problem in $\Rset^3$ with additional isolating integrals of 
motion \cite{Boby,LBLB,CalogeroLeyvraz}, besides those associated with Galilei invariance.
	Restricted to $\Rset$ the motion is even completely integrable for all $N$ \cite{Calogero,Moser}. 

	The large-$s$ behavior of $\langle V_s\rangle (\omega_N)$ ($N$ fixed) is intimately connected 
with the classical \emph{Tammes's problem} (\cite{Tammes}) or \emph{hard sphere (best-packing) problem} (cf. \cite{CoSl99}); 
that is, to find a configuration of $N$ points on the sphere with the minimal pairwise (chordal) distance 
between the points being as large as possible.\footnote{From \cite{BoHaSa07} it readily follows that a 
		certain limit of the leading coefficient in the asymptotic expansion of $v_s (N)$ for large $N$ (and the $d$-sphere) 
		is closely related to the \emph{largest sphere packing density in $\Rset^{\mathrm{d}}$}. 
		Only the densities for $d = 1$, $2$ and $3$ are known, and only quite recently Hales \cite{Ha05} could settle the last 
		case by proving the famous \emph{Kepler Conjecture}, which states that no packing of congruent balls in Euclidean space 
                has density greater than the density of the face-centered cubic packing 
		(which is identical to the density of the hexagonal close packing).}  
	It is not too difficult to see (cf. Appendix~\ref{sec:appdx.B}) that for any $N$-point configuration 
$\omega_N = \{ \qV_1, \dots, \qV_N \} \subset \Sset^2$ the following limit relation holds:
\begin{equation} \label{eq:s.to.infty.A}
\lim_{s \to \infty} \left[ \langle V_s\rangle (\omega_N) +{\textstyle\frac{1}{s}}\right]^{-1/s} 
= \min_{1 \leq i < j \leq N} \left| \qV_i - \qV_j \right|
\equiv \varrho( \omega_N ).
\end{equation}
	Moreover, whenever a family of minimizing configurations $\omega_N^s$ converges to a limit
configuration $\omega_N^\infty$ one has the relation (cf. Appendix~\ref{sec:appdx.B})
\begin{equation} \label{eq:s.to.infty.B}
\lim_{s \to \infty} \left[ v_s (N) +{\textstyle\frac{1}{s}}\right]^{-1/s} = \varrho( \omega_N^\infty) \equiv \rho(N),
\end{equation}
where $\rho(N)$ is the \emph{best-packing (chordal) distance},  which 
maximizes the \emph{least distance} $\varrho(\omega_N)$ among all $N$-point configurations on $\Sset^2$,
and $\omega_N^\infty$ is the \emph{best-packing configuration}.\footnote{See \cite{ClKe1986,MeKnSm1977} for tables of 
		of $\omega_N^\infty$ and numerical values of $\rho(N)$.}
	The best-packing distance $\rho(N)$ is only known for $N = 2,\ 3,\ \dots,\ 12$ and $24$ 
(cf.\! \cite{FejesTothL,SchvdW1951,Rob1961,Bor1983,Da1986}).\footnote{Very recently, a proof for $N=13$ has been 
		proposed in \cite{MuTa2012arXiv}.} 

	Another way of obtaining a nontrivial limit problem is to let $s\to\infty$ in $V_s(r)$, which gives
\begin{equation} \label{limiteINFINITY}
	V_{\infty}(r) \equiv \lim_{s\rightarrow \infty} V_{s}(r)
= \begin{cases}
      \infty & \text{if $r < 1$,} \\
      0 & \text{if $r \geq 1$.}
  \end{cases}
\end{equation}
	In that case $v_\infty(N) = 0$ for $N \leq N_*$, while $v_\infty(N) = \infty$ for $N > N_*$, 
viz. $N_*$ is the maximum number of non-overlapping calottes with \emph{spherical radius} $\pi/6$
which can be placed on the unit sphere.
        By picking any ball, $B$, from a hexagonal close packing (hcp) of $\Rset^3$ with unit balls, 
then projecting its 12 nearest neighbors (unit balls) radially onto the surface of the central ball $B$, one sees that $N_*\geq 12$.
	And dividing the surface area  of the unit sphere, $4\pi$, by the area of a calotte, $\pi(2-\sqrt{3})$, yields $\approx 14.92820323$, 
giving the upper bound $N_*< 15$.
	But how large is $N_*$, exactly?

Interestingly, the sharp value for $N_*$ is found by studying the related Tammes problem.
	For $2 \leq N \leq 12$ one has\footnote{In particular, $\rho(12)\approx 1.051462225$
		implies that the 12 calotte arrangement $\omega_{hcp}$ obtained from the hexagonal close packing of
		$\Rset^3$, which has $\varrho(\omega_{hcp})=1$, is not the optimizer of the Tammes problem with $N=12$, 
		which is $\omega_{12}^\infty$: the regular icosahedron.}
$\rho(N) > 1$. 
	L. Fejes T{\'o}th's famous inequality (\cite{FejesTothL}),
\begin{equation*}
\left[ \rho(N) \right]^2 \leq 4 - \Bigl[ \cosec\Big( \frac{\pi}{6} \frac{N}{N-2} \Big) \Bigr]^2,
\end{equation*}
where equality holds only for $N = 3$, $4$, $6$ and $12$, gives $\rho(N) < 1$ for $N \geq 14$. 
	From \cite{BaVa2008} follows $\rho(13) < 1$. 
	Hence, $N_* = 12$.

        To our best knowledge, the following point sets are the only ones for which 
one can rigorously prove that they have minimal average standardized Riesz pair-energy for \emph{all} $s > -2$. 
	One can \emph{easily} characterize the minimizing configuration explicitly only when $N = 2$ or $3$ 
(as the antipodal and equilateral configuration, respectively).
	The minimizing configuration has been characterized explicitly also
for $N = 4$, $6$, and $12$ as the vertices of Platonic solids\footnote{Surprisingly, perhaps, 
	the vertices of the Platonic cube ($N = 8$) have a higher average 
	pair-energy than the square-antiprism derived from the cube by twisting (angle of 45 degrees) and 
	squeezing together two opposite faces of the cube.
	Similarly, the dodecahedron (N=20) is not a minimizing configuration either, for any $s>-2$.}
(tetrahedron, octahedron, and icosahedron), which are known to be {\em universally optimal} (see \cite{CoKu2007}); 
such configurations minimize the potential energy of \emph{completely monotonic} pair-energy functions.
	The standardized Riesz pair-energies for $s > -2$ (including the logarithmic pair-energy at $s=0$)
fall into this category.
	The listed configurations for $N = 2$, $3$, $4$, $6$, and $12$ exhaust the 
possibilities for universally optimal configurations on $\Sset^2$; cf. \cite{Le1957,CoKu2007}. 

	The surprisingly difficult task of finding a proof of minimality can, perhaps, be best illustrated with 
the only partly resolved \emph{five point problem on $\Sset^2$}. 
	It is clear from \cite[Prop.~14]{CoKu2007} that there is no universally optimal $5$-point configuration 
on $\Sset^2$. 
	Indeed, computational optimization reveals that the minimal-energy arrangement of five labeled points on 
$\Sset^2$ changes many times as $s$ varies over the real line.\footnote{The numerical study for $s\leq -2$ is our own.
	For $1\leq s\leq 400$, and $s\to\infty$, cf. \cite{MeKnSm1977}.} 
	Thus, for $s \leq -2.368335...$ an antipodal arrangement with two labeled points in the South, and three 
in the North Pole (say) is the optimizer; at $s=-2.368335...$ a crossover takes place, and for 
$-2.368335... \leq s \leq  -2$ the energy-minimizing arrangement of five labeled points is an isosceles 
triangle on a great circle, with one point in the North Pole and two labeled points each in the other two corners,
with (numerically) optimized height.
	At $s=-2$ the isosceles arrangement bifurcates off of a continuous family of rectangular pyramids with height
$h=5/4$, all of which have the same energy $-3/4$ at $s=-2$, and of which the isosceles arrangement is the degenerate 
limit.
	At $s=-2$ also another crossover happens, and for $-2\leq s\leq 15.048077392\dots$ the regular triangular 
bi-pyramid is the putative energy-minimizing configuration. 
	At $s = 15.048077392\dots$, yet another crossover happens, at which the triangular bi-pyramid and 
a square pyramid with height $h\approx 1.1385$ have the same average (standardized) Riesz pair-energy.
	A square pyramid with (numerically) optimized height\footnote{Curiously,
	for $-2\leq s\leq 0$ the optimal height of the square-pyramidal configuration is constant, equal to $5/4$. 
	Only for $s>0$ does the optimized height depend on $s$.}
as function of $s$ appears to have lower (standardized) Riesz pair-energy for $s\in[15.04807...,\infty)$.
	Lastly, it is well-known that the triangular bi-pyramid and the square pyramid with 
height $1$ both are particular best-packing configurations, with $\varrho(\omega_5^\infty)=\sqrt{2}$,
so that ``at $s=\infty$'' an ``asymptotic crossover, or degenerate bifurcation,'' happens.

	How much of this has been proved rigorously? 
	By traditional methods (see \cite{DrLeTo2002}) it can be shown
that the triangular bi-pyramid consisting of two antipodal points 
at, say, the North and the South Pole, and three equally spaced points on the Equator, 
is the unique (up to orthogonal transformation) minimizer of the logarithmic average pair-energy.
	The proof that the same configuration maximizes the sum of distances (that is: assumes $v_{-1}(5)$) 
is computer-aided, exploiting interval methods and related techniques (see \cite{HouSh2011}).
        In \cite{Sch2010arXiv} a computer-aided approach is proposed to show optimality of the triangular 
		bi-pyramid for $s = 1$ and $s = 2$.
        The optimality of both the triangular bi-pyramid and the family of rectangular pyramids with height $5/4$ 
at $s = -2$ can be shown with elementary techniques.
        The rest of the $s$-parameter regime still awaits its rigorous treatment.\footnote{The 
five point problem on the sphere can be also studied as (unconstrained) external field problem in the plane 
(J.B., manuscript in preparation).}

	Numerical results in \cite{MeKnSm1977}, carried out with varying $s\in[1,400]$ for
$N\in\{2,...,16\}$ fixed, suggest that also for other values of $N\not\in\{2,3,4,6,12\}$ the minimizing 
configuration $\omega_N^s$ may generally change as $s$ passes through critical values, and their number 
seem to depend on $N$.
	In particular, for $N=7$ there seem to be three(!) critical $s$-values in $[1,6]$ at which the minimizing 
configuration changes, and presumably a few more when $s<1$, cf.~\cite{BermanHanes} for $s=-1$.
	The general dependence on $s$ of the optimal $N$-point configurations $\omega_N^s$ is 
one of the intriguing features of this minimization problem.
        All the same, it makes it plain why the rigorous determination of the optimizers is a highly nontrivial task 
even for moderate $N$-values other than the special ones $2,3,4,6,12$, becoming hopelessly complicated when $N$ increases.

	Yet, the large-$N$ asymptotics of the minimal average standardized Riesz pair-energy $v_s^{}(N)$ can be
determined without seeking the exact Fekete points, see \cite{RSZa}.
        In particular, $\lim_{N\to\infty}v_s^{}(N)$ is for all $s$ determined by the variational principle\footnote{Our usage here of 
		both $\pV$ and $\qV$ as points in space (i.e., on $\Sset^2$) should not be confused with the usage in 
		Hamiltonian dynamics of $(\pV,\qV)$ as a pair of canonical variables.}
(see \cite{Szego1924,Bjorck,Landkof,KieSpoCMP})
\begin{equation}\label{eq:limit.VP}
\lim_{N\to\infty}v_s^{}(N)
=
\inf_{\mu\in \Psp(\Sset^2)} \iint_{\!\!\!\Sset^2\times\Sset^2} V_s(|\pV-\qV|)\mu(d\pV)\mu(d\qV);
\end{equation}
here, $\Psp(\Sset^2)$ is the set of all Borel probability measures supported on $\Sset^2$. 
        For $s\leq -2$ the minimizer is not unique, but all minimizers are known; in particular, 
for $s<-2$ the minimizer, after factoring out $SO(3)$, is a symmetric measure which is concentrated 
on two antipodal points, see \cite{Bjorck}.
	From classical potential theory
(cf. \cite{Landkof} for $s\in [0,2)$ and \cite{Bjorck} for $s\in(-2,0)$)\footnote{The investigation of the 
	``energy integral'' $\int\!\!\int_{\Sset^2\times\Sset^2} | \pV - \qV|^\lambda \mu(d\pV)\mu(d\qV)$, 
	$\lambda$ real, can be traced back to \cite{PoSz1931}.}
it is well-known that the uniform normalized (Lebesgue) surface area measure on $\Sset^2$, denoted by $\sigma$, 
uniquely minimizes the right-hand side in (\ref{eq:limit.VP}) for $-2 < s < 2$.
	For $s\geq 2$ the l.h.s. and r.h.s. of \Ref{eq:limit.VP} are both $\infty$; in this case the rate of 
divergence of $v_s^{}(N)$ can be determined.
	It ``suffices'' to know that for large $N$ the Voronoi cells around the charges are mostly
hexagons of a certain size; see \cite{SaffKuijlaars} for an enlightening discussion.\footnote{This also 
works for $s<2$; cf. \cite{Berezin}, where a ``semi-continuous approach'' is used to model the large-$N$ behavior.}
        The picture one should have in mind, when $N$ is large, is a vast sea of hexagonal Voronoi cells around most 
of the points. 
        Thus, the dual structure of the Voronoi cell decomposition, the Delaunay triangulation, is a network of mostly 
six-fold coordinated sites. 
The reason for the qualification ``mostly'' lies in the topology of the sphere, which gives rise to 
{\em geometric frustration} (see \cite{SaMo1999} for a thorough exposition of this notion). 
	Certain points ``pick up'' a \emph{topological charge} that measures the departure from the ideal 
coordination number, $6$, of the planar triangular lattice. 
	The celebrated Euler theorem of topology yields that the total topological charge on $\Sset^2$ is always 12. 
	This accounts, for example,  for the appearance of 12 (isolated) pentagons in the common soccer ball design. 
	For large $N$ one observes ``scars'' emerging from these isolated centers that attract pentagon-heptagon pairs 
(having total topological charge $0$). 
       Scars and other topology-induced defects of the hexagonal lattice become important when pushing the
asymptotic analysis to higher order, and are not well understood. 
       For instance, to the best of our knowledge it is an unresolved question if there are $n$-gon Voronoi cells with $n \geq 8$ in a 
minimizing configuration. 
	See \cite{BoGi2009} for an approach using elastic continuum formalism.

	The truly hard regime is the vast intermediate range of $N$ which are generically too large to allow 
for an explicit determination of the minimizing configuration, but not large enough for the asymptotic formulas
to yield sufficiently accurate results.

        This concludes our brief survey of this fascinating field.
	Further information can be found in the survey articles \cite{ErberHockneyTWO}, \cite{SaffKuijlaars}, 
\cite{HardinSaffONE}, and on the websites \cite{BCM} and \cite{Womersley}.
See also the delightful article \cite{AtiyahSutcliffe} where, based on numerical evidence, the first dozen minimizers are discussed
mostly for Thomson's problem ($s=1$).
\newpage

\section{} \label{sec:appdx.B}
In this appendix we supply the proofs of relations \eqref{eq:s.to.infty.A} and \eqref{eq:s.to.infty.B} which control the 
limit $s\to\infty$.

\subsection*{Proof of Relation \eqref{eq:s.to.infty.A}} 
\label{sec:appdx.s.to.infty.A}

	Suppose $\omega_N = \{ \qV_1, \dots, \qV_N \} \subset \Sset^2$ is a fixed $N$-point set, 
with separation distance $\varrho( \omega_N ) = \min_{1 \leq i < j \leq N} | \qV_i - \qV_j |$. 
	Then using that the function $f(x) \equiv x^{-1/s}$ is strictly decreasing  for $s > 0$,
we find that
\begin{align*}
\left[ \langle V_s\rangle (\omega_N) +{\frac{1}{s}}\right]^{-1/s} 
&= \left[ \frac{1}{s} \frac{2}{N(N-1)} \mathop{\sum\sum}_{1 \leq i < j \leq N}
 \frac{1}{\left| \qV_i - \qV_j \right|^{s}} \right]^{-1/s} \\
&= \varrho( \omega_N ) 
\left[ \frac{1}{s} \frac{2}{N(N-1)}\mathop{\sum\sum}_{1 \leq i < j \leq N}
\left(\frac{\varrho(\omega_N)}{\left|\qV_i - \qV_j\right|}\right)^{s}\right]^{-1/s}\\
&\geq \varrho( \omega_N ) \left( \frac{1}{s} \right)^{-1/s} \\
&\geq \varrho( \omega_N ).
\end{align*}
	On the other hand, retaining only one of the least distance pairs in the double sum yields
\begin{align*}
\left[ \langle V_s\rangle (\omega_N) +{\frac{1}{s}}\right]^{-1/s} 
&\leq 
 \left( \frac{1}{s} \frac{2}{N(N-2)} \frac{1}{\varrho(\omega_N)^s}\right)^{-1/s}\\
&=
\varrho( \omega_N ) \left( \frac{1}{s} \right)^{-1/s} \left( \frac{2}{N(N-2)} \right)^{-1/s} \\
& \to \varrho( \omega_N )\qquad \mbox{as}\quad s\to\infty.
\end{align*}
This completes the proof of \eqref{eq:s.to.infty.A}.

\subsection*{Proof of Relation \eqref{eq:s.to.infty.B}}
\label{sec:appdx.s.to.infty.B}

	Let $\omega_N^s = \{ \qV_1^s, \dots, \qV_N^s \} \subset \Sset^2$ denote a minimizing $N$-point set, and
suppose $\omega_N^\infty$ is a best-packing configuration with $\varrho( \omega_N^\infty ) = \rho(N)$. 

	Then, first of all,
\begin{equation*} 
\langle V_s \rangle(\omega_N^s ) 
\leq 
\langle V_s \rangle( \omega_N^\infty );
\end{equation*}
but $\langle V_s \rangle(\omega_N^s ) = v_s^{}(N)$, and so, by \eqref{eq:s.to.infty.A}, we have
\begin{equation} \label{eq:s.to.infty.Bb}
\liminf_{s \to \infty} \left[ v_s (N) +{\textstyle\frac{1}{s}}\right]^{-1/s} \geq \varrho( \omega_N^\infty) \equiv \rho(N).
\end{equation}

	On the other hand, using that $| \qV_i^s - \qV_j^s | = \varrho( \omega_N^s )$ for at least one pair $(i,j)$,
and furthermore that $\varrho( \omega_N^s ) \leq \rho(N)$, we obtain
\begin{align*}
\langle V_s \rangle(\omega_N^s ) + {\frac{1}{s}}
& = 
{\frac{1}{s} \frac{2}{N(N-1)}} \mathop{\sum\sum}_{1 \leq i < j \leq N} \left|\qV_i^s -\qV_j^s\right|^{-s} \\
&= \left[ \rho(N) \right]^{-s} \frac{1}{s} \frac{2}{N(N-1)}
 \mathop{\sum\sum}_{1 \leq i < j \leq N} \left[ \frac{\rho(N)}{\left| \qV_i^s - \qV_j^s \right|}\right]^{s}  \\
& \geq \left[ \rho(N) \right]^{-s} \frac{1}{s} \frac{2}{N(N-1)} > 0.
\end{align*}
	Hence, 
\begin{align*}
\left[ \langle V_s \rangle( \omega_N^s ) + {\frac{1}{s}}\right]^{-1/s} 
& \leq 
\rho(N) s^{1/s} \left( N(N-1)/2 \right)^{1/s} \\
& \to \rho(N)\qquad \mbox{as}\quad s\to\infty;
\end{align*}
but again, $\langle V_s \rangle(\omega_N^s ) = v_s^{}(N)$, and so we get
\begin{equation} \label{eq:s.to.infty.Bc}
\limsup_{s \to \infty} \left[ v_s (N) +{\textstyle\frac{1}{s}}\right]^{-1/s} \leq \rho(N).
\end{equation}
	By \eqref{eq:s.to.infty.Bb} and \eqref{eq:s.to.infty.Bc}
\begin{equation*}
\lim_{s \to \infty} \left[ v_s (N) +{\textstyle\frac{1}{s}}\right]^{-1/s} = \rho(N).
\end{equation*}
	This completes the proof of \eqref{eq:s.to.infty.B}.
\newpage

\section{} \label{sec:appdx.C}

We now prove the strict monotonic increase of $s\mapsto v_s^{}(N)$.

\subsection*{Proof of Relation \eqref{eq:vOfsISincreasing}} \label{sec:sTOvISmonotoneUP}
 
We begin with the observation that the map $s\mapsto V_s^{}(r)$ is monotone increasing for all
$r\in(0,2]$, in fact strictly so except when $r = 1$. 
To see this, take the first partial derivative of $V_s^{}(r)$ w.r.t. $s$ to get 
\begin{equation}\label{eq:sDERofVsOFr}
\partial_s^{} V_s^{}(r) 
=    
 -s^{-2}\left(r^{-s}-1\right) -s^{-1} r^{-s} \ln r.
\end{equation}
Clearly, r.h.s.(\ref{eq:sDERofVsOFr}) $=0$ if $r=1$. 
We now show that r.h.s.(\ref{eq:sDERofVsOFr}) $>0$ if $r\neq 1$. 

To this end, now take  the first partial derivative of $\partial_s^{} V_s^{}(r)$ w.r.t. $r$ to get
\begin{equation}\label{eq:srDERofVsOFr}
\partial_{r}^{} \partial_{s}^{} V_s^{}(r) 
= 
r^{-s-1} \ln r.
\end{equation}
Clearly, r.h.s.(\ref{eq:srDERofVsOFr}) $=0$ iff $r=1$; thus, $r\mapsto \partial_s^{} V_s^{}(r)$ has a
critical point at $r=1$, and only this one. 
Finally, take the second partial derivative of $\partial_s^{} V_s^{}(r)$ w.r.t. $r$ to get
\begin{equation}\label{eq:srrDERofVsOFr}
\partial_{r}^{2} \partial_{s}^{} V_s^{}(r) 
=  
r^{-s-2}\left(1- (1+s) \ln r \right).
\end{equation}
Evaluating r.h.s.(\ref{eq:srrDERofVsOFr}) at $r=1$ yields $\partial_{r}^{2} \partial_{s}^{} V_s^{}(1)=1$; thus, 
 for each $s$ the map $r\mapsto \partial_s^{} V_s^{}(r)$ has a non-degenerate minimum at $r=1$ with value $0$, and no
other critical point. 
Therefore, $\partial_s^{} V_s^{}(r)\geq 0$, with ``$=0$'' iff $r=1$.

With the help of this calculus result we now conclude that whenever $s>t$, then
\begin{equation}\label{eq:vOfsISupINs}
v_s^{}(N) 
=  
\langle V_s\rangle(\omega_N^s)
>
\langle V_t\rangle(\omega_N^s)
\geq
\langle V_t\rangle(\omega_N^t)
=
v_t^{}(N);
\end{equation}
here, the strict inequality holds because there is no optimizing $N$-point configuration on $\Sset^2$ with
all pairs having distance~1.

This completes the proof. 

\newpage
\section{} \label{sec:appdx.D}
\subsection*{Spherical digital nets} \label{sec:sph.digital.nets}

	Maximal sum-of-distance points (i.e. optimal configurations $\omega_N^s$ for $s = -1$)
provide optimal integration nodes for equal-weight numerical integration rules on the sphere; 
see \cite{BrDi2013b} and \cite{BrSaSlWo14,BrWoManuscript}. 
	In general, such configurations are obtained by solving a highly non-linear optimization problem which 
makes them impractical for large number of points.

	Digital nets and sequences introduced in \cite{Ni1987} are efficiently computable so-called 
\emph{low-discrepancy} point systems in the unit square that define effective Quasi-Monte Carlo rules for
integrating functions on the unit square
\begin{equation*}
\frac{1}{N} \sum_{j=1}^N f( \xV_j ) \approx \iint_{\!\![0,1]^2} f( \xV ) \, d \xV.
\end{equation*}
	Informally speaking, the points of a digital net in $[0,1]^2$ are distributed in such a way that a large 
number of elementary rectangles contain precisely the fraction of all points that corresponds to their area; cf. \cite{DiPi2010}
and Fig.~\ref{fig:Sobol.pts}.
\begin{figure}[H]
\centering
\includegraphics[scale=.75]{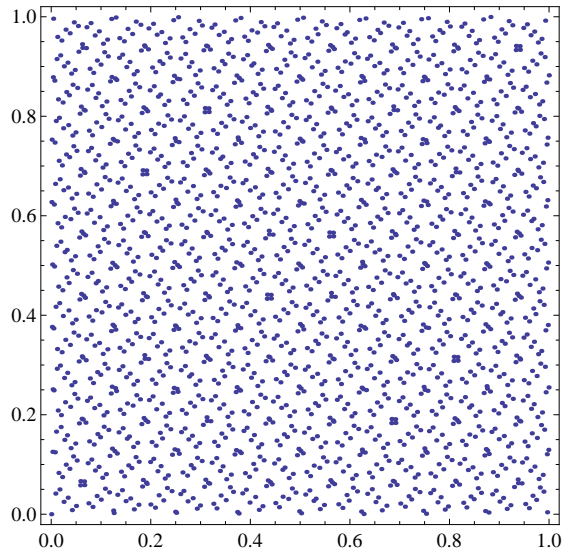} 
\caption{\footnotesize{$2048$ Sobol' points in the unit square generated using a method by \cite{JoKu2003}.}}\label{fig:Sobol.pts}
\end{figure}
	This distribution property 
remains unchanged when lifting a digital net and elementary rectangles (now called spherical rectangles) 
to $\Sset^2$ using the area-preserving Lambert transformation of the map makers.
	These spherical digital nets are studied in \cite{BrDi2012} and \cite{AiBrDi2012}. 
	Of particular interest is that numerically they are comparable with maximal sum-of-distance points. 
	It is conjectured\footnote{Numerical results based on the Sobol' points implemented in Matlab for $N$ up
to $1$ Million points support this conjecture (cf. \cite{BrDi2012}).}
 that $\langle V_{-1} \rangle( \omega_N^{\mathrm{sphDN}} )$ will approach the same limit as $v_{-1}(N)$ 
with the same rate of convergence as $N\to\infty$. 

	We tested for local concavity a sequence of $N$-point spherical digital nets formed by the first $N$ points of a 
Sobol' sequence lifted to the sphere. 
	For the implementation of the Sobol' points we used \cite{JoKu2003}; cf. Fig.~\ref{fig:Sobol.pts}.
	The graph of $N \mapsto \langle V_{-1} \rangle( \omega_N^{\mathrm{sphDN}} )$ (Fig.~\ref{fig:Sobol.vs}) shows an overall 
concave shape. 
\begin{figure}[H]
\centering
\includegraphics[scale=.75]{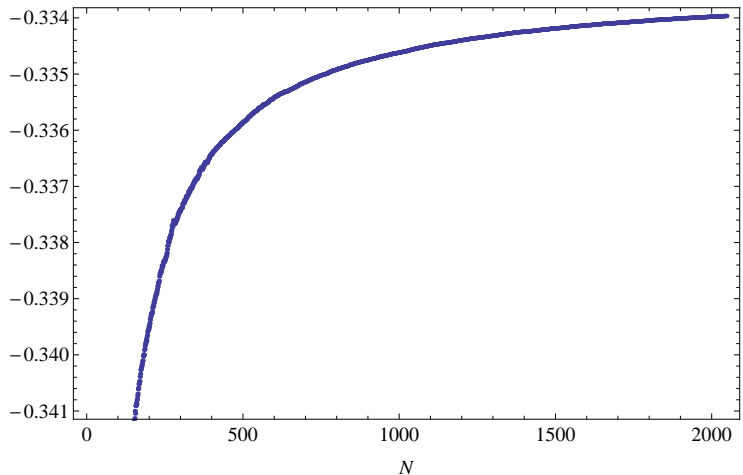} 
\caption{\footnotesize{$\langle V_{-1} \rangle( \omega_N^{\mathrm{sphDN}} )$ based on the $2048$ Sobol' points given in 
			Fig.~\ref{fig:Sobol.pts}.}}\label{fig:Sobol.vs}
\end{figure}
	Yet, some irregularities are clearly discernible in Fig.~\ref{fig:Sobol.vs}; in fact,
the discrete second derivative of $N\mapsto \langle V_{-1} \rangle( \omega_N^{\mathrm{sphDN}})$  reveals
that $N\mapsto\langle V_{-1} \rangle( \omega_N^{\mathrm{sphDN}})$
for this sequence of spherical digital nets is locally highly non-concave; see Fig.~\ref{fig:Sobol.ddot}.
\begin{figure}[H]
\centering
\includegraphics[scale=.75]{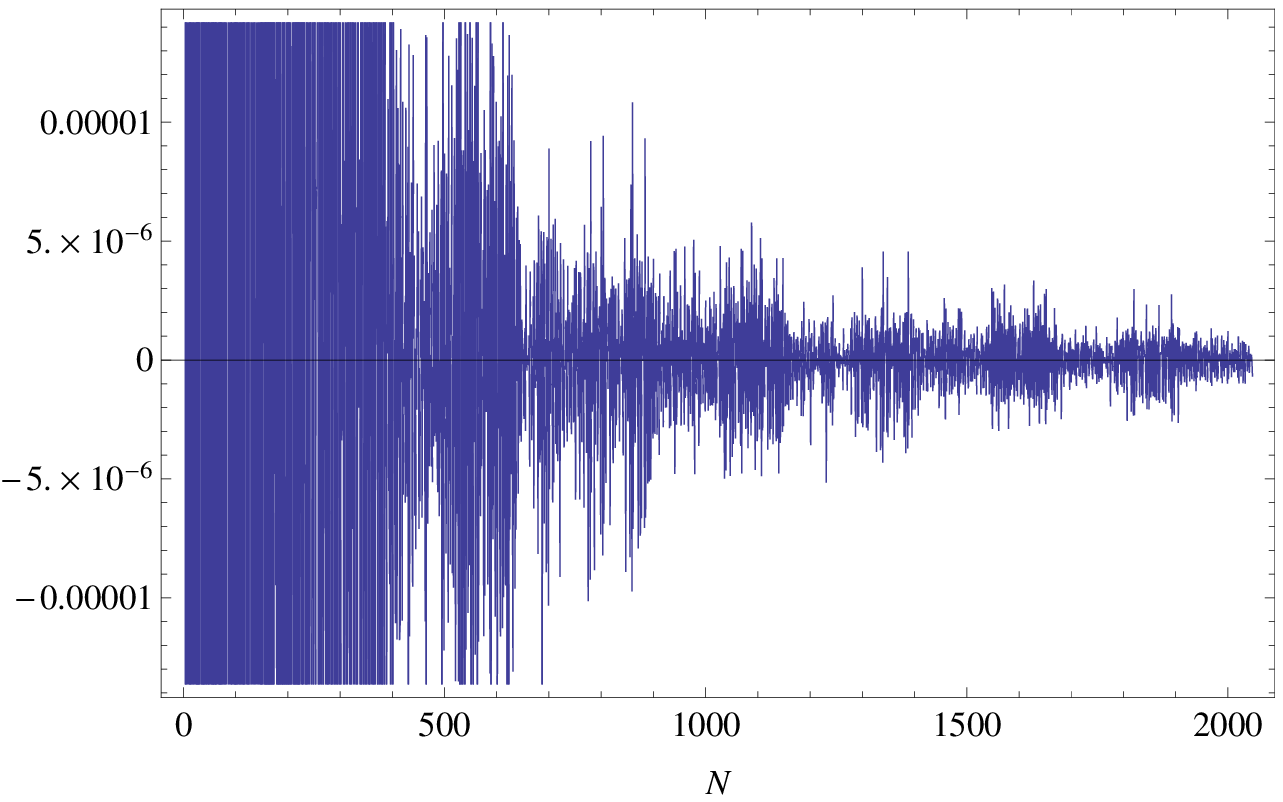} 
\caption{\footnotesize{Discrete second derivative of $\langle V_{-1} \rangle( \omega_N^{\mathrm{sphDN}} )$ based on the $2048$ 
			Sobol' points given in Fig.~\ref{fig:Sobol.pts}}. 
	The discrete points are joint by lines to guide the eye.}\label{fig:Sobol.ddot}
\end{figure}

\end{appendices}

\newpage

\bibliographystyle{modamsplain}


\vfill
\vfill
\hrule
\bigskip

nerattini@fi.infn.it

brauchart@math.tugraz.at

miki@math.rutgers.edu



\newpage

\section*{ ``Magic'' numbers in Smale's 7th problem:\\
           \centerline{Supplement}}

This supplementary section lists the standardized Riesz $s$-energy data used in the main part of our paper.

	\subsection{Tables for $s=-1$}
\def\ifundefined#1{\expandafter\ifx\csname#1\endcsname\relax}


\ifundefined{inputGnumericTable}

	\def\gnumericTableEnd{\end{document}}


\else

   \def\gnumericTableEnd{}


\fi


\providecommand{\gnumericmathit}[1]{#1} 
\providecommand{\gnumericPB}[1]%
{\let\gnumericTemp=\\#1\let\\=\gnumericTemp\hspace{0pt}}
 \ifundefined{gnumericTableWidthDefined}
        \newlength{\gnumericTableWidth}
        \newlength{\gnumericTableWidthComplete}
        \newlength{\gnumericMultiRowLength}
        \global\def\gnumericTableWidthDefined{}
 \fi
 \ifthenelse{\isundefined{\languageshorthands}}{}{\languageshorthands{english}}
\providecommand\gnumbox{\makebox[0pt]}

\setlength{\bigstrutjot}{\jot}
\setlength{\extrarowheight}{\doublerulesep}

\setlongtables

\setlength\gnumericTableWidth{%
	53pt+%
	135pt+%
	135pt+%
0pt}
\def\gumericNumCols{3}
\setlength\gnumericTableWidthComplete{\gnumericTableWidth+%
         \tabcolsep*\gumericNumCols*2+\arrayrulewidth*\gumericNumCols}
\ifthenelse{\lengthtest{\gnumericTableWidthComplete > \linewidth}}%
         {\def\gnumericScale{\ratio{\linewidth-%
                        \tabcolsep*\gumericNumCols*2-%
                        \arrayrulewidth*\gumericNumCols}%
{\gnumericTableWidth}}}%
{\def\gnumericScale{1}}


\ifthenelse{\isundefined{\gnumericColA}}{\newlength{\gnumericColA}}{}\settowidth{\gnumericColA}{
\begin{tabular}{@{}p{53pt*\gnumericScale}@{}}x\end{tabular}}
\ifthenelse{\isundefined{\gnumericColB}}{\newlength{\gnumericColB}}{}\settowidth{\gnumericColB}{\begin{tabular}{@{}p{135pt*\gnumericScale}@{}}x\end{tabular}}
\ifthenelse{\isundefined{\gnumericColC}}{\newlength{\gnumericColC}}{}\settowidth{\gnumericColC}{\begin{tabular}{@{}p{135pt*\gnumericScale}@{}}x\end{tabular}}

\begin{longtable}[c]{%
	|b{\gnumericColA}|%
	b{\gnumericColB}|%
	b{\gnumericColC}|%
	}

\caption*{\footnotesize{Putatively minimal average standardized Riesz pair-energy $v^{x}_{-1}(N)$ and its second derivative $\ddot{v}^{x}_{-1}$ converted from computer-experimental data ${\cal{E}}^{x}_{-1}(N)$ in [RSZ95]}}\\
\hline	

	\multicolumn{1}{|c|}{$N$}	
	&\multicolumn{1}{c|}{$v^{x}_{-1}(N)$}	
	&\multicolumn{1}{c|}{$\ddot{v}^{x}_{-1}(N)$}	\\ 
	\endfirsthead

	\hline
	\multicolumn{3}{r}{\small\slshape continued on next page} \\
	\endfoot
	\multicolumn{3}{l}{\small\slshape continued from previous page} \\ \hline
	\endhead

	\hline 
	\endlastfoot
%
\hline
	 \gnumericPB{\centering}\gnumbox{2}
	&\gnumericPB{\centering}\gnumbox{-1.0000000000}
	&
\\
	 \gnumericPB{\centering}\gnumbox{3}
	&\gnumericPB{\centering}\gnumbox{-0.7320508077}
	&\gnumericPB{\centering}\gnumbox{-0.1688915465}
\\
	 \gnumericPB{\centering}\gnumbox{4}
	&\gnumericPB{\centering}\gnumbox{-0.6329931618}
	&\gnumericPB{\centering}\gnumbox{-0.0342078637}
\\
	 \gnumericPB{\centering}\gnumbox{5}
	&\gnumericPB{\centering}\gnumbox{-0.5681433797}
	&\gnumericPB{\centering}\gnumbox{-0.0280772524}
\\
	 \gnumericPB{\centering}\gnumbox{6}
	&\gnumericPB{\centering}\gnumbox{-0.5313708499}
	&\gnumericPB{\centering}\gnumbox{-0.0068743336}
\\
	 \gnumericPB{\centering}\gnumbox{7}
	&\gnumericPB{\centering}\gnumbox{-0.5014726538}
	&\gnumericPB{\centering}\gnumbox{-0.0096073663}
\\
	 \gnumericPB{\centering}\gnumbox{8}
	&\gnumericPB{\centering}\gnumbox{-0.4811818239}
	&\gnumericPB{\centering}\gnumbox{-0.0042099548}
\\
	 \gnumericPB{\centering}\gnumbox{9}
	&\gnumericPB{\centering}\gnumbox{-0.4651009489}
	&\gnumericPB{\centering}\gnumbox{-0.0031963792}
\\
\end{longtable}


\newpage

	\subsection{Tables for $s=0$}

\begin{longtable}[c]{%
	|b{\gnumericColA}|%
	b{\gnumericColB}|%
	b{\gnumericColC}|%
	}

\caption*{\footnotesize{Putatively minimal average standardized Riesz pair-energy $v^{x}_{0}(N)$ and its second derivative $\ddot{v}^{x}_{0}$ converted from computer-experimental data ${\cal{E}}^{x}_{0}(N)$ in [RSZ95,Ca09]}}\\
\hline	

	\multicolumn{1}{|c|}{$N$}	
	&\multicolumn{1}{c|}{$v^{x}_{0}(N)$}	
	&\multicolumn{1}{c|}{$\ddot{v}^{x}_{0}(N)$}	\\ 
	\endfirsthead

	\hline
	\multicolumn{3}{r}{\small\slshape continued on next page} \\
	\endfoot
	\multicolumn{3}{l}{\small\slshape continued from previous page} \\ \hline
	\endhead

	\endlastfoot
\hline
	 \gnumericPB{\centering}\gnumbox{2}
	&\gnumericPB{\centering}\gnumbox{-0.6931471806}
	&
\\
	 \gnumericPB{\centering}\gnumbox{3}
	&\gnumericPB{\centering}\gnumbox{-0.54930614433405}
	&\gnumericPB{\centering}\gnumbox{-0.08494951843776}
\\
	 \gnumericPB{\centering}\gnumbox{4}
	&\gnumericPB{\centering}\gnumbox{-0.49041462650586}
	&\gnumericPB{\centering}\gnumbox{-0.01052760684652}
\\
	 \gnumericPB{\centering}\gnumbox{5}
	&\gnumericPB{\centering}\gnumbox{-0.44205071552419}
	&\gnumericPB{\centering}\gnumbox{-0.02220150379344}
\\
	 \gnumericPB{\centering}\gnumbox{6}
	&\gnumericPB{\centering}\gnumbox{-0.41588830833597}
	&\gnumericPB{\centering}\gnumbox{8.40980789335588E-05}
\\
	 \gnumericPB{\centering}\gnumbox{7}
	&\gnumericPB{\centering}\gnumbox{-0.38964180306881}
	&\gnumericPB{\centering}\gnumbox{-0.00903390867908}
\\
	 \gnumericPB{\centering}\gnumbox{8}
	&\gnumericPB{\centering}\gnumbox{-0.37242920648072}
	&\gnumericPB{\centering}\gnumbox{-0.00277652137846}
\\
	 \gnumericPB{\centering}\gnumbox{9}
	&\gnumericPB{\centering}\gnumbox{-0.35799313127109}
	&\gnumericPB{\centering}\gnumbox{-0.00229013036125}
\\
\hline
\end{longtable}


\newpage
	\subsection{Tables for $s=1$}

\begin{longtable}[c]{%
	|b{\gnumericColA}|%
	b{\gnumericColB}|%
	b{\gnumericColC}|%
	}

\caption*{\footnotesize{Putatively minimal average standardized Riesz pair-energy $v^{x}_{1}(N)$ and its second derivative $\ddot{v}^{x}_{1}$ converted from computer-experimental data ${\cal{E}}^{x}_{1}(N)$ in [RSZ95,Ca09]}}\\
\hline	

	\multicolumn{1}{|c|}{$N$}	
	&\multicolumn{1}{c|}{$v^{x}_{1}(N)$}	
	&\multicolumn{1}{c|}{$\ddot{v}^{x}_{1}(N)$}	\\ 
	\endfirsthead

	\hline
	\multicolumn{3}{r}{\small\slshape continued on next page} \\
	\endfoot
	\multicolumn{3}{l}{\small\slshape continued from previous page} \\ \hline
	\endhead

	\hline 
	\endlastfoot
%
\hline
	  \gnumericPB{\centering}\gnumbox{2}
	&\gnumericPB{\centering}\gnumbox{-0.5}
	&
\\
	 \gnumericPB{\centering}\gnumbox{3}
	&\gnumericPB{\centering}\gnumbox{-0.42264973066667}
	&\gnumericPB{\centering}\gnumbox{-0.04232810533333}
\\
	 \gnumericPB{\centering}\gnumbox{4}
	&\gnumericPB{\centering}\gnumbox{-0.38762756666667}
	&\gnumericPB{\centering}\gnumbox{7.45526666665608E-05}
\\
	 \gnumericPB{\centering}\gnumbox{5}
	&\gnumericPB{\centering}\gnumbox{-0.35253085}
	&\gnumericPB{\centering}\gnumbox{-0.01688044}
\\
	 \gnumericPB{\centering}\gnumbox{6}
	&\gnumericPB{\centering}\gnumbox{-0.33431457333333}
	&\gnumericPB{\centering}\gnumbox{0.00433531571429}
\\
	 \gnumericPB{\centering}\gnumbox{7}
	&\gnumericPB{\centering}\gnumbox{-0.31176298095238}
	&\gnumericPB{\centering}\gnumbox{-0.00809975785714}
\\
	 \gnumericPB{\centering}\gnumbox{8}
	&\gnumericPB{\centering}\gnumbox{-0.29731114642857}
	&\gnumericPB{\centering}\gnumbox{-0.00158550753968}
\\
	 \gnumericPB{\centering}\gnumbox{9}
	&\gnumericPB{\centering}\gnumbox{-0.28444481944444}
	&\gnumericPB{\centering}\gnumbox{-0.00137818531746}
\\
\end{longtable}


\newpage
	\subsection{Tables for $s=2$}

\begin{longtable}[c]{%
	|b{\gnumericColA}|%
	b{\gnumericColB}|%
	b{\gnumericColC}|%
	}

\caption*{\footnotesize{Putatively minimal average standardized Riesz pair-energy $v^{x}_{2}(N)$ and its second derivative $\ddot{v}^{x}_{2}$ converted from computer-experimental data ${\cal{E}}^{x}_{2}(N)$ in [Ca09,Cec]}}\\
\hline	

	\multicolumn{1}{|c|}{$N$}	
	&\multicolumn{1}{c|}{$v^{x}_{2}(N)$}	
	&\multicolumn{1}{c|}{$\ddot{v}^{x}_{2}(N)$}	\\ 
	\endfirsthead

	\hline
	\multicolumn{3}{r}{\small\slshape continued on next page} \\
	\endfoot
	\multicolumn{3}{l}{\small\slshape continued from previous page} \\ \hline
	\endhead

	\hline 
	\endlastfoot
%
\hline
	\gnumericPB{\centering}\gnumbox{2}
	&\gnumericPB{\centering}\gnumbox{-0.375}
	&
\\
	 \gnumericPB{\centering}\gnumbox{3}
	&\gnumericPB{\centering}\gnumbox{-0.33333333333333}
	&\gnumericPB{\centering}\gnumbox{-0.02083333333333}
\\
	 \gnumericPB{\centering}\gnumbox{4}
	&\gnumericPB{\centering}\gnumbox{-0.3125}
	&\gnumericPB{\centering}\gnumbox{0.00416666666667}
\\
	 \gnumericPB{\centering}\gnumbox{5}
	&\gnumericPB{\centering}\gnumbox{-0.2875}
	&\gnumericPB{\centering}\gnumbox{-0.0125}
\\
	 \gnumericPB{\centering}\gnumbox{6}
	&\gnumericPB{\centering}\gnumbox{-0.275}
	&\gnumericPB{\centering}\gnumbox{0.00654761904762}
\\
	 \gnumericPB{\centering}\gnumbox{7}
	&\gnumericPB{\centering}\gnumbox{-0.25595238095238}
	&\gnumericPB{\centering}\gnumbox{-0.00708111160256}
\\
	 \gnumericPB{\centering}\gnumbox{8}
	&\gnumericPB{\centering}\gnumbox{-0.24398587350732}
	&\gnumericPB{\centering}\gnumbox{-0.00057967860454}
\\
	 \gnumericPB{\centering}\gnumbox{9}
	&\gnumericPB{\centering}\gnumbox{-0.23259904466681}
	&\gnumericPB{\centering}\gnumbox{-0.00055045392804}
\\
\end{longtable}


\newpage
	\subsection{Tables for $s=3$}

\begin{longtable}[c]{%
	|b{\gnumericColA}|%
	b{\gnumericColB}|%
	b{\gnumericColC}|%
	}
\caption*{\footnotesize{Putatively minimal average standardized Riesz pair-energy $v^{x}_{3}(N)$ and its second derivative $\ddot{v}^{x}_{2}$ converted from computer-experimental data ${\cal{E}}^{x}_{3}(N)$ in [Ca09,Cec]}}\\
\hline	

	\multicolumn{1}{|c|}{$N$}	
	&\multicolumn{1}{c|}{$v^{x}_{3}(N)$}	
	&\multicolumn{1}{c|}{$\ddot{v}^{x}_{3}(N)$}	\\ 
	\endfirsthead

	\hline
	\multicolumn{3}{r}{\small\slshape continued on next page} \\
	\endfoot
	\multicolumn{3}{l}{\small\slshape continued from previous page} \\ \hline
	\endhead

	\hline 
	\endlastfoot
%
\hline
	\gnumericPB{\centering}\gnumbox{2}
	&\gnumericPB{\centering}\gnumbox{-0.29166666666667}
	&
\\
	 \gnumericPB{\centering}\gnumbox{3}
	&\gnumericPB{\centering}\gnumbox{-0.2691833}
	&\gnumericPB{\centering}\gnumbox{-0.01008684444444}
\\
	 \gnumericPB{\centering}\gnumbox{4}
	&\gnumericPB{\centering}\gnumbox{-0.25678677777778}
	&\gnumericPB{\centering}\gnumbox{0.00517927555556}
\\
	 \gnumericPB{\centering}\gnumbox{5}
	&\gnumericPB{\centering}\gnumbox{-0.23921098}
	&\gnumericPB{\centering}\gnumbox{-0.00908391333333}
\\
	 \gnumericPB{\centering}\gnumbox{6}
	&\gnumericPB{\centering}\gnumbox{-0.23071909555556}
	&\gnumericPB{\centering}\gnumbox{0.00737168095238}
\\
	 \gnumericPB{\centering}\gnumbox{7}
	&\gnumericPB{\centering}\gnumbox{-0.21485553015873}
	&\gnumericPB{\centering}\gnumbox{-0.00608242214286}
\\
	 \gnumericPB{\centering}\gnumbox{8}
	&\gnumericPB{\centering}\gnumbox{-0.20507438690476}
	&\gnumericPB{\centering}\gnumbox{0.00023825013228}
\\
	 \gnumericPB{\centering}\gnumbox{9}
	&\gnumericPB{\centering}\gnumbox{-0.19505499351852}
	&\gnumericPB{\centering}\gnumbox{0.00015172161376}
\\
\end{longtable}

\ifthenelse{\isundefined{\languageshorthands}}{}{\languageshorthands{\languagename}}
\gnumericTableEnd

\end{document}